\documentclass[11pt]{article}

\usepackage{amsmath,amsfonts,graphpap,amscd,mathrsfs,graphicx,lscape}
\usepackage{amssymb,amstext,xspace}
\usepackage{amsthm}
\usepackage{float}
\usepackage[font={footnotesize}]{caption}
\usepackage{subcaption}
\usepackage{fullpage}

\usepackage[dvipsnames,usenames,table]{xcolor}
\usepackage{tikz}
\tikzset{every picture/.append style={scale=.6}}
\tikzset{every node/.append style={scale=.6}}

\usepackage{color}              
\usepackage[title,header]{appendix}

\usepackage[suppress]{color-edits}
\addauthor{nd}{brown}
\addauthor{bl}{red}
\addauthor{vs}{blue}

\usepackage{url}
\usepackage[colorlinks=true,urlcolor=blue,linkcolor=RoyalBlue,citecolor=OliveGreen]{hyperref}
\usepackage{prettyref, cleveref}
\usepackage[sort&compress,numbers]{natbib}

\usepackage{times}

\newtheorem{theorem}{Theorem}
\newtheorem{lemma}[theorem]{Lemma}

\newtheorem{remark}{Remark}[theorem]

\newtheorem{corollary}[theorem]{Corollary}

\def\squareforqed{\hbox{\rule{2.5mm}{2.5mm}}}

\def\QED{\ifmmode\squareforqed 
  \else{\nobreak\hfil   
    \penalty50                 
    \hskip1em                  
    \null                      
    \nobreak                   
    \hfil                      
    \squareforqed              
    \parfillskip=0pt           
    \finalhyphendemerits=0     
    \endgraf}                  
  \fi}

\def\blksquare{\rule{2mm}{2mm}}
\def\qedsymbol{\blksquare}
\newcommand{\bg}[1]{\medskip\noindent{\bf #1}}
\newcommand{\ed}{{\hfill\qedsymbol}\medskip}

\newenvironment{proofof}[1]{{\it{Proof of #1.}}}{\ed}

\newcommand{\R}{\ensuremath{\mathbb R}}

\newcommand{\argmin}{\ensuremath{\mathrm{argmin}}}

\newcommand{\comment}[1]{}
 {}

\newcommand{\junk}[1]{}

\newlength{\tmp} \newlength{\lpsx} \newlength{\lpsy} \newlength{\upsx} \newlength{\upsy}

\newcommand{\Omit}[1]{}

\newcommand{\allocv}{x^V}
\newcommand{\allocw}{x^W}
\newcommand{\pricev}{p^V}
\newcommand{\pricew}{p^W}



\begin{document}

\title{Truthful Multi-Parameter Auctions with Online Supply:\\ an Impossible Combination
}
\author{
Nikhil R. Devanur\thanks{Microsoft Research. \tt{nikdev@microsoft.com}} \and
Balasubramanian Sivan\thanks{Google research. \tt{balusivan@google.com}} \and
Vasilis Syrgkanis\thanks{Microsoft Research. \tt{vasy@microsoft.com}}
}
\maketitle{}

\begin{abstract}
We study a basic auction design problem with online supply. There are two
unit-demand bidders and two types of items. The first item type will arrive
first for sure, and the second item type may or may not arrive.  The auctioneer
has to decide the allocation of an item immediately after each item arrives, but
is allowed to compute payments after knowing how many items arrived. For this
problem we show that there is no deterministic truthful and individually
rational mechanism that, even with unbounded computational resources, gets any
\emph{finite} approximation factor to the optimal social welfare.  


\end{abstract}
%

\newpage
\setcounter{page}{1}

\section{Introduction}
We consider the following fundamental mechanism design problem with \emph{ online supply}. 
\begin{quote} 
	There are two unit-demand bidders\footnote{A unit demand bidder has a valuation of the following form: for a given set of items $S$, the valuation $v(S) = \max_{j \in S} v_{j} .$ She tries to maximize her utility, which is quasi linear, i.e., it is the valuation of the set of items allocated to her minus her payments. } and two types of items. To begin with, both bidders submit their bids for both items. 
	The first item type arrives first, for sure, and the mechanism has to allocate the item to one of the bidders right away. 
	Then the second item type may or may not arrive and is only allocated if it does. 
	The mechanism can compute payments at the end, after knowing whether the second item arrives. 
	The objective is to optimize social welfare, while ensuring that the mechanism is truthful in dominant strategies. 
\end{quote}

\paragraph{Relevance of the model} There has been a tremendous amount of work motivated by online advertising, but much of it studies either the online or truthfulness aspects separately. E.g., \citet{Mehta2007adwords} study the purely algorithmic problem, and \citet{Varian2007position,Edelman2005} consider just a single auction. 
 The goal here is to study these two aspects together, and the particular choices in our model reflect how online advertising works: 
 advertisers specify bids for possible item types ahead of time.  Items, which are ad opportunities, arrive online and must be allocated immediately upon arrival, as they perish otherwise. Payments can be computed at the end: advertisers are usually billed at regular intervals. The model can be easily extended to include capacity constraints larger than one and/or budget constraints. Optimizing social welfare is an important objective even for a for-profit seller in a competitive environment, as  keeping consumers happy in the long-run is crucial.

We focus on mechanisms that are deterministic, truthful in dominant strategies and individually rational (IR). 
Practical considerations require deterministic and IR mechanisms, while truthful in dominant strategies is the natural 
notion when no Bayesian assumptions are made about bidder types. 
Moreover, given the simplicity of the setting, it is natural to expect such a mechanism.

The  problem we consider (with two bidders and two items) is the simplest possible version that retains 
three crucial aspects: online supply, truthfulness, and multidimensional type spaces. 
Removing any of these aspects immediately yields a simple and easy solution, 
even with many bidders and many items. 
\begin{itemize}
	\item Offline version. If we know which items are available, we can simply run the VCG mechanism to get the optimal social welfare. 
	\item No truthfulness. If we only care about the algorithmic problem of maximizing welfare online, then this is a well known generalization of the online bipartite matching problem. The greedy algorithm is $\tfrac 1 2 $ competitive. 
	\item Single parameter types. There is only a single item type, and we may get an unknown number of copies of this item (\emph{identical items case}).  The type of an agent is just a single value. In this case, allocating the items to agents in the decreasing order of values optimizes social welfare. This is in fact the VCG allocation and is therefore truthful (but only because we can compute payments in the end). 
\end{itemize}
\vsedit{Existing work that combines both the online and then incentive nature of the problem (e.g. \citet{BabaioffBR15}) consider only non-trivial variants of the \emph{identical items case} and no prior work has addressed a multi-dimensional mechanism design setting (see Section \ref{sec:related} for detailed exposition).}

\subsection{Our Contribution}

\paragraph{Some examples} 
Given that there are only two bidders and two item types, 
one might expect some simple mechanism to give something like a $2$ approximation, 
but in reality it is otherwise.
Here are a few examples that we encourage the reader to work out and see what's good or bad about them before proceeding further. Let $v = (v_1, v_2)$ be the bids of bidder $V$ for items $1$ and $2$, and similarly let $w = (w_1, w_2)$ be the bids of bidder $W$. 

We start with some natural allocation algorithms. 
The question for these algorithms is whether there exist payments that make them truthful. 
\paragraph{Example 1: Greedy algorithm} It is well known  \cite{Feldman2009onlineFreeDisposal} that the following online greedy algorithm gets a competitive ratio of 2 for unit demand bidders, for the purely optimization problem without incentives, even with $n$ bidders and $m$ items. The algorithm assigns each item to the bidder with the highest marginal valuation, i.e., if bidder $i$ is currently assigned a set of items, the maximum item of which he values at $v_i'$, and his value for the new item is $v_i$, then his marginal valuation for item $i$ is $v_i - v_i'$.  

\paragraph{Example 2: A $\min\{m,n\}$ approximation algorithm} Assign each item to the highest bidder for that item. As an algorithm, this gets a factor of $\min\{m,n\}$, where $m$ is the number of items and $n$ is the number of agents. 


\paragraph{Example 3: Bundling} 
Both items are allocated to the bidder with the highest bid among $v_1,v_2,w_1,w_2$. 


\paragraph{Example 4: VCG-ish allocation}
Assign the first item to the highest bidder. Assign the second item according to the allocation that maximizes welfare, 
on hindsight, if could reallocate both items (i.e., the VCG allocation if we knew both items would arrive). 


With some effort, one can see that none of these algorithms have any payments that render them truthful.  
The one for which this is easiest to see is perhaps the Bundling example. 
Suppose that $v_2$ is the highest, with $w_1 >  v_1 \approx 0$, and only item 1 arrives.
We will allocate the item to $V$, but what should we charge her? 
IR forces the payment to be essentially 0. 
Then bidding very low on item 1 and very high on item 2 is a beneficial misreport when only item 1 arrives. 

One could ask if there are any reasonable allocation algorithms that are truthful. 
Indeed there are, and here's an example. 
\paragraph{Example 5: Discount mechanism} When item 1 arrives, assign it to the highest bidder. 
When item $2$ arrives, offer it to bidder $V$ at a price of $\max\{w_1,w_2\} + (v_1 - w_1)^+$, and offer it to bidder $W$ at a price of $\max\{v_1, v_2\} + (w_1 - v_1)^+$ (here $x^+ = \max(x,0)$). Whoever takes the second item, pays the posted price on the second item, and doesn't pay anything more even if they won the first item. If there is a bidder who won only the first item (either because second item didn't arrive, or he lost the second item), he pays the other bidder's bid for the first item. 


%
Unfortunately, Example $5$ does not get any finite approximation. 
In light of the above examples, we ask the following question:  


\begin{quote}
{\bf Main Question.} Does there exist a deterministic truthful+IR mechanism that, even with unbounded computational resources, gets a finite approximation factor to social welfare, when items arrive online?  
\end{quote}


Our main result is an {\em impossibility}: there is no deterministic truthful + IR mechanism that, even with unbounded computational resources, can get {\em any finite} approximation to social welfare. This is surprising, and is qualitatively different from other impossibilities, because optimal approximation ratios usually grow as a function of the number of bidders $n$, or the number of items $m$. This result shows a drastic contrast w.r.t the problem with a single item type (with identical copies), where we can implement the optimal allocation (for social welfare) when payments are computed at the end. 
Even if the payments are required to be computed on the fly for the single item case, as in~\citet{BabaioffBR15}, the approximation ratio is $\Theta(n)$, which is in stark contrast to the answer here,  which is $\infty$. 
We emphasize once again that this is the simplest possible problem for which one could have hoped for a positive answer, 
and it is quite surprising that the answer is an impossibility. 

\paragraph{Technical takeaways}
Traditional characterizations of truthfulness such as weak/cyclic monotonicity only consider the allocation to a single agent, fixing the reports of everyone else, and fixing the supply. Our first step is to extend this characterization when the allocation is required to satisfy sequential consistency: the allocation must satisfy weak monotonicity at every point of time, i.e., after the arrival of each item, the allocation for the set of items seen so far must satisfy weak monotonicity for every agent. 
Even this is not sufficient, as there are allocations such as the one in Example 5 that satisfy this sequential consistency. 
We further need to reason about how the allocation of one agent changes when we change the report of the other agent, 
under the constraint that it achieves a finite approximation. We give more details and intuition for how we do this in \Cref{sec:intuition}, and in the numerous figures throughout the paper.\footnote{We have generously added figures since they greatly help in understanding the structure, and as a result they significantly lengthen the paper. We estimate that the figures add about 4 pages.}

\paragraph{Randomization}
Our result implies that some sort of randomization is required to get any finite approximation to social welfare. 
Practically, randomization in the mechanism itself is not appealing. 
More practical options are to assume some randomization in the arrival process, as has been done for the algorithmic versions, 
or a Bayesian setting where there is randomization in the types. See \Cref{sec:conclusion} for a more detailed discussion on these options.

\subsection{Related Work}\label{sec:related}

\paragraph{Closely related work} The results closest to our work are the ones that consider non trivial variants of the \emph{identical items case}, since they still have to deal with the combination of truthfulness and online supply.

\begin{enumerate}
\item \citet{BabaioffBR15} study the variant where the mechanism has to 
\emph{compute payments immediately} after each item arrives. (As observed 
above, the problem is trivial when payments can be computed at the end.)
Even with this more stringent requirement, a trivial deterministic auction obtains an $n$ approximation in this setting, and they construct a simple universally truthful ``powers-of-2'' randomized mechanism that obtains an $O(\log n)$ approximation. On the negative side, they show that no randomized mechanism can achieve a factor better than $O(\log \log n)$ approximation. 
\item \citet{Goel2013clinching} study a variant where each bidder has a\emph{ budget constraint} which specifies a hard upper bound on the bidder's total payment across all items, and additive valuations. 
The goal is now to get a {\em Pareto optimal} allocation, rather than approximating the social welfare. 
They show that the adaptive clinching auction of \citet{Dobzinski2008} actually achieves this,  in the adversarial arrival model. This auction is truthful, and payments can be computed online too. 

\item \citet{Mahdian2006multi} and \citet{Devanur2009limited} consider the \emph{revenue} maximization objective, and 
show $O(1)$ approximation ratios. 
The payments here are allowed to be computed at the end.
\end{enumerate}

\paragraph{Other related Work} 
We now briefly discuss the broader landscape of online mechanisms and algorithms inspired by internet advertising.

There is a substantial body of work under \emph{dynamic mechanism design} that studies truthful mechanism design with online supply, but with the key difference being that the bidders' types themselves are evolving online 
\citep{athey2013efficient,bergemann2010dynamic,pavan2014dynamic,kakade2013optimal,papadimitriou2016complexity,ashlagi2016sequential}. 
The setting is typically Bayesian, and the emphasis is on truthful elicitation of these evolving types. This main source of difficulty is absent in our problem, since the bidder types are fixed and elicited once in the beginning; thus one would expect our problem to be easier. 
The other variant of online mechanisms is one where bidders arrive and depart online, while the supply is fixed and known
\citep{Lavi2000competitive,said2012auctions,gershkov2009dynamic,lavi2005online,Babaioff2007matroids,Hajiaghayi2004adaptive,Kleinberg2005multiple,Cole2008prompt}. 
Once again, the difficulties in these problems are orthogonal to that of ours; see \Cref{sec:intuition} for more discussion on this. 
For a survey on dynamic mechanism design, see \citet{bergemann2011dynamic,Parkes2007online,vohra2012dynamic}. 

Online advertising has been the dominant driver of the internet economy, accounting for almost all the revenue of two of the most important technology companies today, Google and Facebook. It has been the motivation for numerous lines of research, most prominently online matching and its generalizations \cite{Mehta2007adwords,Buchbinder2007online, Devanur2009adwords, Devanur2011near, Feldman2009onlineBeating, Feldman2009onlineFreeDisposal, Mahdian2007allocating}, and various aspects of auction design \cite{Varian2007position,Edelman2007,Aggarwal2006truthful,PaesLeme2010}. These two have largely remained separate: online algorithms ignore strategic considerations and auction design mostly considers a static setting.  It is natural to combine the two aspects, and ask for mechanisms that are both online and truthful, and many papers have done so, but capturing different aspects than us 
\cite{Balseiro2015repeated,Babaioff2010truthful,Babaioff2009characterizing,Devanur2009price}. 

\section{Model and Main Result}
We consider two bidders $V,W$ and two possible items, $1,2$. Bidder $V$ has type $v=(v_1,v_2)$, where $v_i$ is his value for item $i$. Similarly bidder $W$ has a type $w=(w_1,w_2)$. Both bidders are unit-demand. 
Item $1$ arrives in the first period. Item $2$ may or may not arrive in the second period.

 A mechanism is defined by its allocation and payment functions. We focus on deterministic mechanisms: the functions $\allocv_i:  \R_{+}^2 \times \R_{+}^2 \rightarrow \{0,1\}$  and $\allocw_i:  \R^2 \times \R^2 \rightarrow \{0,1\}$ denote the allocation of item $i$ to bidders $V$ and $W$ respectively, as a function of their ($4$-tuple) bids in $\R_{+}^2\times\R_{+}^2$. Similarly $\pricev(v,w)$ and $\pricew(v,w)$  denote the payments of bidders $V$ and $W$ as a function of their bids. We use $\R_{+}$ to denote the set of positive real numbers union $0$. 
 We require the mechanism to irrevocably decide the allocation of item 1 immediately after it arrives. However, the payment for the agents can be computed after two periods (i.e., after knowing how many items arrived).
 
 A  mechanism is dominant strategy incentive compatible (DSIC) if irrespective of whether $1$ or $2$ items arrive, irrespective of bidder $W$'s reports, it is a dominant strategy for bidder $V$ to report his true values for both the items at the beginning of the game, and vice versa. 
A mechanism is IR if similarly under no circumstances does the payment of an agent exceed her (reported) utility obtained by her allocation. 
Since we focus on DSIC mechanisms, we avoid different notation for bids and true values. 

\vsedit{We are interested in analyzing the welfare of the resulting outcome of the mechanism, which corresponds to the utility of the bidders plus the revenue of the auctioneer. This boils down to the total value of the resulting allocation. We are interested in comparing the latter with the welfare maximizing allocation, which in this setting corresponds to the maximum weighted matching between the bidders and the items. We say that a mechanism achieves an approximation factor $H$ if for any possible reports $v,w$, the allocation of the mechanism achieves welfare that is at least $1/H$ times the value of the welfare maximizing allocation for these reports.}

\vsedit{With these definitions at hand, we are now ready to state our main result:
\begin{theorem}[Main Theorem]
\label{thm:Main}
No deterministic DSIC+IR mechanism gets a finite approximation factor.
\end{theorem}
Before delving into the technical exposition of the proof of the main theorem and in order to help the reader get a high level glimpse of our approach and of the main insights behind our proof, we first present in the next subsection a series of insights from our proof. Subsequently in Section \ref{sec:proof} we give the complete proof of Theorem \ref{thm:Main}}. 

\subsection{Intuition, Insights and Proof Map}\label{sec:intuition}

To build some intuition let's start with how an efficient DSIC+IR mechanism would look  had it had perfect foresight of the arrival sequence. If only item $1$ arrives, then the efficient DSIC+IR mechanism is the single item second price auction, which would allocate the item to player $V$ if:$v_1\geq w_1$ and to bidder $W$ otherwise (see Figure \ref{fig:perfect} (left)). However, if both items arrive then the efficient mechanism, i.e. the VCG mechanism, allocates item $1$ to player $V$ and item $2$ to player $W$ if: $v_1+w_2\geq v_2+w_1 \Leftrightarrow v_1-v_2\geq w_1-w_2$. Otherwise it exchanges the allocation of goods (see Figure \ref{fig:perfect} (right)).
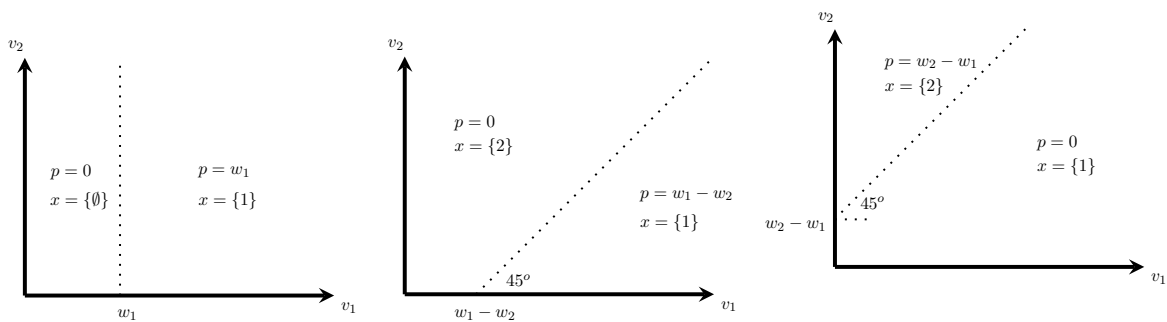
\begin{figure}[h]
\centering
\begin{subfigure}[b]{0.3\textwidth}
\centering
\ifx\du\undefined
  \newlength{\du}
\fi
\setlength{\du}{15\unitlength}
\begin{tikzpicture}
\pgftransformxscale{1.000000}
\pgftransformyscale{-1.000000}
\definecolor{dialinecolor}{rgb}{0.000000, 0.000000, 0.000000}
\pgfsetstrokecolor{dialinecolor}
\definecolor{dialinecolor}{rgb}{1.000000, 1.000000, 1.000000}
\pgfsetfillcolor{dialinecolor}
\pgfsetlinewidth{0.100000\du}
\pgfsetdash{}{0pt}
\pgfsetdash{}{0pt}
\pgfsetbuttcap
{
\definecolor{dialinecolor}{rgb}{0.000000, 0.000000, 0.000000}
\pgfsetfillcolor{dialinecolor}
\pgfsetarrowsend{stealth}
\definecolor{dialinecolor}{rgb}{0.000000, 0.000000, 0.000000}
\pgfsetstrokecolor{dialinecolor}
\draw (5.000000\du,20.000000\du)--(18.000000\du,20.000000\du);
}
\pgfsetlinewidth{0.100000\du}
\pgfsetdash{}{0pt}
\pgfsetdash{}{0pt}
\pgfsetbuttcap
{
\definecolor{dialinecolor}{rgb}{0.000000, 0.000000, 0.000000}
\pgfsetfillcolor{dialinecolor}
\pgfsetarrowsend{stealth}
\definecolor{dialinecolor}{rgb}{0.000000, 0.000000, 0.000000}
\pgfsetstrokecolor{dialinecolor}
\draw (5.000000\du,20.000000\du)--(5.000000\du,10.000000\du);
}
\definecolor{dialinecolor}{rgb}{0.000000, 0.000000, 0.000000}
\pgfsetstrokecolor{dialinecolor}
\node[anchor=west] at (12.000000\du,16.000000\du){$x=\{1\}$};
\definecolor{dialinecolor}{rgb}{0.000000, 0.000000, 0.000000}
\pgfsetstrokecolor{dialinecolor}
\node[anchor=west] at (6.000000\du,13.000000\du){};
\definecolor{dialinecolor}{rgb}{0.000000, 0.000000, 0.000000}
\pgfsetstrokecolor{dialinecolor}
\node[anchor=west] at (5.800000\du,14.800000\du){$p=0$};
\definecolor{dialinecolor}{rgb}{0.000000, 0.000000, 0.000000}
\pgfsetstrokecolor{dialinecolor}
\node[anchor=west] at (18.000000\du,20.500000\du){$v_1$};
\definecolor{dialinecolor}{rgb}{0.000000, 0.000000, 0.000000}
\pgfsetstrokecolor{dialinecolor}
\node[anchor=west] at (4.000000\du,9.500000\du){$v_2$};
\definecolor{dialinecolor}{rgb}{0.000000, 0.000000, 0.000000}
\pgfsetstrokecolor{dialinecolor}
\node[anchor=west] at (8.600000\du,20.800000\du){$w_1$};
\definecolor{dialinecolor}{rgb}{0.000000, 0.000000, 0.000000}
\pgfsetstrokecolor{dialinecolor}
\node[anchor=west] at (12.000000\du,14.800000\du){$p=w_1$};
\pgfsetlinewidth{0.0500000\du}
\pgfsetdash{{\pgflinewidth}{0.200000\du}}{0cm}
\pgfsetdash{{\pgflinewidth}{0.200000\du}}{0cm}
\pgfsetbuttcap
{
\definecolor{dialinecolor}{rgb}{0.000000, 0.000000, 0.000000}
\pgfsetfillcolor{dialinecolor}
\definecolor{dialinecolor}{rgb}{0.000000, 0.000000, 0.000000}
\pgfsetstrokecolor{dialinecolor}
\draw (9.000000\du,20.000000\du)--(9.000000\du,10.000000\du);
}
\definecolor{dialinecolor}{rgb}{0.000000, 0.000000, 0.000000}
\pgfsetstrokecolor{dialinecolor}
\node[anchor=west] at (5.800000\du,16.000000\du){$x=\{\emptyset\}$};
\end{tikzpicture}
\end{subfigure}
\begin{subfigure}[b]{0.3\textwidth}
\centering
\ifx\du\undefined
  \newlength{\du}
\fi
\setlength{\du}{15\unitlength}
\begin{tikzpicture}
\pgftransformxscale{1.000000}
\pgftransformyscale{-1.000000}
\definecolor{dialinecolor}{rgb}{0.000000, 0.000000, 0.000000}
\pgfsetstrokecolor{dialinecolor}
\definecolor{dialinecolor}{rgb}{1.000000, 1.000000, 1.000000}
\pgfsetfillcolor{dialinecolor}
\pgfsetlinewidth{0.100000\du}
\pgfsetdash{}{0pt}
\pgfsetdash{}{0pt}
\pgfsetbuttcap
{
\definecolor{dialinecolor}{rgb}{0.000000, 0.000000, 0.000000}
\pgfsetfillcolor{dialinecolor}
\pgfsetarrowsend{stealth}
\definecolor{dialinecolor}{rgb}{0.000000, 0.000000, 0.000000}
\pgfsetstrokecolor{dialinecolor}
\draw (5.000000\du,20.000000\du)--(18.000000\du,20.000000\du);
}
\pgfsetlinewidth{0.100000\du}
\pgfsetdash{}{0pt}
\pgfsetdash{}{0pt}
\pgfsetbuttcap
{
\definecolor{dialinecolor}{rgb}{0.000000, 0.000000, 0.000000}
\pgfsetfillcolor{dialinecolor}
\pgfsetarrowsend{stealth}
\definecolor{dialinecolor}{rgb}{0.000000, 0.000000, 0.000000}
\pgfsetstrokecolor{dialinecolor}
\draw (5.000000\du,20.000000\du)--(5.000000\du,10.000000\du);
}
\definecolor{dialinecolor}{rgb}{0.000000, 0.000000, 0.000000}
\pgfsetstrokecolor{dialinecolor}
\node[anchor=west] at (6.800000\du,20.800000\du){$w_1-w_2$};
\definecolor{dialinecolor}{rgb}{0.000000, 0.000000, 0.000000}
\pgfsetstrokecolor{dialinecolor}
\node[anchor=west] at (14.600000\du,17.000000\du){$x=\{1\}$};
\definecolor{dialinecolor}{rgb}{0.000000, 0.000000, 0.000000}
\pgfsetstrokecolor{dialinecolor}
\node[anchor=west] at (6.000000\du,13.000000\du){};
\definecolor{dialinecolor}{rgb}{0.000000, 0.000000, 0.000000}
\pgfsetstrokecolor{dialinecolor}
\node[anchor=west] at (6.800000\du,13.800000\du){$x=\{2\}$};
\definecolor{dialinecolor}{rgb}{0.000000, 0.000000, 0.000000}
\pgfsetstrokecolor{dialinecolor}
\node[anchor=west] at (18.000000\du,20.500000\du){$v_1$};
\definecolor{dialinecolor}{rgb}{0.000000, 0.000000, 0.000000}
\pgfsetstrokecolor{dialinecolor}
\node[anchor=west] at (4.000000\du,9.500000\du){$v_2$};
\definecolor{dialinecolor}{rgb}{0.000000, 0.000000, 0.000000}
\pgfsetstrokecolor{dialinecolor}
\node[anchor=west] at (6.800000\du,12.800000\du){$p=0$};
\definecolor{dialinecolor}{rgb}{0.000000, 0.000000, 0.000000}
\pgfsetstrokecolor{dialinecolor}
\node[anchor=west] at (14.600000\du,15.800000\du){$p=w_1-w_2$};
\pgfsetlinewidth{0.0500000\du}
\pgfsetdash{{\pgflinewidth}{0.200000\du}}{0cm}
\pgfsetdash{{\pgflinewidth}{0.200000\du}}{0cm}
\pgfsetbuttcap
{
\definecolor{dialinecolor}{rgb}{0.000000, 0.000000, 0.000000}
\pgfsetfillcolor{dialinecolor}
\definecolor{dialinecolor}{rgb}{0.000000, 0.000000, 0.000000}
\pgfsetstrokecolor{dialinecolor}
\draw (8.000000\du,20.000000\du)--(18.000000\du,10.000000\du);
}
\definecolor{dialinecolor}{rgb}{0.000000, 0.000000, 0.000000}
\pgfsetstrokecolor{dialinecolor}
\node[anchor=west] at (9.000000\du,19.400000\du){$45^o$};
\end{tikzpicture}
\end{subfigure}
\begin{subfigure}[b]{0.3\textwidth}
\centering
\ifx\du\undefined
  \newlength{\du}
\fi
\setlength{\du}{15\unitlength}
\begin{tikzpicture}
\pgftransformxscale{1.000000}
\pgftransformyscale{-1.000000}
\definecolor{dialinecolor}{rgb}{0.000000, 0.000000, 0.000000}
\pgfsetstrokecolor{dialinecolor}
\definecolor{dialinecolor}{rgb}{1.000000, 1.000000, 1.000000}
\pgfsetfillcolor{dialinecolor}
\pgfsetlinewidth{0.100000\du}
\pgfsetdash{}{0pt}
\pgfsetdash{}{0pt}
\pgfsetbuttcap
{
\definecolor{dialinecolor}{rgb}{0.000000, 0.000000, 0.000000}
\pgfsetfillcolor{dialinecolor}
\pgfsetarrowsend{stealth}
\definecolor{dialinecolor}{rgb}{0.000000, 0.000000, 0.000000}
\pgfsetstrokecolor{dialinecolor}
\draw (5.000000\du,20.000000\du)--(18.000000\du,20.000000\du);
}
\pgfsetlinewidth{0.100000\du}
\pgfsetdash{}{0pt}
\pgfsetdash{}{0pt}
\pgfsetbuttcap
{
\definecolor{dialinecolor}{rgb}{0.000000, 0.000000, 0.000000}
\pgfsetfillcolor{dialinecolor}
\pgfsetarrowsend{stealth}
\definecolor{dialinecolor}{rgb}{0.000000, 0.000000, 0.000000}
\pgfsetstrokecolor{dialinecolor}
\draw (5.000000\du,20.000000\du)--(5.000000\du,10.000000\du);
}
\definecolor{dialinecolor}{rgb}{0.000000, 0.000000, 0.000000}
\pgfsetstrokecolor{dialinecolor}
\node[anchor=west] at (13.200000\du,15.800000\du){$x=\{1\}$};
\definecolor{dialinecolor}{rgb}{0.000000, 0.000000, 0.000000}
\pgfsetstrokecolor{dialinecolor}
\node[anchor=west] at (6.000000\du,13.000000\du){};
\definecolor{dialinecolor}{rgb}{0.000000, 0.000000, 0.000000}
\pgfsetstrokecolor{dialinecolor}
\node[anchor=west] at (6.800000\du,12.400000\du){$x=\{2\}$};
\definecolor{dialinecolor}{rgb}{0.000000, 0.000000, 0.000000}
\pgfsetstrokecolor{dialinecolor}
\node[anchor=west] at (18.000000\du,20.500000\du){$v_1$};
\definecolor{dialinecolor}{rgb}{0.000000, 0.000000, 0.000000}
\pgfsetstrokecolor{dialinecolor}
\node[anchor=west] at (4.000000\du,9.500000\du){$v_2$};
\definecolor{dialinecolor}{rgb}{0.000000, 0.000000, 0.000000}
\pgfsetstrokecolor{dialinecolor}
\node[anchor=west] at (6.800000\du,11.400000\du){$p=w_2-w_1$};
\definecolor{dialinecolor}{rgb}{0.000000, 0.000000, 0.000000}
\pgfsetstrokecolor{dialinecolor}
\node[anchor=west] at (13.200000\du,14.800000\du){$p=0$};
\pgfsetlinewidth{0.0500000\du}
\pgfsetdash{{\pgflinewidth}{0.200000\du}}{0cm}
\pgfsetdash{{\pgflinewidth}{0.200000\du}}{0cm}
\pgfsetbuttcap
{
\definecolor{dialinecolor}{rgb}{0.000000, 0.000000, 0.000000}
\pgfsetfillcolor{dialinecolor}
\definecolor{dialinecolor}{rgb}{0.000000, 0.000000, 0.000000}
\pgfsetstrokecolor{dialinecolor}
\draw (4.977232\du,18.001073\du)--(13.000000\du,10.000000\du);
}
\definecolor{dialinecolor}{rgb}{0.000000, 0.000000, 0.000000}
\pgfsetstrokecolor{dialinecolor}
\node[anchor=west] at (5.800000\du,17.300000\du){$45^o$};
\definecolor{dialinecolor}{rgb}{0.000000, 0.000000, 0.000000}
\pgfsetstrokecolor{dialinecolor}
\node[anchor=west] at (1.800000\du,18.200000\du){$w_2-w_1$};
\pgfsetlinewidth{0.0500000\du}
\pgfsetdash{{\pgflinewidth}{0.200000\du}}{0cm}
\pgfsetdash{{\pgflinewidth}{0.200000\du}}{0cm}
\pgfsetbuttcap
{
\definecolor{dialinecolor}{rgb}{0.000000, 0.000000, 0.000000}
\pgfsetfillcolor{dialinecolor}
\definecolor{dialinecolor}{rgb}{0.000000, 0.000000, 0.000000}
\pgfsetstrokecolor{dialinecolor}
\draw (5.000000\du,18.000000\du)--(6.400000\du,18.000000\du);
}
\end{tikzpicture}
\end{subfigure}
\caption{VCG allocation and payment when only one item arrives (left) and when two items arrive (two right).}\label{fig:perfect}
\end{figure} 

Observe that the above two allocation rules are not sequentially consistent, i.e. if the mechanism decides the allocation of item $1$ to player $V$ be simply doing the comparison $v_1\geq w_1$, then there will be cases of $v,w$, where in the second stage, if item $2$ happens to arrive, the efficient mechanism would want to change the allocation of item $1$. In particular, the latter happens if $v_1\geq w_1$ but $v_1-v_2 < w_1-w_2$. Thus we have to change the way that we allocate goods on some of the two days. 

\vsedit{Since VCG is not sequentially consistent, an alternative approach would be to start from a sequentially consistent and approximately efficient allocation algorithm. The most well known and natural one is the \emph{greedy algorithm} which allocates each arriving item greedily to the bidder with the highest current marginal valuation for it. Thus the first item would be allocated to player $V$ if $v_1\geq w_1$. If player $V$ wins item $1$, the second item would also be allocated to him if $v_2-v_1\geq w_1$. If he lost item $1$, the he is allocated item $2$ if $v_2\geq w_2-w_1$. This allocation function is depicted in Figure \ref{fig:greedy}.
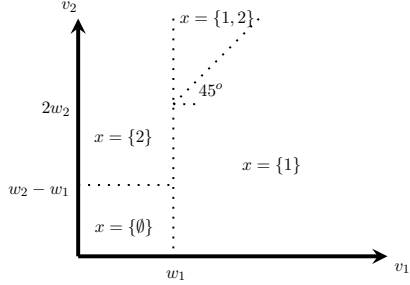
\begin{figure}[h]
\centering
\begin{subfigure}[b]{0.45\textwidth}
\centering
\ifx\du\undefined
  \newlength{\du}
\fi
\setlength{\du}{15\unitlength}
\begin{tikzpicture}
\pgftransformxscale{1.000000}
\pgftransformyscale{-1.000000}
\definecolor{dialinecolor}{rgb}{0.000000, 0.000000, 0.000000}
\pgfsetstrokecolor{dialinecolor}
\definecolor{dialinecolor}{rgb}{1.000000, 1.000000, 1.000000}
\pgfsetfillcolor{dialinecolor}
\pgfsetlinewidth{0.100000\du}
\pgfsetdash{}{0pt}
\pgfsetdash{}{0pt}
\pgfsetbuttcap
{
\definecolor{dialinecolor}{rgb}{0.000000, 0.000000, 0.000000}
\pgfsetfillcolor{dialinecolor}
\pgfsetarrowsend{stealth}
\definecolor{dialinecolor}{rgb}{0.000000, 0.000000, 0.000000}
\pgfsetstrokecolor{dialinecolor}
\draw (5.000000\du,20.000000\du)--(18.000000\du,20.000000\du);
}
\pgfsetlinewidth{0.100000\du}
\pgfsetdash{}{0pt}
\pgfsetdash{}{0pt}
\pgfsetbuttcap
{
\definecolor{dialinecolor}{rgb}{0.000000, 0.000000, 0.000000}
\pgfsetfillcolor{dialinecolor}
\pgfsetarrowsend{stealth}
\definecolor{dialinecolor}{rgb}{0.000000, 0.000000, 0.000000}
\pgfsetstrokecolor{dialinecolor}
\draw (5.000000\du,20.000000\du)--(5.000000\du,10.000000\du);
}
\definecolor{dialinecolor}{rgb}{0.000000, 0.000000, 0.000000}
\pgfsetstrokecolor{dialinecolor}
\node[anchor=west] at (8.400000\du,20.800000\du){$w_1$};
\definecolor{dialinecolor}{rgb}{0.000000, 0.000000, 0.000000}
\pgfsetstrokecolor{dialinecolor}
\node[anchor=west] at (11.600000\du,16.200000\du){$x=\{1\}$};
\definecolor{dialinecolor}{rgb}{0.000000, 0.000000, 0.000000}
\pgfsetstrokecolor{dialinecolor}
\node[anchor=west] at (6.000000\du,13.000000\du){};
\definecolor{dialinecolor}{rgb}{0.000000, 0.000000, 0.000000}
\pgfsetstrokecolor{dialinecolor}
\node[anchor=west] at (5.400000\du,18.800000\du){$x=\{\emptyset\}$};
\definecolor{dialinecolor}{rgb}{0.000000, 0.000000, 0.000000}
\pgfsetstrokecolor{dialinecolor}
\node[anchor=west] at (18.000000\du,20.500000\du){$v_1$};
\definecolor{dialinecolor}{rgb}{0.000000, 0.000000, 0.000000}
\pgfsetstrokecolor{dialinecolor}
\node[anchor=west] at (4.000000\du,9.500000\du){$v_2$};
\pgfsetlinewidth{0.0500000\du}
\pgfsetdash{{\pgflinewidth}{0.200000\du}}{0cm}
\pgfsetdash{{\pgflinewidth}{0.200000\du}}{0cm}
\pgfsetbuttcap
{
\definecolor{dialinecolor}{rgb}{0.000000, 0.000000, 0.000000}
\pgfsetfillcolor{dialinecolor}
\definecolor{dialinecolor}{rgb}{0.000000, 0.000000, 0.000000}
\pgfsetstrokecolor{dialinecolor}
\draw (9.000000\du,13.600000\du)--(12.600000\du,10.000000\du);
}
\definecolor{dialinecolor}{rgb}{0.000000, 0.000000, 0.000000}
\pgfsetstrokecolor{dialinecolor}
\node[anchor=west] at (9.800000\du,13.000000\du){$45^o$};
\pgfsetlinewidth{0.0500000\du}
\pgfsetdash{{\pgflinewidth}{0.200000\du}}{0cm}
\pgfsetdash{{\pgflinewidth}{0.200000\du}}{0cm}
\pgfsetbuttcap
{
\definecolor{dialinecolor}{rgb}{0.000000, 0.000000, 0.000000}
\pgfsetfillcolor{dialinecolor}
\definecolor{dialinecolor}{rgb}{0.000000, 0.000000, 0.000000}
\pgfsetstrokecolor{dialinecolor}
\draw (9.000000\du,10.000000\du)--(9.000000\du,20.000000\du);
}
\pgfsetlinewidth{0.0500000\du}
\pgfsetdash{{\pgflinewidth}{0.200000\du}}{0cm}
\pgfsetdash{{\pgflinewidth}{0.200000\du}}{0cm}
\pgfsetbuttcap
{
\definecolor{dialinecolor}{rgb}{0.000000, 0.000000, 0.000000}
\pgfsetfillcolor{dialinecolor}
\definecolor{dialinecolor}{rgb}{0.000000, 0.000000, 0.000000}
\pgfsetstrokecolor{dialinecolor}
\draw (9.000000\du,13.600000\du)--(10.000000\du,13.600000\du);
}
\pgfsetlinewidth{0.0500000\du}
\pgfsetdash{{\pgflinewidth}{0.200000\du}}{0cm}
\pgfsetdash{{\pgflinewidth}{0.200000\du}}{0cm}
\pgfsetbuttcap
{
\definecolor{dialinecolor}{rgb}{0.000000, 0.000000, 0.000000}
\pgfsetfillcolor{dialinecolor}
\definecolor{dialinecolor}{rgb}{0.000000, 0.000000, 0.000000}
\pgfsetstrokecolor{dialinecolor}
\draw (5.000000\du,17.000000\du)--(9.000000\du,17.000000\du);
}
\definecolor{dialinecolor}{rgb}{0.000000, 0.000000, 0.000000}
\pgfsetstrokecolor{dialinecolor}
\node[anchor=west] at (1.800000\du,17.200000\du){$w_2-w_1$};
\definecolor{dialinecolor}{rgb}{0.000000, 0.000000, 0.000000}
\pgfsetstrokecolor{dialinecolor}
\node[anchor=west] at (3.200000\du,13.800000\du){$2w_2$};
\definecolor{dialinecolor}{rgb}{0.000000, 0.000000, 0.000000}
\pgfsetstrokecolor{dialinecolor}
\node[anchor=west] at (5.400000\du,15.000000\du){$x=\{2\}$};
\definecolor{dialinecolor}{rgb}{0.000000, 0.000000, 0.000000}
\pgfsetstrokecolor{dialinecolor}
\node[anchor=west] at (9.000000\du,10.00000\du){$x=\{1,2\}$};
\end{tikzpicture}
\end{subfigure}
\caption{The allocation function of the online greedy algorithm.}\label{fig:greedy}
\end{figure}
However, this allocation cannot be made truthful by any payment scheme. The main reason behind this is the discontinuity of the allocation of item $2$ as the value of player $V$ for item $1$ ranges from $[0,\infty)$. In the region around $v_1\in [w_1-\delta,w_1+\delta]$, the threshold value for winning item $2$ is a discontinuous function of the value of player $V$ for item $1$. This creates an incentive for a type $v$ who lies around this point of discontinuity and loses item $2$, to lie about his value for item $1$ and win item $2$ instead, paying the lower of the two thresholds around the threshold discontinuity point.}

\vsedit{After examining how simple solutions fail, we now move to sketching the proof of our main result and outlaying our characterization of DSIC+IR mechanisms. Characterizing truthful mechanisms, even in a multi-dimensional setting, is fairly well understood 
in the case of a single agent  \cite{Rochet1987necessary,Archer2014truthful,Saks2005weak,Frongillo2014general}
(equivalently, when the bids of the other agents are fixed). As a first step, our characterization also needs to blend in the requirement of sequential consistency to these characterization.\footnote{We note that potentially some preliminary lemmas that we present could be derived by ``heavy hammers'' in mechanism design, such as weak monotonicity. However, we felt that invoking such results, for proving very preliminary basic lemmas would do more to distract the reader than to help in the understanding. Hence we prove them from first principles, since our setting of two bidders and two items is simple enough). The full characterization presented in Figure \ref{fig:char-main} cannot be shown by simply invoking such results and requires more subtle arguments that we present.}}

Our characterization starts from observing that since the mechanism has to be DSIC+IR even if one item arrives, then it has to be that for any value $w$ of $W$, the allocation of item $1$ to player $V$ has to look like a single item allocation, i.e. he gets the item if his value $v_1$ for item $1$ is above some threshold $\pi_1(w)$. Given this, DSIC+IR constraints, together with the fact that the mechanism must be achieving a finite approximation, gives us a complete characterization of how the allocation and payment function of a single player must look  for the second item, for any fixed valuation of his opponent (modulo a corner case that we deal with in the technical part). 

In particular, we show that for any fixed $w$ of player $W$, the allocation of item $2$ to player $V$, happens if:
\begin{itemize}
\item Player $V$ loses item $1$, i.e. $v_1\leq \pi_1(w)$ and $v_2\geq \pi_2(w)$ for some threshold $\pi_2(w)\geq \pi_1(w)$
\item Player $V$ wins item $1$ and $v_2-v_1 \geq \pi_2(w)-\pi_1(w)$.
\end{itemize}
In any case if player $V$ happens to win item $2$, then he has to pay $\pi_2(w)$, irrespective of whether he also won or not item $1$. This characterization is given pictorially in Figure \ref{fig:char-main}. The analogue characterization also holds for player $W$, keeping player $V$'s valuation $v$ fixed.
\begin{figure}[h]
\centering
\begin{subfigure}[b]{0.45\textwidth}
\centering
\ifx\du\undefined
  \newlength{\du}
\fi
\setlength{\du}{15\unitlength}
\begin{tikzpicture}
\pgftransformxscale{1.000000}
\pgftransformyscale{-1.000000}
\definecolor{dialinecolor}{rgb}{0.000000, 0.000000, 0.000000}
\pgfsetstrokecolor{dialinecolor}
\definecolor{dialinecolor}{rgb}{1.000000, 1.000000, 1.000000}
\pgfsetfillcolor{dialinecolor}
\pgfsetlinewidth{0.100000\du}
\pgfsetdash{}{0pt}
\pgfsetdash{}{0pt}
\pgfsetbuttcap
{
\definecolor{dialinecolor}{rgb}{0.000000, 0.000000, 0.000000}
\pgfsetfillcolor{dialinecolor}
\pgfsetarrowsend{stealth}
\definecolor{dialinecolor}{rgb}{0.000000, 0.000000, 0.000000}
\pgfsetstrokecolor{dialinecolor}
\draw (5.000000\du,20.000000\du)--(18.000000\du,20.000000\du);
}
\pgfsetlinewidth{0.100000\du}
\pgfsetdash{}{0pt}
\pgfsetdash{}{0pt}
\pgfsetbuttcap
{
\definecolor{dialinecolor}{rgb}{0.000000, 0.000000, 0.000000}
\pgfsetfillcolor{dialinecolor}
\pgfsetarrowsend{stealth}
\definecolor{dialinecolor}{rgb}{0.000000, 0.000000, 0.000000}
\pgfsetstrokecolor{dialinecolor}
\draw (5.000000\du,20.000000\du)--(5.000000\du,10.000000\du);
}
\pgfsetlinewidth{0.0500000\du}
\pgfsetdash{{\pgflinewidth}{0.200000\du}}{0cm}
\pgfsetdash{{\pgflinewidth}{0.200000\du}}{0cm}
\pgfsetbuttcap
{
\definecolor{dialinecolor}{rgb}{0.000000, 0.000000, 0.000000}
\pgfsetfillcolor{dialinecolor}
\definecolor{dialinecolor}{rgb}{0.000000, 0.000000, 0.000000}
\pgfsetstrokecolor{dialinecolor}
\draw (8.400000\du,19.800000\du)--(8.400000\du,9.800000\du);
}
\definecolor{dialinecolor}{rgb}{0.000000, 0.000000, 0.000000}
\pgfsetstrokecolor{dialinecolor}
\node[anchor=west] at (7.700000\du,20.800000\du){$\pi_1(w)$};
\definecolor{dialinecolor}{rgb}{0.000000, 0.000000, 0.000000}
\pgfsetstrokecolor{dialinecolor}
\node[anchor=west] at (12.381105\du,16.261912\du){$x=\{1\}$};
\definecolor{dialinecolor}{rgb}{0.000000, 0.000000, 0.000000}
\pgfsetstrokecolor{dialinecolor}
\node[anchor=west] at (6.000000\du,13.000000\du){};
\definecolor{dialinecolor}{rgb}{0.000000, 0.000000, 0.000000}
\pgfsetstrokecolor{dialinecolor}
\node[anchor=west] at (5.100000\du,14.000000\du){$x=\{2\}$};
\definecolor{dialinecolor}{rgb}{0.000000, 0.000000, 0.000000}
\pgfsetstrokecolor{dialinecolor}
\node[anchor=west] at (18.000000\du,20.500000\du){$v_1$};
\definecolor{dialinecolor}{rgb}{0.000000, 0.000000, 0.000000}
\pgfsetstrokecolor{dialinecolor}
\node[anchor=west] at (4.000000\du,9.500000\du){$v_2$};
\pgfsetlinewidth{0.0500000\du}
\pgfsetdash{{\pgflinewidth}{0.200000\du}}{0cm}
\pgfsetdash{{\pgflinewidth}{0.200000\du}}{0cm}
\pgfsetbuttcap
{
\definecolor{dialinecolor}{rgb}{0.000000, 0.000000, 0.000000}
\pgfsetfillcolor{dialinecolor}
\definecolor{dialinecolor}{rgb}{0.000000, 0.000000, 0.000000}
\pgfsetstrokecolor{dialinecolor}
\draw (5.000000\du,14.600000\du)--(8.300000\du,14.600000\du);
}
\definecolor{dialinecolor}{rgb}{0.000000, 0.000000, 0.000000}
\pgfsetstrokecolor{dialinecolor}
\node[anchor=west] at (5.100000\du,18.000000\du){$x=\{\emptyset\}$};
\definecolor{dialinecolor}{rgb}{0.000000, 0.000000, 0.000000}
\pgfsetstrokecolor{dialinecolor}
\node[anchor=west] at (5.100000\du,13.000000\du){$p=\pi_2(w)$};
\definecolor{dialinecolor}{rgb}{0.000000, 0.000000, 0.000000}
\pgfsetstrokecolor{dialinecolor}
\node[anchor=west] at (2.800000\du,14.600000\du){$\pi_2(w)$};
\definecolor{dialinecolor}{rgb}{0.000000, 0.000000, 0.000000}
\pgfsetstrokecolor{dialinecolor}
\node[anchor=west] at (12.381105\du,15.261912\du){$p=\pi_1(w)$};
\pgfsetlinewidth{0.0500000\du}
\pgfsetdash{{\pgflinewidth}{0.200000\du}}{0cm}
\pgfsetdash{{\pgflinewidth}{0.200000\du}}{0cm}
\pgfsetbuttcap
{
\definecolor{dialinecolor}{rgb}{0.000000, 0.000000, 0.000000}
\pgfsetfillcolor{dialinecolor}
\definecolor{dialinecolor}{rgb}{0.000000, 0.000000, 0.000000}
\pgfsetstrokecolor{dialinecolor}
\draw (8.400000\du,14.600000\du)--(14.000000\du,10.000000\du);
}
\pgfsetlinewidth{0.0500000\du}
\pgfsetdash{{\pgflinewidth}{0.200000\du}}{0cm}
\pgfsetdash{{\pgflinewidth}{0.200000\du}}{0cm}
\pgfsetbuttcap
{
\definecolor{dialinecolor}{rgb}{0.000000, 0.000000, 0.000000}
\pgfsetfillcolor{dialinecolor}
\definecolor{dialinecolor}{rgb}{0.000000, 0.000000, 0.000000}
\pgfsetstrokecolor{dialinecolor}
\draw (8.400000\du,14.600000\du)--(10.600000\du,14.600000\du);
}
\definecolor{dialinecolor}{rgb}{0.000000, 0.000000, 0.000000}
\pgfsetstrokecolor{dialinecolor}
\node[anchor=west] at (9.410850\du,14.073000\du){$45^o$};
\definecolor{dialinecolor}{rgb}{0.000000, 0.000000, 0.000000}
\pgfsetstrokecolor{dialinecolor}
\node[anchor=west] at (8.800000\du,11.000000\du){$x=\{1,2\}$};
\definecolor{dialinecolor}{rgb}{0.000000, 0.000000, 0.000000}
\pgfsetstrokecolor{dialinecolor}
\node[anchor=west] at (8.800000\du,10.200000\du){$p=\pi_2(w)$};
\end{tikzpicture}
\end{subfigure}
\caption{The complete characterization of allocation and payment functions of DSIC+IR mechanisms that achieve finite approximation factors.}\label{fig:char-main}
\end{figure}
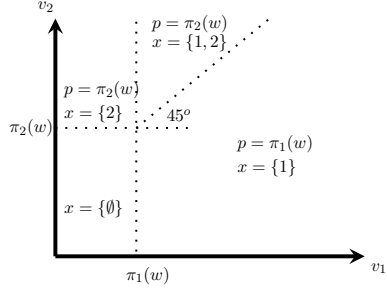
Unlike the VCG mechanism the latter figure has the vertical line that determines the allocation to item $1$. This line is what is causing most of the trouble in achieving constant factor approximations. Intuitively, what VCG would want to do in the case that item $2$ actually arrives, is erase the vertical line and continue the diagonal line until it hits the axes. Finally, unlike the greedy algorithm, a DSIC mechanism has to have the boundary that defines the region for allocation of item $2$, be a continuous function (i.e. the $45^o$ line must meet the horizontal line at $\pi_2(w)$). These are two crucial properties that any DSIC+IR and sequentially consistent mechanism has to satisfy and the reason why simple solutions do not work.

However, this single-player characterization is still not sufficient to draw the conclusion that no DSIC+IR mechanism can achieve a finite approximation.  \vsedit{To arrive at the impossibility, we need to argue about the allocation and payment of both players simultaneously. Given that this requires arguing about functions in the four dimensional space of $v_1,v_2,w_1,w_2$, this presents an uphill battle.} The key argument in the epilogue of the proof of our main result, which is also the hardest and most intricate part of our analysis is as follows: assume for now that all items are always allocated to some bidder, hence the allocation of one player is the complement of the allocation of the other. 
Now consider how the allocation changes as bidder $W$'s value for item $1$ shifts (while his value for item $2$ remains fixed). 
A crucial and non-trivial step in the proof is to show that if the allocation of bidder $V$ changes as a result, it must be that there are points in bidder $V$'s type space that go from  winning both items to losing item $2$. 
Hence, bidder $W$ goes from winning nothing to winning at least item $2$. This means that solely the value for item $1$ of bidder $W$ determined whether he wins nothing or item $2$. This cannot be truthful for bidder $W$, since he would want to slightly misreport his value for item $1$ and win item $2$.
Hence, bidder $V$'s allocation remains completely unchanged and subsequently bidder's $W$ allocation remains completely unchanged. However, this also implies that bidder's $W$ allocation for item $1$ is also insensitive to his bid for item $1$. The latter can lead to unbounded approximation ratios. 

Finally, the proof of the theorem requires dropping the assumption that all items are always allocated. To achieve this we focus only on small sub-regions in bidder $W$'s and $V$'s type space, in which both items have to be allocated to some bidder if we want to have finite approximation and which are sufficient for making the argument that insensitivity of bidder $W$'s allocation of item $1$ is insensitive to his report for item $1$ in this sub-region leads to unbounded approximation factor. 
Defining these sub-regions requires a series of delicate lemmas that argue about potential changes to one player's allocation function as the other player's type changes, under the assumption that the mechanism achieves some finite approximation.



\section{Proof of Main Result}\label{sec:proof}


\vsedit{In this section we provide the complete proof of our main theorem. We begin by some basic characterizations of properties that all DSIC and IR mechanisms need to satisfy. We then refine further this characterization when we also add the requirement that the mechanism achieves a constant factor approximation. Finally, we present the core final arguments of our proof which unlike the first two sections, requires arguing simultaneously about the allocation and payment functions of both players, rather than how the allocation of a single player looks like, holding the other player fixed.}

In most of the Lemma statements, we show that the allocation and payment functions are restricted in some ways. These restrictions apply to both $\allocv(\cdot,\cdot)$ and $\allocw(\cdot,\cdot)$, and similarly to both $\pricev(\cdot,\cdot)$ and $\pricew(\cdot, \cdot)$. However we prove the statements only for $\allocv(\cdot, \cdot)$ and $\pricev(\cdot, \cdot)$ as the corresponding proof for bidder $W$ is identical. 


\subsection{DSIC + IR Mechanisms}
In this section, we show necessary conditions for a deterministic online mechanism to be DSIC and IR. These conditions are very similar to characterizations for the offline case \cite{Chawla2010power,Thanassoulis2004haggling}. 
Essentially, what this amounts to is that every DSIC+ IR mechanism looks like this from the point of view of $V$: 
the two items are offered at prices $\pi_1(w)$ and $\pi_2(w)$, that do not depend on $v$. In case $V$ is allocated only one of the two items, he pays the corresponding price. We will deal with the case where she is allocated both items in the next part. 
In the next 3 lemmas, we only state the conclusion for $V$, and the {\em symmetric statement for $W$ is true as well}. 

\begin{lemma}
\label{lem:VerticalLine}
For every deterministic DSIC+IR mechanism, there must exist a threshold function $\pi_1(w)$  such that:
\begin{align*}
&\text{$\forall w \in \R_{+}^2$, $\forall v \in \R_{+}^2:$}\ \  \allocv_1(v,w) = 
\begin{cases}
1 & \text{ if } v_1 > \pi_1(w)\\
0 & \text{ if } v_1 < \pi_1(w)
\end{cases}
\comment{
&\text{$\forall v \in \R_{+}^2$, $\forall w \in \R_{+}^2:$}\ \ 
\allocw_1(v,w) = 
\begin{cases}
1 & \text{ if } w_1 > \phi_1(v)\\
0 & \text{ if } w_1 < \phi_1(v)
\end{cases}
}
\end{align*}
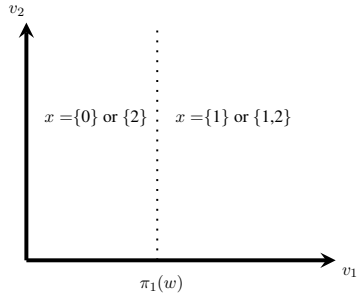
\begin{figure}[h]
\centering
\ifx\du\undefined
  \newlength{\du}
\fi
\setlength{\du}{15\unitlength}
\begin{tikzpicture}
\pgftransformxscale{1.000000}
\pgftransformyscale{-1.000000}
\definecolor{dialinecolor}{rgb}{0.000000, 0.000000, 0.000000}
\pgfsetstrokecolor{dialinecolor}
\definecolor{dialinecolor}{rgb}{1.000000, 1.000000, 1.000000}
\pgfsetfillcolor{dialinecolor}
\pgfsetlinewidth{0.100000\du}
\pgfsetdash{}{0pt}
\pgfsetdash{}{0pt}
\pgfsetbuttcap
{
\definecolor{dialinecolor}{rgb}{0.000000, 0.000000, 0.000000}
\pgfsetfillcolor{dialinecolor}
\pgfsetarrowsend{stealth}
\definecolor{dialinecolor}{rgb}{0.000000, 0.000000, 0.000000}
\pgfsetstrokecolor{dialinecolor}
\draw (5.000000\du,20.000000\du)--(18.000000\du,20.000000\du);
}
\pgfsetlinewidth{0.100000\du}
\pgfsetdash{}{0pt}
\pgfsetdash{}{0pt}
\pgfsetbuttcap
{
\definecolor{dialinecolor}{rgb}{0.000000, 0.000000, 0.000000}
\pgfsetfillcolor{dialinecolor}
\pgfsetarrowsend{stealth}
\definecolor{dialinecolor}{rgb}{0.000000, 0.000000, 0.000000}
\pgfsetstrokecolor{dialinecolor}
\draw (5.000000\du,20.000000\du)--(5.000000\du,10.000000\du);
}
\pgfsetlinewidth{0.0500000\du}
\pgfsetdash{{\pgflinewidth}{0.200000\du}}{0cm}
\pgfsetdash{{\pgflinewidth}{0.200000\du}}{0cm}
\pgfsetbuttcap
{
\definecolor{dialinecolor}{rgb}{0.000000, 0.000000, 0.000000}
\pgfsetfillcolor{dialinecolor}
\definecolor{dialinecolor}{rgb}{0.000000, 0.000000, 0.000000}
\pgfsetstrokecolor{dialinecolor}
\draw (10.500000\du,20.000000\du)--(10.500000\du,10.000000\du);
}
\definecolor{dialinecolor}{rgb}{0.000000, 0.000000, 0.000000}
\pgfsetstrokecolor{dialinecolor}
\node[anchor=west] at (9.500000\du,21.000000\du){$\pi_1(w)$};
\definecolor{dialinecolor}{rgb}{0.000000, 0.000000, 0.000000}
\pgfsetstrokecolor{dialinecolor}
\node[anchor=west] at (11.000000\du,14.000000\du){$x=$\{1\} or \{1,2\}};
\definecolor{dialinecolor}{rgb}{0.000000, 0.000000, 0.000000}
\pgfsetstrokecolor{dialinecolor}
\node[anchor=west] at (6.000000\du,13.000000\du){};
\definecolor{dialinecolor}{rgb}{0.000000, 0.000000, 0.000000}
\pgfsetstrokecolor{dialinecolor}
\node[anchor=west] at (5.500000\du,14.000000\du){$x=$\{0\} or \{2\}};
\definecolor{dialinecolor}{rgb}{0.000000, 0.000000, 0.000000}
\pgfsetstrokecolor{dialinecolor}
\node[anchor=west] at (18.000000\du,20.500000\du){$v_1$};
\definecolor{dialinecolor}{rgb}{0.000000, 0.000000, 0.000000}
\pgfsetstrokecolor{dialinecolor}
\node[anchor=west] at (4.000000\du,9.500000\du){$v_2$};
\end{tikzpicture}
\caption{Characterization of allocation for item $1$}
\end{figure}
\end{lemma}
Proof in Appendix~\ref{app:missing_proofs}.

\begin{remark}
It is good to remind oneself that while the statement of Lemma~\ref{lem:VerticalLine} is true always, the claims made in the proof of Lemma~\ref{lem:VerticalLine} regarding payments are true only in the case that was used in the proof, namely, where item $1$ alone arrives. This is because we allow the mechanism designer to compute payments after knowing whether one item arrived or two items arrived, and thus the conclusions regarding payments drawn from one case need not apply to the other. The allocation function on the other hand has to be determined immediately after each item arrives. Thus the conclusions drawn in the proof regarding allocation function by referring to the case where only one item arrives extends to the case where both items arrive too. Note that the Lemma statement refers only to the allocation function.  
\end{remark}


\begin{lemma}
\label{lem:HorizontalLine}
For every deterministic DSIC+IR mechanism, there must exist a  threshold function $\pi_2(w)$ such that:
\begin{align*}
\text{$\forall w \in \R_+^2$, } \forall v \in \R_{+}^2 \text{ s.t. } v_1 < \pi_1(w):   &~~\allocv_2(v,w) = 
\begin{cases}
1 & \text{ if  } v_2 > \pi_2(w) \\
0 & \text{ if } v_2 < \pi_2 (w)
\end{cases}\\
&~~
\pricev(v,w) = 
\begin{cases}
\pi_2(w) & \text{ if  } v_2 > \pi_2(w) \\
0  & \text{ if } v_2 < \pi_2(w)
\end{cases}
\end{align*}
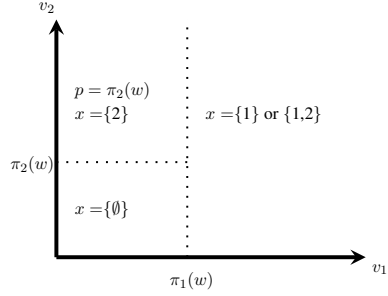
\begin{figure}[h]
\centering
\ifx\du\undefined
  \newlength{\du}
\fi
\setlength{\du}{15\unitlength}
\begin{tikzpicture}
\pgftransformxscale{1.000000}
\pgftransformyscale{-1.000000}
\definecolor{dialinecolor}{rgb}{0.000000, 0.000000, 0.000000}
\pgfsetstrokecolor{dialinecolor}
\definecolor{dialinecolor}{rgb}{1.000000, 1.000000, 1.000000}
\pgfsetfillcolor{dialinecolor}
\pgfsetlinewidth{0.100000\du}
\pgfsetdash{}{0pt}
\pgfsetdash{}{0pt}
\pgfsetbuttcap
{
\definecolor{dialinecolor}{rgb}{0.000000, 0.000000, 0.000000}
\pgfsetfillcolor{dialinecolor}
\pgfsetarrowsend{stealth}
\definecolor{dialinecolor}{rgb}{0.000000, 0.000000, 0.000000}
\pgfsetstrokecolor{dialinecolor}
\draw (5.000000\du,20.000000\du)--(18.000000\du,20.000000\du);
}
\pgfsetlinewidth{0.100000\du}
\pgfsetdash{}{0pt}
\pgfsetdash{}{0pt}
\pgfsetbuttcap
{
\definecolor{dialinecolor}{rgb}{0.000000, 0.000000, 0.000000}
\pgfsetfillcolor{dialinecolor}
\pgfsetarrowsend{stealth}
\definecolor{dialinecolor}{rgb}{0.000000, 0.000000, 0.000000}
\pgfsetstrokecolor{dialinecolor}
\draw (5.000000\du,20.000000\du)--(5.000000\du,10.000000\du);
}
\pgfsetlinewidth{0.0500000\du}
\pgfsetdash{{\pgflinewidth}{0.200000\du}}{0cm}
\pgfsetdash{{\pgflinewidth}{0.200000\du}}{0cm}
\pgfsetbuttcap
{
\definecolor{dialinecolor}{rgb}{0.000000, 0.000000, 0.000000}
\pgfsetfillcolor{dialinecolor}
\definecolor{dialinecolor}{rgb}{0.000000, 0.000000, 0.000000}
\pgfsetstrokecolor{dialinecolor}
\draw (10.500000\du,20.000000\du)--(10.500000\du,10.000000\du);
}
\definecolor{dialinecolor}{rgb}{0.000000, 0.000000, 0.000000}
\pgfsetstrokecolor{dialinecolor}
\node[anchor=west] at (9.500000\du,21.000000\du){$\pi_1(w)$};
\definecolor{dialinecolor}{rgb}{0.000000, 0.000000, 0.000000}
\pgfsetstrokecolor{dialinecolor}
\node[anchor=west] at (11.000000\du,14.000000\du){$x=$\{1\} or \{1,2\}};
\definecolor{dialinecolor}{rgb}{0.000000, 0.000000, 0.000000}
\pgfsetstrokecolor{dialinecolor}
\node[anchor=west] at (6.000000\du,13.000000\du){};
\definecolor{dialinecolor}{rgb}{0.000000, 0.000000, 0.000000}
\pgfsetstrokecolor{dialinecolor}
\node[anchor=west] at (5.500000\du,14.000000\du){$x=$\{2\}};
\definecolor{dialinecolor}{rgb}{0.000000, 0.000000, 0.000000}
\pgfsetstrokecolor{dialinecolor}
\node[anchor=west] at (18.000000\du,20.500000\du){$v_1$};
\definecolor{dialinecolor}{rgb}{0.000000, 0.000000, 0.000000}
\pgfsetstrokecolor{dialinecolor}
\node[anchor=west] at (4.000000\du,9.500000\du){$v_2$};
\pgfsetlinewidth{0.0500000\du}
\pgfsetdash{{\pgflinewidth}{0.200000\du}}{0cm}
\pgfsetdash{{\pgflinewidth}{0.200000\du}}{0cm}
\pgfsetbuttcap
{
\definecolor{dialinecolor}{rgb}{0.000000, 0.000000, 0.000000}
\pgfsetfillcolor{dialinecolor}
\definecolor{dialinecolor}{rgb}{0.000000, 0.000000, 0.000000}
\pgfsetstrokecolor{dialinecolor}
\draw (5.000000\du,16.000000\du)--(10.500000\du,16.000000\du);
}
\definecolor{dialinecolor}{rgb}{0.000000, 0.000000, 0.000000}
\pgfsetstrokecolor{dialinecolor}
\node[anchor=west] at (5.500000\du,18.000000\du){$x=$\{$\emptyset$\}};
\definecolor{dialinecolor}{rgb}{0.000000, 0.000000, 0.000000}
\pgfsetstrokecolor{dialinecolor}
\node[anchor=west] at (5.500000\du,13.000000\du){$p=\pi_2(w)$};
\definecolor{dialinecolor}{rgb}{0.000000, 0.000000, 0.000000}
\pgfsetstrokecolor{dialinecolor}
\node[anchor=west] at (2.800000\du,16.102776\du){$\pi_2(w)$};
\end{tikzpicture}
\caption{Characterization of allocation and payment for item $2$, conditional on losing item $1$.}
\end{figure}
\end{lemma}
Proof in Appendix~\ref{app:missing_proofs}.

\begin{remark}
	A possible source of confusion here is for the reader to point out ``The proof of Lemma~\ref{lem:HorizontalLine} point 3 just proves that $\pi_2$ remains the same for all $v: v_1 < \pi_1$. But what about $v: v_1 \geq \pi_1$? How could you conclude that $\pi_2$ is purely a function of $w$ without arguing about $v: v_1 \geq \pi_1$?''. This is not a meaningful question because $\pi_2$ is a threshold function that was introduced to describe the allocation in the region $v: v_1 < \pi_1$, i.e., $\pi_2$ has a domain of $[0,\pi_1) \times \R_{+}^2$.  
	Thus the question of ``what happens to $\pi_2$ when $v_1 \geq \pi_1$'' is meaningless.
\end{remark}

\comment{\begin{figure}[htpb]
\centering{
\includegraphics[scale=.5]{charac2.png}
}
\caption{Characterization of allocation functions.} \label{fig:char}
\end{figure}
}

\begin{lemma}
\label{lem:1sPayPi1}
For every deterministic DSIC+IR mechanism,  
$$\forall w \in \R_+^2, \forall v \in \R_{+}^2  \text{ s.t. } \{\allocv_1(v,w) = 1, \ \allocv_2(v,w) = 0\},\text{ we have } \pricev(v,w) = \pi_1$$ 
\begin{figure}[h]
\centering
\ifx\du\undefined
  \newlength{\du}
\fi
\setlength{\du}{15\unitlength}
\begin{tikzpicture}
\pgftransformxscale{1.000000}
\pgftransformyscale{-1.000000}
\definecolor{dialinecolor}{rgb}{0.000000, 0.000000, 0.000000}
\pgfsetstrokecolor{dialinecolor}
\definecolor{dialinecolor}{rgb}{1.000000, 1.000000, 1.000000}
\pgfsetfillcolor{dialinecolor}
\pgfsetlinewidth{0.100000\du}
\pgfsetdash{}{0pt}
\pgfsetdash{}{0pt}
\pgfsetbuttcap
{
\definecolor{dialinecolor}{rgb}{0.000000, 0.000000, 0.000000}
\pgfsetfillcolor{dialinecolor}
\pgfsetarrowsend{stealth}
\definecolor{dialinecolor}{rgb}{0.000000, 0.000000, 0.000000}
\pgfsetstrokecolor{dialinecolor}
\draw (5.000000\du,20.000000\du)--(18.000000\du,20.000000\du);
}
\pgfsetlinewidth{0.100000\du}
\pgfsetdash{}{0pt}
\pgfsetdash{}{0pt}
\pgfsetbuttcap
{
\definecolor{dialinecolor}{rgb}{0.000000, 0.000000, 0.000000}
\pgfsetfillcolor{dialinecolor}
\pgfsetarrowsend{stealth}
\definecolor{dialinecolor}{rgb}{0.000000, 0.000000, 0.000000}
\pgfsetstrokecolor{dialinecolor}
\draw (5.000000\du,20.000000\du)--(5.000000\du,10.000000\du);
}
\pgfsetlinewidth{0.0500000\du}
\pgfsetdash{{\pgflinewidth}{0.200000\du}}{0cm}
\pgfsetdash{{\pgflinewidth}{0.200000\du}}{0cm}
\pgfsetbuttcap
{
\definecolor{dialinecolor}{rgb}{0.000000, 0.000000, 0.000000}
\pgfsetfillcolor{dialinecolor}
\definecolor{dialinecolor}{rgb}{0.000000, 0.000000, 0.000000}
\pgfsetstrokecolor{dialinecolor}
\draw (10.500000\du,20.000000\du)--(10.500000\du,10.000000\du);
}
\definecolor{dialinecolor}{rgb}{0.000000, 0.000000, 0.000000}
\pgfsetstrokecolor{dialinecolor}
\node[anchor=west] at (9.500000\du,21.000000\du){$\pi_1(w)$};
\definecolor{dialinecolor}{rgb}{0.000000, 0.000000, 0.000000}
\pgfsetstrokecolor{dialinecolor}
\node[anchor=west] at (16.000000\du,11.500000\du){$x=$\{1\}};
\definecolor{dialinecolor}{rgb}{0.000000, 0.000000, 0.000000}
\pgfsetstrokecolor{dialinecolor}
\node[anchor=west] at (6.000000\du,13.000000\du){};
\definecolor{dialinecolor}{rgb}{0.000000, 0.000000, 0.000000}
\pgfsetstrokecolor{dialinecolor}
\node[anchor=west] at (5.500000\du,14.000000\du){$x=$\{2\}};
\definecolor{dialinecolor}{rgb}{0.000000, 0.000000, 0.000000}
\pgfsetstrokecolor{dialinecolor}
\node[anchor=west] at (18.000000\du,20.500000\du){$v_1$};
\definecolor{dialinecolor}{rgb}{0.000000, 0.000000, 0.000000}
\pgfsetstrokecolor{dialinecolor}
\node[anchor=west] at (4.000000\du,9.500000\du){$v_2$};
\pgfsetlinewidth{0.0500000\du}
\pgfsetdash{{\pgflinewidth}{0.200000\du}}{0cm}
\pgfsetdash{{\pgflinewidth}{0.200000\du}}{0cm}
\pgfsetbuttcap
{
\definecolor{dialinecolor}{rgb}{0.000000, 0.000000, 0.000000}
\pgfsetfillcolor{dialinecolor}
\definecolor{dialinecolor}{rgb}{0.000000, 0.000000, 0.000000}
\pgfsetstrokecolor{dialinecolor}
\draw (5.000000\du,16.000000\du)--(10.500000\du,16.000000\du);
}
\definecolor{dialinecolor}{rgb}{0.000000, 0.000000, 0.000000}
\pgfsetstrokecolor{dialinecolor}
\node[anchor=west] at (5.500000\du,18.000000\du){$x=\{\emptyset\}$};
\definecolor{dialinecolor}{rgb}{0.000000, 0.000000, 0.000000}
\pgfsetstrokecolor{dialinecolor}
\node[anchor=west] at (5.500000\du,13.000000\du){$p=\pi_2(w)$};
\definecolor{dialinecolor}{rgb}{0.000000, 0.000000, 0.000000}
\pgfsetstrokecolor{dialinecolor}
\node[anchor=west] at (2.800000\du,16.12776\du){$\pi_2(w)$};
\pgfsetlinewidth{0.0500000\du}
\pgfsetdash{{\pgflinewidth}{0.200000\du}}{0cm}
\pgfsetdash{{\pgflinewidth}{0.200000\du}}{0cm}
\pgfsetmiterjoin
\pgfsetbuttcap
\definecolor{dialinecolor}{rgb}{1.000000, 1.000000, 1.000000}
\pgfsetfillcolor{dialinecolor}
\pgfpathmoveto{\pgfpoint{19.000000\du}{13.000000\du}}
\pgfpathcurveto{\pgfpoint{20.000000\du}{13.000000\du}}{\pgfpoint{19.000000\du}{17.500000\du}}{\pgfpoint{16.500000\du}{15.000000\du}}
\pgfpathcurveto{\pgfpoint{14.000000\du}{12.500000\du}}{\pgfpoint{18.000000\du}{13.000000\du}}{\pgfpoint{19.000000\du}{13.000000\du}}
\pgfusepath{fill}
\definecolor{dialinecolor}{rgb}{0.000000, 0.000000, 0.000000}
\pgfsetstrokecolor{dialinecolor}
\pgfpathmoveto{\pgfpoint{19.000000\du}{13.000000\du}}
\pgfpathcurveto{\pgfpoint{20.000000\du}{13.000000\du}}{\pgfpoint{19.000000\du}{17.500000\du}}{\pgfpoint{16.500000\du}{15.000000\du}}
\pgfpathcurveto{\pgfpoint{14.000000\du}{12.500000\du}}{\pgfpoint{18.000000\du}{13.000000\du}}{\pgfpoint{19.000000\du}{13.000000\du}}
\pgfusepath{stroke}
\pgfsetlinewidth{0.0500000\du}
\pgfsetdash{{\pgflinewidth}{0.200000\du}}{0cm}
\pgfsetdash{{\pgflinewidth}{0.200000\du}}{0cm}
\pgfsetmiterjoin
\pgfsetbuttcap
\definecolor{dialinecolor}{rgb}{1.000000, 1.000000, 1.000000}
\pgfsetfillcolor{dialinecolor}
\pgfpathmoveto{\pgfpoint{17.000000\du}{17.000000\du}}
\pgfpathcurveto{\pgfpoint{18.000000\du}{17.000000\du}}{\pgfpoint{16.000000\du}{21.000000\du}}{\pgfpoint{13.500000\du}{18.500000\du}}
\pgfpathcurveto{\pgfpoint{11.000000\du}{16.000000\du}}{\pgfpoint{16.000000\du}{17.000000\du}}{\pgfpoint{17.000000\du}{17.000000\du}}
\pgfusepath{fill}
\definecolor{dialinecolor}{rgb}{0.000000, 0.000000, 0.000000}
\pgfsetstrokecolor{dialinecolor}
\pgfpathmoveto{\pgfpoint{17.000000\du}{17.000000\du}}
\pgfpathcurveto{\pgfpoint{18.000000\du}{17.000000\du}}{\pgfpoint{16.000000\du}{21.000000\du}}{\pgfpoint{13.500000\du}{18.500000\du}}
\pgfpathcurveto{\pgfpoint{11.000000\du}{16.000000\du}}{\pgfpoint{16.000000\du}{17.000000\du}}{\pgfpoint{17.000000\du}{17.000000\du}}
\pgfusepath{stroke}
\pgfsetlinewidth{0.0500000\du}
\pgfsetdash{{\pgflinewidth}{0.200000\du}}{0cm}
\pgfsetdash{{\pgflinewidth}{0.200000\du}}{0cm}
\pgfsetmiterjoin
\pgfsetbuttcap
\definecolor{dialinecolor}{rgb}{1.000000, 1.000000, 1.000000}
\pgfsetfillcolor{dialinecolor}
\pgfpathmoveto{\pgfpoint{15.000000\du}{12.000000\du}}
\pgfpathcurveto{\pgfpoint{16.000000\du}{12.000000\du}}{\pgfpoint{15.000000\du}{16.500000\du}}{\pgfpoint{12.500000\du}{14.000000\du}}
\pgfpathcurveto{\pgfpoint{10.000000\du}{11.500000\du}}{\pgfpoint{14.000000\du}{12.000000\du}}{\pgfpoint{15.000000\du}{12.000000\du}}
\pgfusepath{fill}
\definecolor{dialinecolor}{rgb}{0.000000, 0.000000, 0.000000}
\pgfsetstrokecolor{dialinecolor}
\pgfpathmoveto{\pgfpoint{15.000000\du}{12.000000\du}}
\pgfpathcurveto{\pgfpoint{16.000000\du}{12.000000\du}}{\pgfpoint{15.000000\du}{16.500000\du}}{\pgfpoint{12.500000\du}{14.000000\du}}
\pgfpathcurveto{\pgfpoint{10.000000\du}{11.500000\du}}{\pgfpoint{14.000000\du}{12.000000\du}}{\pgfpoint{15.000000\du}{12.000000\du}}
\pgfusepath{stroke}
\definecolor{dialinecolor}{rgb}{0.000000, 0.000000, 0.000000}
\pgfsetstrokecolor{dialinecolor}
\node[anchor=west] at (16.000000\du,10.500000\du){$p=\pi_1(w)$};
\definecolor{dialinecolor}{rgb}{0.000000, 0.000000, 0.000000}
\pgfsetstrokecolor{dialinecolor}
\node[anchor=west] at (12.00000\du,13.000000\du){$x=$\{1,2\}};
\definecolor{dialinecolor}{rgb}{0.000000, 0.000000, 0.000000}
\pgfsetstrokecolor{dialinecolor}
\node[anchor=west] at (16.00000\du,14.000000\du){$x=$\{1,2\}};
\definecolor{dialinecolor}{rgb}{0.000000, 0.000000, 0.000000}
\pgfsetstrokecolor{dialinecolor}
\node[anchor=west] at (13.500000\du,18.000000\du){$x=$\{1,2\}};
\end{tikzpicture}
\caption{Characterization of payment when only item $1$ is allocated.}
\end{figure}
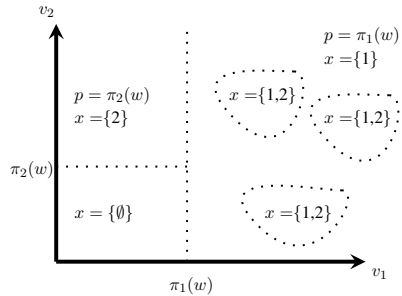
\end{lemma}
Proof in Appendix~\ref{app:missing_proofs}.

\subsection{Refining Characterization under Finite Approximation}
From now on, we assume that the mechanism also gets a finite approximation, and further refine the necessary conditions we have so far. Basically, we show finiteness and positivity of the thresholds. 
Further we show that $\pi_2 \geq \pi_1$ always, and when $\pi_2 > \pi_1$ these two thresholds together completely determine the allocation function for bidder $V$: in case he wins the first item, he is also allocated the second item if 
$v_2 - \pi_2 > v_1 - \pi_1$, and his payment is only $\pi_2$. 
If $\pi_2 = \pi_1$, then in one of the $4$ regions, namely $B_R^V$, the mechanism has the freedom to allocate to bidder $V$ 
just item $1$ or both items. This ambiguity in allocation is eliminated when $\pi_2 > \pi_1$, and  $V$ does not get item 2 in   
$B_R^V$, and pays $\pi_1$. 
All this is crucial when we construct the main argument in Section~\ref{sec:HApprox}.

The next two lemmas follow by considering $v,w$ values that would lead to unbounded approximations if the thresholds were 
either $\infty$ or $0$. Once again, for Lemmas \ref{lem:Pi1NotZero}-\ref{lem:1Below45Deg}, we only state the conclusion for $V$ and 
the symmetric statement for $W$ holds as well. 
\begin{lemma}
\label{lem:Pi1NotZero}
For every deterministic DSIC+IR mechanism that gets a finite approximation factor, 
\begin{enumerate}
\item It always allocates item 1 to some player. 
\item $\forall w \in \R_{+}^2$, we have $\pi_1 < \infty$. Further, $\forall w \in \R_{+}^2: w_1 > 0$, we have $\pi_1 > 0$. 
\end{enumerate}
\end{lemma}
\begin{proof}
Consider a mechanism that gets a finite approximation factor $H$, and consider the case where item 2 does not arrive.
	If for some $v,w$, the mechanism does not allocate item 1, then the approximation factor for the mechanism is unbounded.
	
	If at some $w$ the mechanism had $\pi_1 = \infty$,  let $v_1 = (H+1)w_1$. In this case, bidder $V$ still doesn't get the item, making the approximation factor to be $H+1$ which is larger than $H$. 
	
	If for some $w: w_1 \neq 0$, the said mechanism had $\pi_1 = 0$, let $v_1 = \frac{w_1}{H+1}$. In this case bidder $V$ still gets the item, making the approximation factor to be $H+1$ which is larger than $H$. 
\end{proof}

\begin{lemma}
\label{lem:Pi2NotZero}
For every deterministic DSIC+IR mechanism that gets a finite approximation factor,
$ \forall w \in \R_{+}^2: w_1 > 0, \text{ we  have } \pi_2 < \infty .$
\comment{\begin{enumerate}
\item $\forall w \in \R_{+}^2: w_1 > 0$, we have $\pi_2 < \infty$. 
\item $\forall v \in \R_{+}^2: v_1 > 0$, we have $\phi_2 < \infty$. 
\end{enumerate}
}
\end{lemma}
\begin{proof}
Fix a $w \in \R_+^2: w_1 > 0$. By Lemma~\ref{lem:Pi1NotZero}, note that $0 < \pi_1 < \infty$ for such a $w$. Suppose $\pi_2 = \infty$ for the said point $w$. Now, consider a point $v$ s.t. $v_2 >> \max\{v_1, w_1, w_2\}$ and $v_1 < \pi_1$. At such a $v$, item $2$ should necessarily go to bidder $V$ for a finite approximation. Having $\pi_2 = \infty$ implies this will not happen. 
\end{proof}

\paragraph{Notation} 
For every deterministic DSIC+IR mechanism that gets a finite approximation factor, the following four sets are well defined. Fix a bid $w \in \R_{+}^2: w_1 > 0 $ of bidder $W$. By Lemma~\ref{lem:Pi1NotZero}, such a mechanism has $\pi_1 > 0$, and thus the domain of $\pi_2$, namely $[0,\pi_1)\times \R_{+}^2$ is non-empty. Thus the value of $\pi_2$ is well-defined for any fixed $w \in \R_+^2$ with $w_1>0$, making the $4$ sets below well-defined. \begin{enumerate}
\item Bottom-Left-V: $B_L^V(w) = \{v \in \R_{+}^2: v_1 < \pi_1, v_2 < \pi_2\}$ . 
\item Bottom-Right-V: $B_R^V(w) = \{v \in \R_{+}^2: v_1 > \pi_1, v_2 - \pi_2 < v_1 - \pi_1\}$ .
\item Top-Left-V: $T_L^V(w) = \{v \in \R_{+}^2: v_1 < \pi_1, v_2 > \pi_2\}$ . 
\item Top-Right-V: $T_R^V(w) = \{v  \in \R_{+}^2: v_1 > \pi_1, v_2 - \pi_2 > v_1 - \pi_1\}$.
\end{enumerate}
Further both $\pi_1$ and $\pi_2$ are finite, making all four sets non-empty. We will use the non-emptiness of $T_R^V$ later in the proof. 
\begin{figure}[h]
\centering
\ifx\du\undefined
  \newlength{\du}
\fi
\setlength{\du}{15\unitlength}
\begin{tikzpicture}
\pgftransformxscale{1.000000}
\pgftransformyscale{-1.000000}
\definecolor{dialinecolor}{rgb}{0.000000, 0.000000, 0.000000}
\pgfsetstrokecolor{dialinecolor}
\definecolor{dialinecolor}{rgb}{1.000000, 1.000000, 1.000000}
\pgfsetfillcolor{dialinecolor}
\pgfsetlinewidth{0.100000\du}
\pgfsetdash{}{0pt}
\pgfsetdash{}{0pt}
\pgfsetbuttcap
{
\definecolor{dialinecolor}{rgb}{0.000000, 0.000000, 0.000000}
\pgfsetfillcolor{dialinecolor}
\pgfsetarrowsend{stealth}
\definecolor{dialinecolor}{rgb}{0.000000, 0.000000, 0.000000}
\pgfsetstrokecolor{dialinecolor}
\draw (5.000000\du,20.000000\du)--(18.000000\du,20.000000\du);
}
\pgfsetlinewidth{0.100000\du}
\pgfsetdash{}{0pt}
\pgfsetdash{}{0pt}
\pgfsetbuttcap
{
\definecolor{dialinecolor}{rgb}{0.000000, 0.000000, 0.000000}
\pgfsetfillcolor{dialinecolor}
\pgfsetarrowsend{stealth}
\definecolor{dialinecolor}{rgb}{0.000000, 0.000000, 0.000000}
\pgfsetstrokecolor{dialinecolor}
\draw (5.000000\du,20.000000\du)--(5.000000\du,10.000000\du);
}
\pgfsetlinewidth{0.0500000\du}
\pgfsetdash{{\pgflinewidth}{0.200000\du}}{0cm}
\pgfsetdash{{\pgflinewidth}{0.200000\du}}{0cm}
\pgfsetbuttcap
{
\definecolor{dialinecolor}{rgb}{0.000000, 0.000000, 0.000000}
\pgfsetfillcolor{dialinecolor}
\definecolor{dialinecolor}{rgb}{0.000000, 0.000000, 0.000000}
\pgfsetstrokecolor{dialinecolor}
\draw (10.500000\du,20.000000\du)--(10.500000\du,10.000000\du);
}
\definecolor{dialinecolor}{rgb}{0.000000, 0.000000, 0.000000}
\pgfsetstrokecolor{dialinecolor}
\node[anchor=west] at (9.500000\du,21.000000\du){$\pi_1(w)$};
\definecolor{dialinecolor}{rgb}{0.000000, 0.000000, 0.000000}
\pgfsetstrokecolor{dialinecolor}
\node[anchor=west] at (6.000000\du,13.000000\du){};
\definecolor{dialinecolor}{rgb}{0.000000, 0.000000, 0.000000}
\pgfsetstrokecolor{dialinecolor}
\node[anchor=west] at (18.000000\du,20.500000\du){$v_1$};
\definecolor{dialinecolor}{rgb}{0.000000, 0.000000, 0.000000}
\pgfsetstrokecolor{dialinecolor}
\node[anchor=west] at (4.000000\du,9.500000\du){$v_2$};
\pgfsetlinewidth{0.0500000\du}
\pgfsetdash{{\pgflinewidth}{0.200000\du}}{0cm}
\pgfsetdash{{\pgflinewidth}{0.200000\du}}{0cm}
\pgfsetbuttcap
{
\definecolor{dialinecolor}{rgb}{0.000000, 0.000000, 0.000000}
\pgfsetfillcolor{dialinecolor}
\definecolor{dialinecolor}{rgb}{0.000000, 0.000000, 0.000000}
\pgfsetstrokecolor{dialinecolor}
\draw (5.000000\du,16.000000\du)--(10.500000\du,16.000000\du);
}
\definecolor{dialinecolor}{rgb}{0.000000, 0.000000, 0.000000}
\pgfsetstrokecolor{dialinecolor}
\node[anchor=west] at (2.500000\du,16.027382\du){$\pi_2(w)$};
\pgfsetlinewidth{0.0500000\du}
\pgfsetdash{{\pgflinewidth}{0.200000\du}}{0cm}
\pgfsetdash{{\pgflinewidth}{0.200000\du}}{0cm}
\pgfsetbuttcap
{
\definecolor{dialinecolor}{rgb}{0.000000, 0.000000, 0.000000}
\pgfsetfillcolor{dialinecolor}
\definecolor{dialinecolor}{rgb}{0.000000, 0.000000, 0.000000}
\pgfsetstrokecolor{dialinecolor}
\draw (10.500000\du,16.000000\du)--(17.500000\du,9.500000\du);
}
\pgfsetlinewidth{0.0500000\du}
\pgfsetdash{{\pgflinewidth}{0.200000\du}}{0cm}
\pgfsetdash{{\pgflinewidth}{0.200000\du}}{0cm}
\pgfsetbuttcap
{
\definecolor{dialinecolor}{rgb}{0.000000, 0.000000, 0.000000}
\pgfsetfillcolor{dialinecolor}
\definecolor{dialinecolor}{rgb}{0.000000, 0.000000, 0.000000}
\pgfsetstrokecolor{dialinecolor}
\draw (10.500000\du,16.000000\du)--(12.500000\du,16.000000\du);
}
\pgfsetlinewidth{0.0500000\du}
\pgfsetdash{{\pgflinewidth}{0.200000\du}}{0cm}
\pgfsetdash{{\pgflinewidth}{0.200000\du}}{0cm}
\pgfsetbuttcap
{
\definecolor{dialinecolor}{rgb}{0.000000, 0.000000, 0.000000}
\pgfsetfillcolor{dialinecolor}
\definecolor{dialinecolor}{rgb}{0.000000, 0.000000, 0.000000}
\pgfsetstrokecolor{dialinecolor}
\pgfpathmoveto{\pgfpoint{11.999976\du}{16.000078\du}}
\pgfpatharc{18}{-70}{0.810575\du and 0.810575\du}
\pgfusepath{stroke}
}
\definecolor{dialinecolor}{rgb}{0.000000, 0.000000, 0.000000}
\pgfsetstrokecolor{dialinecolor}
\node[anchor=west] at (12.000000\du,15.200000\du){$45^o$};
\definecolor{dialinecolor}{rgb}{0.000000, 0.000000, 0.000000}
\pgfsetstrokecolor{dialinecolor}
\node[anchor=west] at (6.500000\du,11.500000\du){$T_L^V$};
\definecolor{dialinecolor}{rgb}{0.000000, 0.000000, 0.000000}
\pgfsetstrokecolor{dialinecolor}
\node[anchor=west] at (12.000000\du,11.500000\du){$T_R^V$};
\definecolor{dialinecolor}{rgb}{0.000000, 0.000000, 0.000000}
\pgfsetstrokecolor{dialinecolor}
\node[anchor=west] at (6.500000\du,18.000000\du){$B_L^V$};
\definecolor{dialinecolor}{rgb}{0.000000, 0.000000, 0.000000}
\pgfsetstrokecolor{dialinecolor}
\node[anchor=west] at (15.500000\du,16.000000\du){$B_R^V$};
\end{tikzpicture}
\caption{Four crucial regions for characterization of allocation functions. All these regions are non-empty.}
\end{figure}
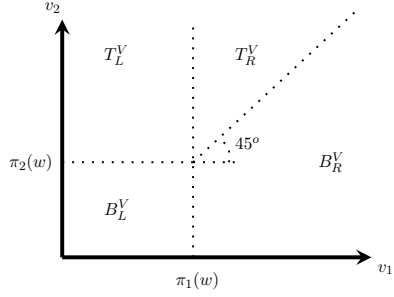
The names are meant to be indicative of which portions of the positive quadrant of bidder $V$'s type space these sets occupy. Similarly, fixing a bid $v \in \R_+^2: v_1 > 0$ of bidder $V$, we define four sets in the type space of bidder $W$: 
$B_L^W(v)$,  $B_R^W(v)$, $T_L^W(v)$ and  $T_R^W(v)$. 
\comment{\begin{enumerate}
\item Bottom-Left-W: $B_L^W(v) = \{w \in \R_{+}^2: w_1 < \phi_1, w_2 < \phi_2\}$ . 
\item Bottom-Right-W: $B_R^W(v) = \{w \in  \R_{+}^2: w_1 > \phi_1, w_2 - \phi_2 < w_1 - \phi_1\}$ .
\item Top-Left-W: $T_L^W(v) = \{w \in \R_{+}^2: w_1 < \phi_1, w_2 > \phi_2\}$ . 
\item Top-Right-W: $T_R^W(v) = \{w \in \R_{+}^2: w_1 > \phi_1, w_2 - \phi_2 > w_1 - \phi_1\}$.
\end{enumerate}}
We drop the arguments  when it is clear from the context, i.e., we say $T_R^V$ instead of $T_R^V(w)$. 
Given that these regions are non-trivial, 
 Lemma \ref{lem:BothAbove45Deg} shows that $V$ gets both items in $T_R^V$, 
 Lemma \ref{lem:TallRectangle} shows that $\pi_2 \geq \pi_1$ and 
 Lemma \ref{lem:1Below45Deg} shows that she { \em doesn't} get item 2 
 in $B_R^V$ when $\pi_2 > \pi_1$. 
\begin{lemma}
\label{lem:BothAbove45Deg}
For every deterministic DSIC+IR mechanism that gets a finite approximation factor,
$$\forall w \in \R_+^2: w_1 > 0, \forall v \in T_R^V \text{ we have } \{\allocv_1(v,w) = \allocv_2(v,w) = 1,\ \pricev(v,w) = \pi_2\}$$
\begin{figure}[h]
\centering
\ifx\du\undefined
  \newlength{\du}
\fi
\setlength{\du}{15\unitlength}
\begin{tikzpicture}
\pgftransformxscale{1.000000}
\pgftransformyscale{-1.000000}
\definecolor{dialinecolor}{rgb}{0.000000, 0.000000, 0.000000}
\pgfsetstrokecolor{dialinecolor}
\definecolor{dialinecolor}{rgb}{1.000000, 1.000000, 1.000000}
\pgfsetfillcolor{dialinecolor}
\pgfsetlinewidth{0.100000\du}
\pgfsetdash{}{0pt}
\pgfsetdash{}{0pt}
\pgfsetbuttcap
{
\definecolor{dialinecolor}{rgb}{0.000000, 0.000000, 0.000000}
\pgfsetfillcolor{dialinecolor}
\pgfsetarrowsend{stealth}
\definecolor{dialinecolor}{rgb}{0.000000, 0.000000, 0.000000}
\pgfsetstrokecolor{dialinecolor}
\draw (5.000000\du,20.000000\du)--(18.000000\du,20.000000\du);
}
\pgfsetlinewidth{0.100000\du}
\pgfsetdash{}{0pt}
\pgfsetdash{}{0pt}
\pgfsetbuttcap
{
\definecolor{dialinecolor}{rgb}{0.000000, 0.000000, 0.000000}
\pgfsetfillcolor{dialinecolor}
\pgfsetarrowsend{stealth}
\definecolor{dialinecolor}{rgb}{0.000000, 0.000000, 0.000000}
\pgfsetstrokecolor{dialinecolor}
\draw (5.000000\du,20.000000\du)--(5.000000\du,10.000000\du);
}
\pgfsetlinewidth{0.0500000\du}
\pgfsetdash{{\pgflinewidth}{0.200000\du}}{0cm}
\pgfsetdash{{\pgflinewidth}{0.200000\du}}{0cm}
\pgfsetbuttcap
{
\definecolor{dialinecolor}{rgb}{0.000000, 0.000000, 0.000000}
\pgfsetfillcolor{dialinecolor}
\definecolor{dialinecolor}{rgb}{0.000000, 0.000000, 0.000000}
\pgfsetstrokecolor{dialinecolor}
\draw (10.600000\du,20.000000\du)--(10.600000\du,10.000000\du);
}
\definecolor{dialinecolor}{rgb}{0.000000, 0.000000, 0.000000}
\pgfsetstrokecolor{dialinecolor}
\node[anchor=west] at (9.500000\du,21.000000\du){$\pi_1(w)$};
\definecolor{dialinecolor}{rgb}{0.000000, 0.000000, 0.000000}
\pgfsetstrokecolor{dialinecolor}
\node[anchor=west] at (16.200000\du,13.600000\du){$x=$\{1\}};
\definecolor{dialinecolor}{rgb}{0.000000, 0.000000, 0.000000}
\pgfsetstrokecolor{dialinecolor}
\node[anchor=west] at (6.000000\du,13.000000\du){};
\definecolor{dialinecolor}{rgb}{0.000000, 0.000000, 0.000000}
\pgfsetstrokecolor{dialinecolor}
\node[anchor=west] at (5.500000\du,14.000000\du){$x=$\{2\}};
\definecolor{dialinecolor}{rgb}{0.000000, 0.000000, 0.000000}
\pgfsetstrokecolor{dialinecolor}
\node[anchor=west] at (18.000000\du,20.500000\du){$v_1$};
\definecolor{dialinecolor}{rgb}{0.000000, 0.000000, 0.000000}
\pgfsetstrokecolor{dialinecolor}
\node[anchor=west] at (4.000000\du,9.500000\du){$v_2$};
\pgfsetlinewidth{0.0500000\du}
\pgfsetdash{{\pgflinewidth}{0.200000\du}}{0cm}
\pgfsetdash{{\pgflinewidth}{0.200000\du}}{0cm}
\pgfsetbuttcap
{
\definecolor{dialinecolor}{rgb}{0.000000, 0.000000, 0.000000}
\pgfsetfillcolor{dialinecolor}
\definecolor{dialinecolor}{rgb}{0.000000, 0.000000, 0.000000}
\pgfsetstrokecolor{dialinecolor}
\draw (5.000000\du,16.000000\du)--(10.500000\du,16.000000\du);
}
\definecolor{dialinecolor}{rgb}{0.000000, 0.000000, 0.000000}
\pgfsetstrokecolor{dialinecolor}
\node[anchor=west] at (5.500000\du,18.000000\du){$x=$\{$\emptyset$\}};
\definecolor{dialinecolor}{rgb}{0.000000, 0.000000, 0.000000}
\pgfsetstrokecolor{dialinecolor}
\node[anchor=west] at (5.500000\du,13.000000\du){$p=\pi_2(w)$};
\definecolor{dialinecolor}{rgb}{0.000000, 0.000000, 0.000000}
\pgfsetstrokecolor{dialinecolor}
\node[anchor=west] at (2.800000\du,16.002800\du){$\pi_2(w)$};
\pgfsetlinewidth{0.0500000\du}
\pgfsetdash{{\pgflinewidth}{0.200000\du}}{0cm}
\pgfsetdash{{\pgflinewidth}{0.200000\du}}{0cm}
\pgfsetmiterjoin
\pgfsetbuttcap
\definecolor{dialinecolor}{rgb}{1.000000, 1.000000, 1.000000}
\pgfsetfillcolor{dialinecolor}
\pgfpathmoveto{\pgfpoint{19.400000\du}{14.600000\du}}
\pgfpathcurveto{\pgfpoint{20.400000\du}{14.600000\du}}{\pgfpoint{19.400000\du}{19.100000\du}}{\pgfpoint{16.900000\du}{16.600000\du}}
\pgfpathcurveto{\pgfpoint{14.400000\du}{14.100000\du}}{\pgfpoint{18.400000\du}{14.600000\du}}{\pgfpoint{19.400000\du}{14.600000\du}}
\pgfusepath{fill}
\definecolor{dialinecolor}{rgb}{0.000000, 0.000000, 0.000000}
\pgfsetstrokecolor{dialinecolor}
\pgfpathmoveto{\pgfpoint{19.400000\du}{14.600000\du}}
\pgfpathcurveto{\pgfpoint{20.400000\du}{14.600000\du}}{\pgfpoint{19.400000\du}{19.100000\du}}{\pgfpoint{16.900000\du}{16.600000\du}}
\pgfpathcurveto{\pgfpoint{14.400000\du}{14.100000\du}}{\pgfpoint{18.400000\du}{14.600000\du}}{\pgfpoint{19.400000\du}{14.600000\du}}
\pgfusepath{stroke}
\pgfsetlinewidth{0.0500000\du}
\pgfsetdash{{\pgflinewidth}{0.200000\du}}{0cm}
\pgfsetdash{{\pgflinewidth}{0.200000\du}}{0cm}
\pgfsetmiterjoin
\pgfsetbuttcap
\definecolor{dialinecolor}{rgb}{1.000000, 1.000000, 1.000000}
\pgfsetfillcolor{dialinecolor}
\pgfpathmoveto{\pgfpoint{16.000000\du}{17.000000\du}}
\pgfpathcurveto{\pgfpoint{17.000000\du}{17.000000\du}}{\pgfpoint{15.000000\du}{21.000000\du}}{\pgfpoint{12.500000\du}{18.500000\du}}
\pgfpathcurveto{\pgfpoint{10.000000\du}{16.000000\du}}{\pgfpoint{15.000000\du}{17.000000\du}}{\pgfpoint{16.000000\du}{17.000000\du}}
\pgfusepath{fill}
\definecolor{dialinecolor}{rgb}{0.000000, 0.000000, 0.000000}
\pgfsetstrokecolor{dialinecolor}
\pgfpathmoveto{\pgfpoint{16.000000\du}{17.000000\du}}
\pgfpathcurveto{\pgfpoint{17.000000\du}{17.000000\du}}{\pgfpoint{15.000000\du}{21.000000\du}}{\pgfpoint{12.500000\du}{18.500000\du}}
\pgfpathcurveto{\pgfpoint{10.000000\du}{16.000000\du}}{\pgfpoint{15.000000\du}{17.000000\du}}{\pgfpoint{16.000000\du}{17.000000\du}}
\pgfusepath{stroke}
\definecolor{dialinecolor}{rgb}{0.000000, 0.000000, 0.000000}
\pgfsetstrokecolor{dialinecolor}
\node[anchor=west] at (16.200000\du,12.600000\du){$p=\pi_1(w)$};
\definecolor{dialinecolor}{rgb}{0.000000, 0.000000, 0.000000}
\pgfsetstrokecolor{dialinecolor}
\node[anchor=west] at (11.627686\du,11.200481\du){$x=$\{1,2\}};
\definecolor{dialinecolor}{rgb}{0.000000, 0.000000, 0.000000}
\pgfsetstrokecolor{dialinecolor}
\node[anchor=west] at (16.400000\du,15.600000\du){$x=$\{1,2\}};
\definecolor{dialinecolor}{rgb}{0.000000, 0.000000, 0.000000}
\pgfsetstrokecolor{dialinecolor}
\node[anchor=west] at (12.700000\du,18.000000\du){$x=$\{1,2\}};
\pgfsetlinewidth{0.0500000\du}
\pgfsetdash{{\pgflinewidth}{0.200000\du}}{0cm}
\pgfsetdash{{\pgflinewidth}{0.200000\du}}{0cm}
\pgfsetbuttcap
{
\definecolor{dialinecolor}{rgb}{0.000000, 0.000000, 0.000000}
\pgfsetfillcolor{dialinecolor}
\definecolor{dialinecolor}{rgb}{0.000000, 0.000000, 0.000000}
\pgfsetstrokecolor{dialinecolor}
\draw (10.600000\du,16.000000\du)--(17.000000\du,10.400000\du);
}
\pgfsetlinewidth{0.0500000\du}
\pgfsetdash{{\pgflinewidth}{0.200000\du}}{0cm}
\pgfsetdash{{\pgflinewidth}{0.200000\du}}{0cm}
\pgfsetbuttcap
{
\definecolor{dialinecolor}{rgb}{0.000000, 0.000000, 0.000000}
\pgfsetfillcolor{dialinecolor}
\definecolor{dialinecolor}{rgb}{0.000000, 0.000000, 0.000000}
\pgfsetstrokecolor{dialinecolor}
\draw (10.600000\du,16.000000\du)--(12.800000\du,16.000000\du);
}
\definecolor{dialinecolor}{rgb}{0.000000, 0.000000, 0.000000}
\pgfsetstrokecolor{dialinecolor}
\node[anchor=west] at (12.100000\du,15.20000\du){$45^o$};
\definecolor{dialinecolor}{rgb}{0.000000, 0.000000, 0.000000}
\pgfsetstrokecolor{dialinecolor}
\node[anchor=west] at (11.600000\du,10.200000\du){$p=\pi_2(w)$};
\end{tikzpicture}
\caption{Characterization of the allocation and payment function in the four crucial regions. All values in $T_R^V$ are allocated both items and charged a price of $\pi_2(w)$.}
\end{figure}
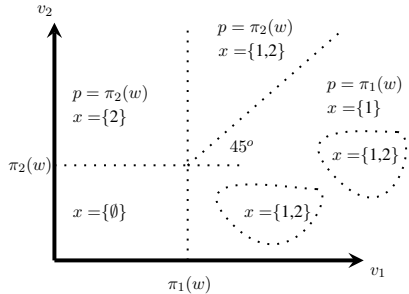
\end{lemma}
\begin{proof}
	To prove the allocation part, fix a $w$ with $w_1 > 0$, and consider a $v \in T_R^V$. For any such $v$, by Lemma~\ref{lem:VerticalLine} $\allocv_1(v,w) = 1$. Suppose, $\allocv_2(v,w) = 0$. Then bidder $V$ with value $v \in T_R^V$, has an incentive to report a $v \in T_L^V$. By Lemma~\ref{lem:HorizontalLine}, such a report earns the bidder a utility of $v_2 - \pi_2$ which, by the definition of $T_R^V$, is strictly larger than his current utility of  $v_1 - \pi_1$. To see that his current utility is $v_1 - \pi_1$ note that when a bidder gets only item $1$, by Lemma~\ref{lem:1sPayPi1} he has to pay $\pi_1$.  
\begin{figure}[h]
\centering
\begin{subfigure}[b]{0.45\textwidth}
\centering
\input{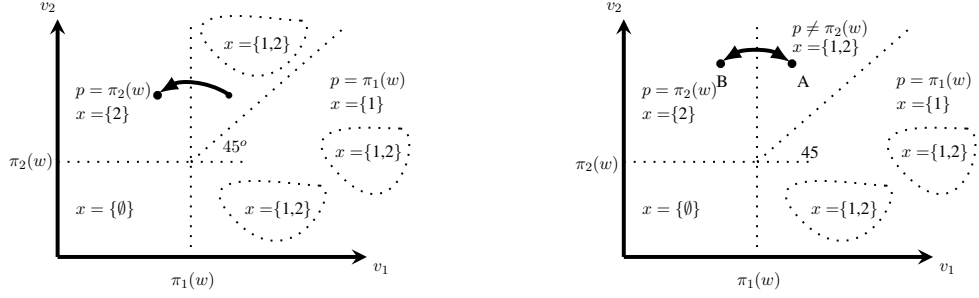}
\end{subfigure}
\begin{subfigure}[b]{0.45\textwidth}
\centering
\ifx\du\undefined
  \newlength{\du}
\fi
\setlength{\du}{15\unitlength}
\begin{tikzpicture}
\pgftransformxscale{1.000000}
\pgftransformyscale{-1.000000}
\definecolor{dialinecolor}{rgb}{0.000000, 0.000000, 0.000000}
\pgfsetstrokecolor{dialinecolor}
\definecolor{dialinecolor}{rgb}{1.000000, 1.000000, 1.000000}
\pgfsetfillcolor{dialinecolor}
\pgfsetlinewidth{0.100000\du}
\pgfsetdash{}{0pt}
\pgfsetdash{}{0pt}
\pgfsetbuttcap
{
\definecolor{dialinecolor}{rgb}{0.000000, 0.000000, 0.000000}
\pgfsetfillcolor{dialinecolor}
\pgfsetarrowsend{stealth}
\definecolor{dialinecolor}{rgb}{0.000000, 0.000000, 0.000000}
\pgfsetstrokecolor{dialinecolor}
\draw (5.000000\du,20.000000\du)--(18.000000\du,20.000000\du);
}
\pgfsetlinewidth{0.100000\du}
\pgfsetdash{}{0pt}
\pgfsetdash{}{0pt}
\pgfsetbuttcap
{
\definecolor{dialinecolor}{rgb}{0.000000, 0.000000, 0.000000}
\pgfsetfillcolor{dialinecolor}
\pgfsetarrowsend{stealth}
\definecolor{dialinecolor}{rgb}{0.000000, 0.000000, 0.000000}
\pgfsetstrokecolor{dialinecolor}
\draw (5.000000\du,20.000000\du)--(5.000000\du,10.000000\du);
}
\pgfsetlinewidth{0.0500000\du}
\pgfsetdash{{\pgflinewidth}{0.200000\du}}{0cm}
\pgfsetdash{{\pgflinewidth}{0.200000\du}}{0cm}
\pgfsetbuttcap
{
\definecolor{dialinecolor}{rgb}{0.000000, 0.000000, 0.000000}
\pgfsetfillcolor{dialinecolor}
\definecolor{dialinecolor}{rgb}{0.000000, 0.000000, 0.000000}
\pgfsetstrokecolor{dialinecolor}
\draw (10.600000\du,20.000000\du)--(10.600000\du,10.000000\du);
}
\definecolor{dialinecolor}{rgb}{0.000000, 0.000000, 0.000000}
\pgfsetstrokecolor{dialinecolor}
\node[anchor=west] at (9.500000\du,21.000000\du){$\pi_1(w)$};
\definecolor{dialinecolor}{rgb}{0.000000, 0.000000, 0.000000}
\pgfsetstrokecolor{dialinecolor}
\node[anchor=west] at (16.200000\du,13.600000\du){$x=$\{1\}};
\definecolor{dialinecolor}{rgb}{0.000000, 0.000000, 0.000000}
\pgfsetstrokecolor{dialinecolor}
\node[anchor=west] at (6.000000\du,13.000000\du){};
\definecolor{dialinecolor}{rgb}{0.000000, 0.000000, 0.000000}
\pgfsetstrokecolor{dialinecolor}
\node[anchor=west] at (5.500000\du,14.000000\du){$x=$\{2\}};
\definecolor{dialinecolor}{rgb}{0.000000, 0.000000, 0.000000}
\pgfsetstrokecolor{dialinecolor}
\node[anchor=west] at (18.000000\du,20.500000\du){$v_1$};
\definecolor{dialinecolor}{rgb}{0.000000, 0.000000, 0.000000}
\pgfsetstrokecolor{dialinecolor}
\node[anchor=west] at (4.000000\du,9.500000\du){$v_2$};
\pgfsetlinewidth{0.0500000\du}
\pgfsetdash{{\pgflinewidth}{0.200000\du}}{0cm}
\pgfsetdash{{\pgflinewidth}{0.200000\du}}{0cm}
\pgfsetbuttcap
{
\definecolor{dialinecolor}{rgb}{0.000000, 0.000000, 0.000000}
\pgfsetfillcolor{dialinecolor}
\definecolor{dialinecolor}{rgb}{0.000000, 0.000000, 0.000000}
\pgfsetstrokecolor{dialinecolor}
\draw (5.000000\du,16.000000\du)--(10.500000\du,16.000000\du);
}
\definecolor{dialinecolor}{rgb}{0.000000, 0.000000, 0.000000}
\pgfsetstrokecolor{dialinecolor}
\node[anchor=west] at (5.500000\du,18.000000\du){$x=\{\emptyset\}$};
\definecolor{dialinecolor}{rgb}{0.000000, 0.000000, 0.000000}
\pgfsetstrokecolor{dialinecolor}
\node[anchor=west] at (5.500000\du,13.000000\du){$p=\pi_2(w)$};
\definecolor{dialinecolor}{rgb}{0.000000, 0.000000, 0.000000}
\pgfsetstrokecolor{dialinecolor}
\node[anchor=west] at (2.700000\du,16.102800\du){$\pi_2(w)$};
\pgfsetlinewidth{0.0500000\du}
\pgfsetdash{{\pgflinewidth}{0.200000\du}}{0cm}
\pgfsetdash{{\pgflinewidth}{0.200000\du}}{0cm}
\pgfsetmiterjoin
\pgfsetbuttcap
\definecolor{dialinecolor}{rgb}{1.000000, 1.000000, 1.000000}
\pgfsetfillcolor{dialinecolor}
\pgfpathmoveto{\pgfpoint{19.400000\du}{14.600000\du}}
\pgfpathcurveto{\pgfpoint{20.400000\du}{14.600000\du}}{\pgfpoint{19.400000\du}{19.100000\du}}{\pgfpoint{16.900000\du}{16.600000\du}}
\pgfpathcurveto{\pgfpoint{14.400000\du}{14.100000\du}}{\pgfpoint{18.400000\du}{14.600000\du}}{\pgfpoint{19.400000\du}{14.600000\du}}
\pgfusepath{fill}
\definecolor{dialinecolor}{rgb}{0.000000, 0.000000, 0.000000}
\pgfsetstrokecolor{dialinecolor}
\pgfpathmoveto{\pgfpoint{19.400000\du}{14.600000\du}}
\pgfpathcurveto{\pgfpoint{20.400000\du}{14.600000\du}}{\pgfpoint{19.400000\du}{19.100000\du}}{\pgfpoint{16.900000\du}{16.600000\du}}
\pgfpathcurveto{\pgfpoint{14.400000\du}{14.100000\du}}{\pgfpoint{18.400000\du}{14.600000\du}}{\pgfpoint{19.400000\du}{14.600000\du}}
\pgfusepath{stroke}
\pgfsetlinewidth{0.0500000\du}
\pgfsetdash{{\pgflinewidth}{0.200000\du}}{0cm}
\pgfsetdash{{\pgflinewidth}{0.200000\du}}{0cm}
\pgfsetmiterjoin
\pgfsetbuttcap
\definecolor{dialinecolor}{rgb}{1.000000, 1.000000, 1.000000}
\pgfsetfillcolor{dialinecolor}
\pgfpathmoveto{\pgfpoint{16.000000\du}{17.000000\du}}
\pgfpathcurveto{\pgfpoint{17.000000\du}{17.000000\du}}{\pgfpoint{15.000000\du}{21.000000\du}}{\pgfpoint{12.500000\du}{18.500000\du}}
\pgfpathcurveto{\pgfpoint{10.000000\du}{16.000000\du}}{\pgfpoint{15.000000\du}{17.000000\du}}{\pgfpoint{16.000000\du}{17.000000\du}}
\pgfusepath{fill}
\definecolor{dialinecolor}{rgb}{0.000000, 0.000000, 0.000000}
\pgfsetstrokecolor{dialinecolor}
\pgfpathmoveto{\pgfpoint{16.000000\du}{17.000000\du}}
\pgfpathcurveto{\pgfpoint{17.000000\du}{17.000000\du}}{\pgfpoint{15.000000\du}{21.000000\du}}{\pgfpoint{12.500000\du}{18.500000\du}}
\pgfpathcurveto{\pgfpoint{10.000000\du}{16.000000\du}}{\pgfpoint{15.000000\du}{17.000000\du}}{\pgfpoint{16.000000\du}{17.000000\du}}
\pgfusepath{stroke}
\definecolor{dialinecolor}{rgb}{0.000000, 0.000000, 0.000000}
\pgfsetstrokecolor{dialinecolor}
\node[anchor=west] at (16.200000\du,12.600000\du){$p=\pi_1(w)$};
\definecolor{dialinecolor}{rgb}{0.000000, 0.000000, 0.000000}
\pgfsetstrokecolor{dialinecolor}
\node[anchor=west] at (16.400000\du,15.600000\du){$x=$\{1,2\}};
\definecolor{dialinecolor}{rgb}{0.000000, 0.000000, 0.000000}
\pgfsetstrokecolor{dialinecolor}
\node[anchor=west] at (12.600000\du,18.000000\du){$x=$\{1,2\}};
\pgfsetlinewidth{0.0500000\du}
\pgfsetdash{{\pgflinewidth}{0.200000\du}}{0cm}
\pgfsetdash{{\pgflinewidth}{0.200000\du}}{0cm}
\pgfsetbuttcap
{
\definecolor{dialinecolor}{rgb}{0.000000, 0.000000, 0.000000}
\pgfsetfillcolor{dialinecolor}
\definecolor{dialinecolor}{rgb}{0.000000, 0.000000, 0.000000}
\pgfsetstrokecolor{dialinecolor}
\draw (10.600000\du,16.000000\du)--(17.000000\du,10.400000\du);
}
\pgfsetlinewidth{0.0500000\du}
\pgfsetdash{{\pgflinewidth}{0.200000\du}}{0cm}
\pgfsetdash{{\pgflinewidth}{0.200000\du}}{0cm}
\pgfsetbuttcap
{
\definecolor{dialinecolor}{rgb}{0.000000, 0.000000, 0.000000}
\pgfsetfillcolor{dialinecolor}
\definecolor{dialinecolor}{rgb}{0.000000, 0.000000, 0.000000}
\pgfsetstrokecolor{dialinecolor}
\draw (10.600000\du,16.000000\du)--(12.800000\du,16.000000\du);
}
\definecolor{dialinecolor}{rgb}{0.000000, 0.000000, 0.000000}
\pgfsetstrokecolor{dialinecolor}
\node[anchor=west] at (12.200000\du,15.600000\du){45};
\definecolor{dialinecolor}{rgb}{0.000000, 0.000000, 0.000000}
\pgfsetfillcolor{dialinecolor}
\pgfpathellipse{\pgfpoint{12.070700\du}{11.870700\du}}{\pgfpoint{0.100000\du}{0\du}}{\pgfpoint{0\du}{0.100000\du}}
\pgfusepath{fill}
\pgfsetlinewidth{0.100000\du}
\pgfsetdash{}{0pt}
\pgfsetdash{}{0pt}
\definecolor{dialinecolor}{rgb}{0.000000, 0.000000, 0.000000}
\pgfsetstrokecolor{dialinecolor}
\pgfpathellipse{\pgfpoint{12.070700\du}{11.870700\du}}{\pgfpoint{0.100000\du}{0\du}}{\pgfpoint{0\du}{0.100000\du}}
\pgfusepath{stroke}
\definecolor{dialinecolor}{rgb}{0.000000, 0.000000, 0.000000}
\pgfsetfillcolor{dialinecolor}
\pgfpathellipse{\pgfpoint{9.070710\du}{11.870700\du}}{\pgfpoint{0.100000\du}{0\du}}{\pgfpoint{0\du}{0.100000\du}}
\pgfusepath{fill}
\pgfsetlinewidth{0.100000\du}
\pgfsetdash{}{0pt}
\pgfsetdash{}{0pt}
\definecolor{dialinecolor}{rgb}{0.000000, 0.000000, 0.000000}
\pgfsetstrokecolor{dialinecolor}
\pgfpathellipse{\pgfpoint{9.070710\du}{11.870700\du}}{\pgfpoint{0.100000\du}{0\du}}{\pgfpoint{0\du}{0.100000\du}}
\pgfusepath{stroke}
\pgfsetlinewidth{0.100000\du}
\pgfsetdash{}{0pt}
\pgfsetdash{}{0pt}
\pgfsetbuttcap
{
\definecolor{dialinecolor}{rgb}{0.000000, 0.000000, 0.000000}
\pgfsetfillcolor{dialinecolor}
\pgfsetarrowsstart{latex}
\pgfsetarrowsend{latex}
\definecolor{dialinecolor}{rgb}{0.000000, 0.000000, 0.000000}
\pgfsetstrokecolor{dialinecolor}
\pgfpathmoveto{\pgfpoint{12.000040\du}{11.800025\du}}
\pgfpatharc{302}{239}{2.753587\du and 2.753587\du}
\pgfusepath{stroke}
}
\definecolor{dialinecolor}{rgb}{0.000000, 0.000000, 0.000000}
\pgfsetstrokecolor{dialinecolor}
\node[anchor=west] at (11.800000\du,11.200000\du){$x=$\{1,2\}};
\definecolor{dialinecolor}{rgb}{0.000000, 0.000000, 0.000000}
\pgfsetstrokecolor{dialinecolor}
\node[anchor=west] at (11.800000\du,10.400000\du){$p\neq \pi_2(w)$};
\definecolor{dialinecolor}{rgb}{0.000000, 0.000000, 0.000000}
\pgfsetstrokecolor{dialinecolor}
\node[anchor=west] at (12.000000\du,12.600000\du){A};
\definecolor{dialinecolor}{rgb}{0.000000, 0.000000, 0.000000}
\pgfsetstrokecolor{dialinecolor}
\node[anchor=west] at (8.600000\du,12.600000\du){B};
\end{tikzpicture}
\end{subfigure}
\caption{Profitable deviations if both items are not allocated at some point above the $45^o$ line (left) and if when both items are allocated above
the $45^o$ line the payment is not equal to $\pi_2(w)$ (right).}\label{fig:deviations1}
\end{figure}
	To prove the payment part, fix a $w$ and consider a $v \in T_R^V$. Since we just showed that all $v\in T_R^V$ have $\allocv_1(v,w) = \allocv_2(v,w) = 1$, it follows from DSIC that $\pricev(v,w)$ for all $v\in T_R^V$ must be the same, say $p$. We know from Lemma~\ref{lem:HorizontalLine} that for all $v \in T_L^V$ we have $\pricev(v,w) = \pi_2$. We claim that this immediately implies that $p = \pi_2$. Suppose on the contrary $p \neq \pi_2$. 
	\begin{enumerate}
		\item Say $p > \pi_2$. Then $v$'s in $T_R^V$ with $v_2 \geq v_1$ (e.g. point $A$ in Figure \ref{fig:deviations1}) have an incentive to report in $T_L^V$ (e.g. point $B$ in Figure \ref{fig:deviations1}) to receive a utility $v_2 - \pi_2$ which is strictly larger than the utility $v_2 - p$ they get currently. 
		\item And if $p < \pi_2$, then $v$'s in $T_L^V$ with $v_2 \geq v_1$ (e.g. point $B$ in Figure \ref{fig:deviations1}) have an incentive to report to a $v$ in $T_R^V$ (e.g. point $A$ in Figure \ref{fig:deviations1}) to a get a utility of $v_2 - p$ which is strictly larger than the utility $v_2 - \pi_2$ they get currently. 
	\end{enumerate} 
Therefore we have $\pricev(v,w) = \pi_2$ for all $v \in T_R^V$.
\end{proof}

\begin{lemma}
\label{lem:TallRectangle}
For every deterministic DSIC+IR mechanism that gets a finite approximation factor, 
$\forall w \in \R_+^2: w_1 > 0,\text{ we have }\pi_2(w) \geq \pi_1(w).$ 
\begin{figure}[h]
\centering
\begin{subfigure}[b]{0.45\textwidth}
\centering
\ifx\du\undefined
  \newlength{\du}
\fi
\setlength{\du}{15\unitlength}
\begin{tikzpicture}
\pgftransformxscale{1.000000}
\pgftransformyscale{-1.000000}
\definecolor{dialinecolor}{rgb}{0.000000, 0.000000, 0.000000}
\pgfsetstrokecolor{dialinecolor}
\definecolor{dialinecolor}{rgb}{1.000000, 1.000000, 1.000000}
\pgfsetfillcolor{dialinecolor}
\pgfsetlinewidth{0.100000\du}
\pgfsetdash{}{0pt}
\pgfsetdash{}{0pt}
\pgfsetbuttcap
{
\definecolor{dialinecolor}{rgb}{0.000000, 0.000000, 0.000000}
\pgfsetfillcolor{dialinecolor}
\pgfsetarrowsend{stealth}
\definecolor{dialinecolor}{rgb}{0.000000, 0.000000, 0.000000}
\pgfsetstrokecolor{dialinecolor}
\draw (5.000000\du,20.000000\du)--(18.000000\du,20.000000\du);
}
\pgfsetlinewidth{0.100000\du}
\pgfsetdash{}{0pt}
\pgfsetdash{}{0pt}
\pgfsetbuttcap
{
\definecolor{dialinecolor}{rgb}{0.000000, 0.000000, 0.000000}
\pgfsetfillcolor{dialinecolor}
\pgfsetarrowsend{stealth}
\definecolor{dialinecolor}{rgb}{0.000000, 0.000000, 0.000000}
\pgfsetstrokecolor{dialinecolor}
\draw (5.000000\du,20.000000\du)--(5.000000\du,10.000000\du);
}
\pgfsetlinewidth{0.0500000\du}
\pgfsetdash{{\pgflinewidth}{0.200000\du}}{0cm}
\pgfsetdash{{\pgflinewidth}{0.200000\du}}{0cm}
\pgfsetbuttcap
{
\definecolor{dialinecolor}{rgb}{0.000000, 0.000000, 0.000000}
\pgfsetfillcolor{dialinecolor}
\definecolor{dialinecolor}{rgb}{0.000000, 0.000000, 0.000000}
\pgfsetstrokecolor{dialinecolor}
\draw (8.400000\du,19.800000\du)--(8.400000\du,9.800000\du);
}
\definecolor{dialinecolor}{rgb}{0.000000, 0.000000, 0.000000}
\pgfsetstrokecolor{dialinecolor}
\node[anchor=west] at (7.7000000\du,20.800000\du){$\pi_1(w)$};
\definecolor{dialinecolor}{rgb}{0.000000, 0.000000, 0.000000}
\pgfsetstrokecolor{dialinecolor}
\node[anchor=west] at (12.600000\du,14.000000\du){$x=$\{1\}};
\definecolor{dialinecolor}{rgb}{0.000000, 0.000000, 0.000000}
\pgfsetstrokecolor{dialinecolor}
\node[anchor=west] at (6.000000\du,13.000000\du){};
\definecolor{dialinecolor}{rgb}{0.000000, 0.000000, 0.000000}
\pgfsetstrokecolor{dialinecolor}
\node[anchor=west] at (5.100000\du,14.000000\du){$x=$\{2\}};
\definecolor{dialinecolor}{rgb}{0.000000, 0.000000, 0.000000}
\pgfsetstrokecolor{dialinecolor}
\node[anchor=west] at (18.000000\du,20.500000\du){$v_1$};
\definecolor{dialinecolor}{rgb}{0.000000, 0.000000, 0.000000}
\pgfsetstrokecolor{dialinecolor}
\node[anchor=west] at (4.000000\du,9.500000\du){$v_2$};
\pgfsetlinewidth{0.0500000\du}
\pgfsetdash{{\pgflinewidth}{0.200000\du}}{0cm}
\pgfsetdash{{\pgflinewidth}{0.200000\du}}{0cm}
\pgfsetbuttcap
{
\definecolor{dialinecolor}{rgb}{0.000000, 0.000000, 0.000000}
\pgfsetfillcolor{dialinecolor}
\definecolor{dialinecolor}{rgb}{0.000000, 0.000000, 0.000000}
\pgfsetstrokecolor{dialinecolor}
\draw (5.000000\du,14.600000\du)--(8.300000\du,14.600000\du);
}
\definecolor{dialinecolor}{rgb}{0.000000, 0.000000, 0.000000}
\pgfsetstrokecolor{dialinecolor}
\node[anchor=west] at (5.500000\du,18.000000\du){$x=\{\emptyset\}$};
\definecolor{dialinecolor}{rgb}{0.000000, 0.000000, 0.000000}
\pgfsetstrokecolor{dialinecolor}
\node[anchor=west] at (5.100000\du,13.000000\du){$p=\pi_2(w)$};
\definecolor{dialinecolor}{rgb}{0.000000, 0.000000, 0.000000}
\pgfsetstrokecolor{dialinecolor}
\node[anchor=west] at (2.700000\du,14.600000\du){$\pi_2(w)$};
\pgfsetlinewidth{0.0500000\du}
\pgfsetdash{{\pgflinewidth}{0.200000\du}}{0cm}
\pgfsetdash{{\pgflinewidth}{0.200000\du}}{0cm}
\pgfsetmiterjoin
\pgfsetbuttcap
\definecolor{dialinecolor}{rgb}{1.000000, 1.000000, 1.000000}
\pgfsetfillcolor{dialinecolor}
\pgfpathmoveto{\pgfpoint{19.400000\du}{14.600000\du}}
\pgfpathcurveto{\pgfpoint{20.400000\du}{14.600000\du}}{\pgfpoint{19.400000\du}{19.100000\du}}{\pgfpoint{16.900000\du}{16.600000\du}}
\pgfpathcurveto{\pgfpoint{14.400000\du}{14.100000\du}}{\pgfpoint{18.400000\du}{14.600000\du}}{\pgfpoint{19.400000\du}{14.600000\du}}
\pgfusepath{fill}
\definecolor{dialinecolor}{rgb}{0.000000, 0.000000, 0.000000}
\pgfsetstrokecolor{dialinecolor}
\pgfpathmoveto{\pgfpoint{19.400000\du}{14.600000\du}}
\pgfpathcurveto{\pgfpoint{20.400000\du}{14.600000\du}}{\pgfpoint{19.400000\du}{19.100000\du}}{\pgfpoint{16.900000\du}{16.600000\du}}
\pgfpathcurveto{\pgfpoint{14.400000\du}{14.100000\du}}{\pgfpoint{18.400000\du}{14.600000\du}}{\pgfpoint{19.400000\du}{14.600000\du}}
\pgfusepath{stroke}
\pgfsetlinewidth{0.0500000\du}
\pgfsetdash{{\pgflinewidth}{0.200000\du}}{0cm}
\pgfsetdash{{\pgflinewidth}{0.200000\du}}{0cm}
\pgfsetmiterjoin
\pgfsetbuttcap
\definecolor{dialinecolor}{rgb}{1.000000, 1.000000, 1.000000}
\pgfsetfillcolor{dialinecolor}
\pgfpathmoveto{\pgfpoint{16.000000\du}{17.000000\du}}
\pgfpathcurveto{\pgfpoint{17.000000\du}{17.000000\du}}{\pgfpoint{15.000000\du}{21.000000\du}}{\pgfpoint{12.500000\du}{18.500000\du}}
\pgfpathcurveto{\pgfpoint{10.000000\du}{16.000000\du}}{\pgfpoint{15.000000\du}{17.000000\du}}{\pgfpoint{16.000000\du}{17.000000\du}}
\pgfusepath{fill}
\definecolor{dialinecolor}{rgb}{0.000000, 0.000000, 0.000000}
\pgfsetstrokecolor{dialinecolor}
\pgfpathmoveto{\pgfpoint{16.000000\du}{17.000000\du}}
\pgfpathcurveto{\pgfpoint{17.000000\du}{17.000000\du}}{\pgfpoint{15.000000\du}{21.000000\du}}{\pgfpoint{12.500000\du}{18.500000\du}}
\pgfpathcurveto{\pgfpoint{10.000000\du}{16.000000\du}}{\pgfpoint{15.000000\du}{17.000000\du}}{\pgfpoint{16.000000\du}{17.000000\du}}
\pgfusepath{stroke}
\definecolor{dialinecolor}{rgb}{0.000000, 0.000000, 0.000000}
\pgfsetstrokecolor{dialinecolor}
\node[anchor=west] at (12.600000\du,13.000000\du){$p=\pi_1(w)$};
\definecolor{dialinecolor}{rgb}{0.000000, 0.000000, 0.000000}
\pgfsetstrokecolor{dialinecolor}
\node[anchor=west] at (16.400000\du,15.600000\du){$x=$\{1,2\}};
\definecolor{dialinecolor}{rgb}{0.000000, 0.000000, 0.000000}
\pgfsetstrokecolor{dialinecolor}
\node[anchor=west] at (12.500000\du,18.000000\du){$x=$\{1,2\}};
\pgfsetlinewidth{0.0500000\du}
\pgfsetdash{{\pgflinewidth}{0.200000\du}}{0cm}
\pgfsetdash{{\pgflinewidth}{0.200000\du}}{0cm}
\pgfsetbuttcap
{
\definecolor{dialinecolor}{rgb}{0.000000, 0.000000, 0.000000}
\pgfsetfillcolor{dialinecolor}
\definecolor{dialinecolor}{rgb}{0.000000, 0.000000, 0.000000}
\pgfsetstrokecolor{dialinecolor}
\draw (8.400000\du,14.600000\du)--(14.000000\du,10.000000\du);
}
\pgfsetlinewidth{0.0500000\du}
\pgfsetdash{{\pgflinewidth}{0.200000\du}}{0cm}
\pgfsetdash{{\pgflinewidth}{0.200000\du}}{0cm}
\pgfsetbuttcap
{
\definecolor{dialinecolor}{rgb}{0.000000, 0.000000, 0.000000}
\pgfsetfillcolor{dialinecolor}
\definecolor{dialinecolor}{rgb}{0.000000, 0.000000, 0.000000}
\pgfsetstrokecolor{dialinecolor}
\draw (8.400000\du,14.600000\du)--(10.600000\du,14.600000\du);
}
\definecolor{dialinecolor}{rgb}{0.000000, 0.000000, 0.000000}
\pgfsetstrokecolor{dialinecolor}
\node[anchor=west] at (9.410854\du,14.072965\du){$45^o$};
\definecolor{dialinecolor}{rgb}{0.000000, 0.000000, 0.000000}
\pgfsetstrokecolor{dialinecolor}
\node[anchor=west] at (8.800000\du,11.000000\du){x=\{1,2\}};
\definecolor{dialinecolor}{rgb}{0.000000, 0.000000, 0.000000}
\pgfsetstrokecolor{dialinecolor}
\node[anchor=west] at (8.800000\du,10.200000\du){$p=\pi_2(w)$};
\end{tikzpicture}
\end{subfigure}
\begin{subfigure}[b]{0.45\textwidth}
\centering
\input{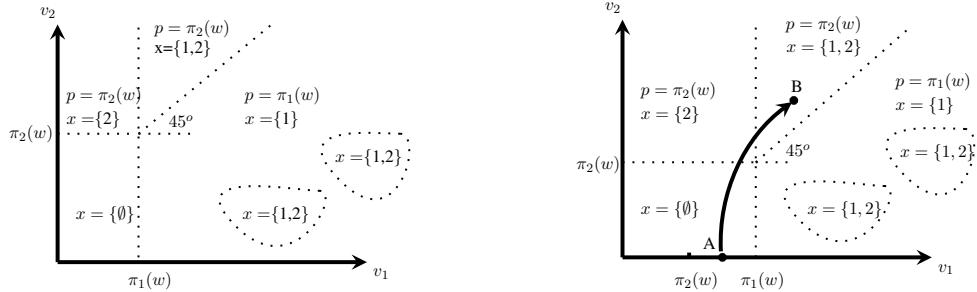}
\end{subfigure}
\caption{The threshold for item $2$ has to be at last as large as the threshold for item $1$. Otherwise the right figure portrays a profitable deviation.}\label{fig:deviations2}
\end{figure}
\end{lemma}
\begin{proof}
Fix a $w$ s.t. $w_1 > 0$. Note that by Lemma~\ref{lem:Pi1NotZero}, finite approximation factor immediately implies that $\pi_1 >0$ and is finite. Suppose in contrary to the lemma, we had $\pi_2(w) < \pi_1(w)$. Consider a $v = (\frac{\pi_1+\pi_2}{2},0)$ (point $A$ in Figure \ref{fig:deviations2}). Such a $v$ in $B_L^V$ has an incentive to report to be in $T_R^V$ (e.g. point $B$ in Figure \ref{fig:deviations2}). By Lemma~\ref{lem:BothAbove45Deg} such a report earns him a utility of $\frac{\pi_1+\pi_2}{2} - \pi_2 > 0$, instead of the $0$ utility he currently gets.  
\end{proof}

\begin{lemma}
\label{lem:1Below45Deg}
For every deterministic DSIC+IR mechanism that gets a finite approximation factor,  
$$\forall w \in \R_+^2: w_1 > 0 ~\& ~\pi_2 > \pi_1, \forall v \in B_R^V,\text{ we have } \{\allocv_1(v,w) = 1, \allocv_2(v,w) = 0, \pricev(v,w) = \pi_1(w)\}$$
\end{lemma}
\begin{proof}
Fix a $w \in \R_+^2$ s.t. $w_1 > 0$. From Lemma~\ref{lem:VerticalLine} it is clear that for all $v \in B_R^V$, $\allocv_1(v,w) = 1$. As far as allocation is concerned we only have to prove that for all $v \in B_R^V$, we have $\allocv_2(v,w) = 0$. Once we prove this, the fact that $\pricev(v,w) = \pi_1$ for all $v \in B_R^V$ immediately follows from Lemma~\ref{lem:1sPayPi1}. 
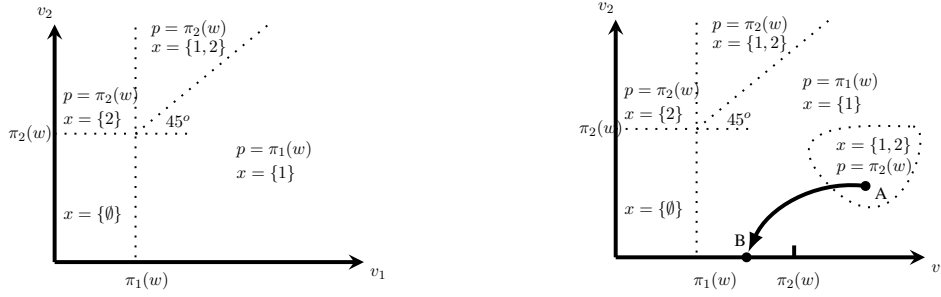
\begin{figure}[h]
\centering
\begin{subfigure}[b]{0.45\textwidth}
\centering
\ifx\du\undefined
  \newlength{\du}
\fi
\setlength{\du}{15\unitlength}
\begin{tikzpicture}
\pgftransformxscale{1.000000}
\pgftransformyscale{-1.000000}
\definecolor{dialinecolor}{rgb}{0.000000, 0.000000, 0.000000}
\pgfsetstrokecolor{dialinecolor}
\definecolor{dialinecolor}{rgb}{1.000000, 1.000000, 1.000000}
\pgfsetfillcolor{dialinecolor}
\pgfsetlinewidth{0.100000\du}
\pgfsetdash{}{0pt}
\pgfsetdash{}{0pt}
\pgfsetbuttcap
{
\definecolor{dialinecolor}{rgb}{0.000000, 0.000000, 0.000000}
\pgfsetfillcolor{dialinecolor}
\pgfsetarrowsend{stealth}
\definecolor{dialinecolor}{rgb}{0.000000, 0.000000, 0.000000}
\pgfsetstrokecolor{dialinecolor}
\draw (5.000000\du,20.000000\du)--(18.000000\du,20.000000\du);
}
\pgfsetlinewidth{0.100000\du}
\pgfsetdash{}{0pt}
\pgfsetdash{}{0pt}
\pgfsetbuttcap
{
\definecolor{dialinecolor}{rgb}{0.000000, 0.000000, 0.000000}
\pgfsetfillcolor{dialinecolor}
\pgfsetarrowsend{stealth}
\definecolor{dialinecolor}{rgb}{0.000000, 0.000000, 0.000000}
\pgfsetstrokecolor{dialinecolor}
\draw (5.000000\du,20.000000\du)--(5.000000\du,10.000000\du);
}
\pgfsetlinewidth{0.0500000\du}
\pgfsetdash{{\pgflinewidth}{0.200000\du}}{0cm}
\pgfsetdash{{\pgflinewidth}{0.200000\du}}{0cm}
\pgfsetbuttcap
{
\definecolor{dialinecolor}{rgb}{0.000000, 0.000000, 0.000000}
\pgfsetfillcolor{dialinecolor}
\definecolor{dialinecolor}{rgb}{0.000000, 0.000000, 0.000000}
\pgfsetstrokecolor{dialinecolor}
\draw (8.400000\du,19.800000\du)--(8.400000\du,9.800000\du);
}
\definecolor{dialinecolor}{rgb}{0.000000, 0.000000, 0.000000}
\pgfsetstrokecolor{dialinecolor}
\node[anchor=west] at (7.700000\du,20.800000\du){$\pi_1(w)$};
\definecolor{dialinecolor}{rgb}{0.000000, 0.000000, 0.000000}
\pgfsetstrokecolor{dialinecolor}
\node[anchor=west] at (12.381105\du,16.261912\du){$x=\{1\}$};
\definecolor{dialinecolor}{rgb}{0.000000, 0.000000, 0.000000}
\pgfsetstrokecolor{dialinecolor}
\node[anchor=west] at (6.000000\du,13.000000\du){};
\definecolor{dialinecolor}{rgb}{0.000000, 0.000000, 0.000000}
\pgfsetstrokecolor{dialinecolor}
\node[anchor=west] at (5.100000\du,14.000000\du){$x=\{2\}$};
\definecolor{dialinecolor}{rgb}{0.000000, 0.000000, 0.000000}
\pgfsetstrokecolor{dialinecolor}
\node[anchor=west] at (18.000000\du,20.500000\du){$v_1$};
\definecolor{dialinecolor}{rgb}{0.000000, 0.000000, 0.000000}
\pgfsetstrokecolor{dialinecolor}
\node[anchor=west] at (4.000000\du,9.500000\du){$v_2$};
\pgfsetlinewidth{0.0500000\du}
\pgfsetdash{{\pgflinewidth}{0.200000\du}}{0cm}
\pgfsetdash{{\pgflinewidth}{0.200000\du}}{0cm}
\pgfsetbuttcap
{
\definecolor{dialinecolor}{rgb}{0.000000, 0.000000, 0.000000}
\pgfsetfillcolor{dialinecolor}
\definecolor{dialinecolor}{rgb}{0.000000, 0.000000, 0.000000}
\pgfsetstrokecolor{dialinecolor}
\draw (5.000000\du,14.600000\du)--(8.300000\du,14.600000\du);
}
\definecolor{dialinecolor}{rgb}{0.000000, 0.000000, 0.000000}
\pgfsetstrokecolor{dialinecolor}
\node[anchor=west] at (5.100000\du,18.000000\du){$x=\{\emptyset\}$};
\definecolor{dialinecolor}{rgb}{0.000000, 0.000000, 0.000000}
\pgfsetstrokecolor{dialinecolor}
\node[anchor=west] at (5.100000\du,13.000000\du){$p=\pi_2(w)$};
\definecolor{dialinecolor}{rgb}{0.000000, 0.000000, 0.000000}
\pgfsetstrokecolor{dialinecolor}
\node[anchor=west] at (2.800000\du,14.600000\du){$\pi_2(w)$};
\definecolor{dialinecolor}{rgb}{0.000000, 0.000000, 0.000000}
\pgfsetstrokecolor{dialinecolor}
\node[anchor=west] at (12.381105\du,15.261912\du){$p=\pi_1(w)$};
\pgfsetlinewidth{0.0500000\du}
\pgfsetdash{{\pgflinewidth}{0.200000\du}}{0cm}
\pgfsetdash{{\pgflinewidth}{0.200000\du}}{0cm}
\pgfsetbuttcap
{
\definecolor{dialinecolor}{rgb}{0.000000, 0.000000, 0.000000}
\pgfsetfillcolor{dialinecolor}
\definecolor{dialinecolor}{rgb}{0.000000, 0.000000, 0.000000}
\pgfsetstrokecolor{dialinecolor}
\draw (8.400000\du,14.600000\du)--(14.000000\du,10.000000\du);
}
\pgfsetlinewidth{0.0500000\du}
\pgfsetdash{{\pgflinewidth}{0.200000\du}}{0cm}
\pgfsetdash{{\pgflinewidth}{0.200000\du}}{0cm}
\pgfsetbuttcap
{
\definecolor{dialinecolor}{rgb}{0.000000, 0.000000, 0.000000}
\pgfsetfillcolor{dialinecolor}
\definecolor{dialinecolor}{rgb}{0.000000, 0.000000, 0.000000}
\pgfsetstrokecolor{dialinecolor}
\draw (8.400000\du,14.600000\du)--(10.600000\du,14.600000\du);
}
\definecolor{dialinecolor}{rgb}{0.000000, 0.000000, 0.000000}
\pgfsetstrokecolor{dialinecolor}
\node[anchor=west] at (9.410850\du,14.073000\du){$45^o$};
\definecolor{dialinecolor}{rgb}{0.000000, 0.000000, 0.000000}
\pgfsetstrokecolor{dialinecolor}
\node[anchor=west] at (8.800000\du,11.000000\du){$x=\{1,2\}$};
\definecolor{dialinecolor}{rgb}{0.000000, 0.000000, 0.000000}
\pgfsetstrokecolor{dialinecolor}
\node[anchor=west] at (8.800000\du,10.200000\du){$p=\pi_2(w)$};
\end{tikzpicture}
\end{subfigure}
\begin{subfigure}[b]{0.45\textwidth}
\centering
\ifx\du\undefined
  \newlength{\du}
\fi
\setlength{\du}{15\unitlength}
\begin{tikzpicture}
\pgftransformxscale{1.000000}
\pgftransformyscale{-1.000000}
\definecolor{dialinecolor}{rgb}{0.000000, 0.000000, 0.000000}
\pgfsetstrokecolor{dialinecolor}
\definecolor{dialinecolor}{rgb}{1.000000, 1.000000, 1.000000}
\pgfsetfillcolor{dialinecolor}
\pgfsetlinewidth{0.100000\du}
\pgfsetdash{}{0pt}
\pgfsetdash{}{0pt}
\pgfsetbuttcap
{
\definecolor{dialinecolor}{rgb}{0.000000, 0.000000, 0.000000}
\pgfsetfillcolor{dialinecolor}
\pgfsetarrowsend{stealth}
\definecolor{dialinecolor}{rgb}{0.000000, 0.000000, 0.000000}
\pgfsetstrokecolor{dialinecolor}
\draw (5.000000\du,20.000000\du)--(18.000000\du,20.000000\du);
}
\pgfsetlinewidth{0.100000\du}
\pgfsetdash{}{0pt}
\pgfsetdash{}{0pt}
\pgfsetbuttcap
{
\definecolor{dialinecolor}{rgb}{0.000000, 0.000000, 0.000000}
\pgfsetfillcolor{dialinecolor}
\pgfsetarrowsend{stealth}
\definecolor{dialinecolor}{rgb}{0.000000, 0.000000, 0.000000}
\pgfsetstrokecolor{dialinecolor}
\draw (5.000000\du,20.000000\du)--(5.000000\du,10.000000\du);
}
\pgfsetlinewidth{0.0500000\du}
\pgfsetdash{{\pgflinewidth}{0.200000\du}}{0cm}
\pgfsetdash{{\pgflinewidth}{0.200000\du}}{0cm}
\pgfsetbuttcap
{
\definecolor{dialinecolor}{rgb}{0.000000, 0.000000, 0.000000}
\pgfsetfillcolor{dialinecolor}
\definecolor{dialinecolor}{rgb}{0.000000, 0.000000, 0.000000}
\pgfsetstrokecolor{dialinecolor}
\draw (8.400000\du,19.800000\du)--(8.400000\du,9.800000\du);
}
\definecolor{dialinecolor}{rgb}{0.000000, 0.000000, 0.000000}
\pgfsetstrokecolor{dialinecolor}
\node[anchor=west] at (8.000000\du,21.000000\du){$\pi_1(w)$};
\definecolor{dialinecolor}{rgb}{0.000000, 0.000000, 0.000000}
\pgfsetstrokecolor{dialinecolor}
\node[anchor=west] at (12.600000\du,13.600000\du){$x=\{1\}$};
\definecolor{dialinecolor}{rgb}{0.000000, 0.000000, 0.000000}
\pgfsetstrokecolor{dialinecolor}
\node[anchor=west] at (6.000000\du,13.000000\du){};
\definecolor{dialinecolor}{rgb}{0.000000, 0.000000, 0.000000}
\pgfsetstrokecolor{dialinecolor}
\node[anchor=west] at (5.100000\du,14.000000\du){$x=\{2\}$};
\definecolor{dialinecolor}{rgb}{0.000000, 0.000000, 0.000000}
\pgfsetstrokecolor{dialinecolor}
\node[anchor=west] at (18.000000\du,20.500000\du){$v_1$};
\definecolor{dialinecolor}{rgb}{0.000000, 0.000000, 0.000000}
\pgfsetstrokecolor{dialinecolor}
\node[anchor=west] at (4.000000\du,9.500000\du){$v_2$};
\pgfsetlinewidth{0.0500000\du}
\pgfsetdash{{\pgflinewidth}{0.200000\du}}{0cm}
\pgfsetdash{{\pgflinewidth}{0.200000\du}}{0cm}
\pgfsetbuttcap
{
\definecolor{dialinecolor}{rgb}{0.000000, 0.000000, 0.000000}
\pgfsetfillcolor{dialinecolor}
\definecolor{dialinecolor}{rgb}{0.000000, 0.000000, 0.000000}
\pgfsetstrokecolor{dialinecolor}
\draw (5.000000\du,14.600000\du)--(8.300000\du,14.600000\du);
}
\definecolor{dialinecolor}{rgb}{0.000000, 0.000000, 0.000000}
\pgfsetstrokecolor{dialinecolor}
\node[anchor=west] at (5.100000\du,18.000000\du){$x=\{\emptyset\}$};
\definecolor{dialinecolor}{rgb}{0.000000, 0.000000, 0.000000}
\pgfsetstrokecolor{dialinecolor}
\node[anchor=west] at (5.100000\du,13.000000\du){$p=\pi_2(w)$};
\definecolor{dialinecolor}{rgb}{0.000000, 0.000000, 0.000000}
\pgfsetstrokecolor{dialinecolor}
\node[anchor=west] at (3.200000\du,14.600000\du){$\pi_2(w)$};
\pgfsetlinewidth{0.0500000\du}
\pgfsetdash{{\pgflinewidth}{0.200000\du}}{0cm}
\pgfsetdash{{\pgflinewidth}{0.200000\du}}{0cm}
\pgfsetmiterjoin
\pgfsetbuttcap
\definecolor{dialinecolor}{rgb}{1.000000, 1.000000, 1.000000}
\pgfsetfillcolor{dialinecolor}
\pgfpathmoveto{\pgfpoint{17.500000\du}{14.500000\du}}
\pgfpathcurveto{\pgfpoint{18.500000\du}{14.500000\du}}{\pgfpoint{17.000000\du}{20.000000\du}}{\pgfpoint{14.000000\du}{17.000000\du}}
\pgfpathcurveto{\pgfpoint{11.000000\du}{14.000000\du}}{\pgfpoint{16.500000\du}{14.500000\du}}{\pgfpoint{17.500000\du}{14.500000\du}}
\pgfusepath{fill}
\definecolor{dialinecolor}{rgb}{0.000000, 0.000000, 0.000000}
\pgfsetstrokecolor{dialinecolor}
\pgfpathmoveto{\pgfpoint{17.500000\du}{14.500000\du}}
\pgfpathcurveto{\pgfpoint{18.500000\du}{14.500000\du}}{\pgfpoint{17.000000\du}{20.000000\du}}{\pgfpoint{14.000000\du}{17.000000\du}}
\pgfpathcurveto{\pgfpoint{11.000000\du}{14.000000\du}}{\pgfpoint{16.500000\du}{14.500000\du}}{\pgfpoint{17.500000\du}{14.500000\du}}
\pgfusepath{stroke}
\definecolor{dialinecolor}{rgb}{0.000000, 0.000000, 0.000000}
\pgfsetstrokecolor{dialinecolor}
\node[anchor=west] at (12.600000\du,12.600000\du){$p=\pi_1(w)$};
\definecolor{dialinecolor}{rgb}{0.000000, 0.000000, 0.000000}
\pgfsetstrokecolor{dialinecolor}
\node[anchor=west] at (14.000000\du,15.300000\du){$x=\{1,2\}$};
\pgfsetlinewidth{0.0500000\du}
\pgfsetdash{{\pgflinewidth}{0.200000\du}}{0cm}
\pgfsetdash{{\pgflinewidth}{0.200000\du}}{0cm}
\pgfsetbuttcap
{
\definecolor{dialinecolor}{rgb}{0.000000, 0.000000, 0.000000}
\pgfsetfillcolor{dialinecolor}
\definecolor{dialinecolor}{rgb}{0.000000, 0.000000, 0.000000}
\pgfsetstrokecolor{dialinecolor}
\draw (8.400000\du,14.600000\du)--(14.000000\du,10.000000\du);
}
\pgfsetlinewidth{0.0500000\du}
\pgfsetdash{{\pgflinewidth}{0.200000\du}}{0cm}
\pgfsetdash{{\pgflinewidth}{0.200000\du}}{0cm}
\pgfsetbuttcap
{
\definecolor{dialinecolor}{rgb}{0.000000, 0.000000, 0.000000}
\pgfsetfillcolor{dialinecolor}
\definecolor{dialinecolor}{rgb}{0.000000, 0.000000, 0.000000}
\pgfsetstrokecolor{dialinecolor}
\draw (8.400000\du,14.600000\du)--(10.600000\du,14.600000\du);
}
\definecolor{dialinecolor}{rgb}{0.000000, 0.000000, 0.000000}
\pgfsetstrokecolor{dialinecolor}
\node[anchor=west] at (9.410850\du,14.173000\du){$45^o$};
\definecolor{dialinecolor}{rgb}{0.000000, 0.000000, 0.000000}
\pgfsetstrokecolor{dialinecolor}
\node[anchor=west] at (8.800000\du,11.000000\du){$x=\{1,2\}$};
\definecolor{dialinecolor}{rgb}{0.000000, 0.000000, 0.000000}
\pgfsetstrokecolor{dialinecolor}
\node[anchor=west] at (8.800000\du,10.200000\du){$p=\pi_2(w)$};
\definecolor{dialinecolor}{rgb}{0.000000, 0.000000, 0.000000}
\pgfsetstrokecolor{dialinecolor}
\node[anchor=west] at (11.500000\du,21.000000\du){$\pi_2(w)$};
\pgfsetlinewidth{0.100000\du}
\pgfsetdash{}{0pt}
\pgfsetdash{}{0pt}
\pgfsetbuttcap
{
\definecolor{dialinecolor}{rgb}{0.000000, 0.000000, 0.000000}
\pgfsetfillcolor{dialinecolor}
\definecolor{dialinecolor}{rgb}{0.000000, 0.000000, 0.000000}
\pgfsetstrokecolor{dialinecolor}
\draw (12.500000\du,20.000000\du)--(12.500000\du,19.500000\du);
}
\definecolor{dialinecolor}{rgb}{0.000000, 0.000000, 0.000000}
\pgfsetfillcolor{dialinecolor}
\pgfpathellipse{\pgfpoint{10.500000\du}{20.000000\du}}{\pgfpoint{0.125000\du}{0\du}}{\pgfpoint{0\du}{0.125000\du}}
\pgfusepath{fill}
\pgfsetlinewidth{0.100000\du}
\pgfsetdash{}{0pt}
\pgfsetdash{}{0pt}
\definecolor{dialinecolor}{rgb}{0.000000, 0.000000, 0.000000}
\pgfsetstrokecolor{dialinecolor}
\pgfpathellipse{\pgfpoint{10.500000\du}{20.000000\du}}{\pgfpoint{0.125000\du}{0\du}}{\pgfpoint{0\du}{0.125000\du}}
\pgfusepath{stroke}
\definecolor{dialinecolor}{rgb}{0.000000, 0.000000, 0.000000}
\pgfsetfillcolor{dialinecolor}
\pgfpathellipse{\pgfpoint{15.500000\du}{17.000000\du}}{\pgfpoint{0.125000\du}{0\du}}{\pgfpoint{0\du}{0.125000\du}}
\pgfusepath{fill}
\pgfsetlinewidth{0.100000\du}
\pgfsetdash{}{0pt}
\pgfsetdash{}{0pt}
\definecolor{dialinecolor}{rgb}{0.000000, 0.000000, 0.000000}
\pgfsetstrokecolor{dialinecolor}
\pgfpathellipse{\pgfpoint{15.500000\du}{17.000000\du}}{\pgfpoint{0.125000\du}{0\du}}{\pgfpoint{0\du}{0.125000\du}}
\pgfusepath{stroke}
\definecolor{dialinecolor}{rgb}{0.000000, 0.000000, 0.000000}
\pgfsetstrokecolor{dialinecolor}
\node[anchor=west] at (14.000000\du,16.200000\du){$p=\pi_2(w)$};
\pgfsetlinewidth{0.100000\du}
\pgfsetdash{}{0pt}
\pgfsetdash{}{0pt}
\pgfsetbuttcap
{
\definecolor{dialinecolor}{rgb}{0.000000, 0.000000, 0.000000}
\pgfsetfillcolor{dialinecolor}
\pgfsetarrowsend{latex}
\definecolor{dialinecolor}{rgb}{0.000000, 0.000000, 0.000000}
\pgfsetstrokecolor{dialinecolor}
\pgfpathmoveto{\pgfpoint{15.325924\du}{16.982182\du}}
\pgfpatharc{275}{204}{4.750000\du and 4.750000\du}
\pgfusepath{stroke}
}
\definecolor{dialinecolor}{rgb}{0.000000, 0.000000, 0.000000}
\pgfsetstrokecolor{dialinecolor}
\node[anchor=west] at (9.700000\du,19.300000\du){B};
\definecolor{dialinecolor}{rgb}{0.000000, 0.000000, 0.000000}
\pgfsetstrokecolor{dialinecolor}
\node[anchor=west] at (15.600000\du,17.300000\du){A};
\end{tikzpicture}
\end{subfigure}
\caption{If the threshold for item $2$ is strictly larger than that of item $1$, then every bid in $B_R^V$ is allocated only item $1$ and charged $\pi_1(w)$ (left). Otherwise the right figure portrays a profitable deviation.}\label{fig:deviations4}
\end{figure}
	To prove that for all $v \in B_R^V$ we have $\allocv_2(v,w) = 0$:
	\begin{enumerate}
		\item We first show that for \emph{some} point $v \in B_R^V$ we have $\allocv_2(v,w) = 0$. Suppose not. 
		Note that by Lemma~\ref{lem:BothAbove45Deg} there is a non-empty set $T_R^V$ where bidder $V$ gets both items and pays $\pi_2$. 
		DSIC therefore implies that at any point where bidder $V$ gets both items he pays $\pi_2$, including the region $B_R^V$ under consideration. 
		Consider a point $(\frac{\pi_1+\pi_2}{2}, 0) \in B_R^V$ (point $B$ in Figure \ref{fig:deviations4}). 
		Such a bidder, by our assumption gets both items allocated and has to therefore pay $\pi_2$ which is strictly larger than his value, violating IR. 
		Therefore there is at least one point in $B_R^V$ where $\allocv_2(v,w) = 0$. 
		\item We now show that for all points in $B_R^V$ we have have  $\allocv_2(v,w) = 0$. To see this, consider a point $v \in B_R^V$ where $\allocv_2(v,w) = 1$ (point $A$ in Figure \ref{fig:deviations4}). Such a point, by the discussion in this paragraph, pays $\pi_2$. However, switching to a point where item $1$ alone is allocated (there is at least one such point by the discussion above) gives the bidder a utility of $v_1 - \pi_1$, which is strictly larger than his current utility of $\max\{v_1,v_2\} - \pi_2$.
	\end{enumerate}
\end{proof}

At this point, it is convenient to summarize our findings on how the threshold functions $\pi_2$ and $\pi_1$ almost entirely determine the allocation for bidder $V$ (and similarly $\phi_1$ and $\phi_2$ for bidder $W$). Except for the region $B_R^V$, the threshold functions $\pi_2$ and $\pi_1$ completely determine the allocation function for bidder $V$. In the region $B_R^V$, bidder $V$ could receive just item $1$ or both items. Even in $B_R^V$, when $\pi_2 > \pi_1$, bidder $V$ can just receive item $1$. 

\begin{corollary}
\label{cor:AllocationRegions}
Fix any $w \in \R_+^2: w_1 > 0$. 
For every deterministic DSIC+IR mechanism that gets a finite approximation factor, the allocation in the regions $B_L^V, B_R^V, T_L^V, T_R^V$ should satisfy the following:
\begin{enumerate}
\item $\forall v \in B_L^V,\ \allocv_1(v,w) = \allocv_2(v,w) = 0$. 
\item $\forall v \in B_R^V,\ \allocv_1(v,w) = 1, \allocv_2(v,w) \in \{0,1\}$. Further, if $\pi_2 > \pi_1$, we have $\allocv_2(v,w) = 0$. 
\item $\forall v \in T_L^V, \ \allocv_1(v,w) = 0, \allocv_2(v,w) = 1$. 
\item $\forall v \in T_R^V,\ \allocv_2(v,w) = \allocv_2(v,w) = 1$. 
\item Analogous characterization holds for $B_L^W, B_R^W, T_L^W, T_R^W$. 
\end{enumerate}
\begin{figure}[h]
\centering
\begin{subfigure}[b]{0.45\textwidth}
\centering
\ifx\du\undefined
  \newlength{\du}
\fi
\setlength{\du}{15\unitlength}
\begin{tikzpicture}
\pgftransformxscale{1.000000}
\pgftransformyscale{-1.000000}
\definecolor{dialinecolor}{rgb}{0.000000, 0.000000, 0.000000}
\pgfsetstrokecolor{dialinecolor}
\definecolor{dialinecolor}{rgb}{1.000000, 1.000000, 1.000000}
\pgfsetfillcolor{dialinecolor}
\pgfsetlinewidth{0.100000\du}
\pgfsetdash{}{0pt}
\pgfsetdash{}{0pt}
\pgfsetbuttcap
{
\definecolor{dialinecolor}{rgb}{0.000000, 0.000000, 0.000000}
\pgfsetfillcolor{dialinecolor}
\pgfsetarrowsend{stealth}
\definecolor{dialinecolor}{rgb}{0.000000, 0.000000, 0.000000}
\pgfsetstrokecolor{dialinecolor}
\draw (5.000000\du,20.000000\du)--(18.000000\du,20.000000\du);
}
\pgfsetlinewidth{0.100000\du}
\pgfsetdash{}{0pt}
\pgfsetdash{}{0pt}
\pgfsetbuttcap
{
\definecolor{dialinecolor}{rgb}{0.000000, 0.000000, 0.000000}
\pgfsetfillcolor{dialinecolor}
\pgfsetarrowsend{stealth}
\definecolor{dialinecolor}{rgb}{0.000000, 0.000000, 0.000000}
\pgfsetstrokecolor{dialinecolor}
\draw (5.000000\du,20.000000\du)--(5.000000\du,10.000000\du);
}
\pgfsetlinewidth{0.0500000\du}
\pgfsetdash{{\pgflinewidth}{0.200000\du}}{0cm}
\pgfsetdash{{\pgflinewidth}{0.200000\du}}{0cm}
\pgfsetbuttcap
{
\definecolor{dialinecolor}{rgb}{0.000000, 0.000000, 0.000000}
\pgfsetfillcolor{dialinecolor}
\definecolor{dialinecolor}{rgb}{0.000000, 0.000000, 0.000000}
\pgfsetstrokecolor{dialinecolor}
\draw (8.400000\du,19.800000\du)--(8.400000\du,9.800000\du);
}
\definecolor{dialinecolor}{rgb}{0.000000, 0.000000, 0.000000}
\pgfsetstrokecolor{dialinecolor}
\node[anchor=west] at (7.700000\du,20.800000\du){$\pi_1(w)$};
\definecolor{dialinecolor}{rgb}{0.000000, 0.000000, 0.000000}
\pgfsetstrokecolor{dialinecolor}
\node[anchor=west] at (12.381105\du,16.261912\du){$x=\{1\}$};
\definecolor{dialinecolor}{rgb}{0.000000, 0.000000, 0.000000}
\pgfsetstrokecolor{dialinecolor}
\node[anchor=west] at (6.000000\du,13.000000\du){};
\definecolor{dialinecolor}{rgb}{0.000000, 0.000000, 0.000000}
\pgfsetstrokecolor{dialinecolor}
\node[anchor=west] at (5.100000\du,14.000000\du){$x=\{2\}$};
\definecolor{dialinecolor}{rgb}{0.000000, 0.000000, 0.000000}
\pgfsetstrokecolor{dialinecolor}
\node[anchor=west] at (18.000000\du,20.500000\du){$v_1$};
\definecolor{dialinecolor}{rgb}{0.000000, 0.000000, 0.000000}
\pgfsetstrokecolor{dialinecolor}
\node[anchor=west] at (4.000000\du,9.500000\du){$v_2$};
\pgfsetlinewidth{0.0500000\du}
\pgfsetdash{{\pgflinewidth}{0.200000\du}}{0cm}
\pgfsetdash{{\pgflinewidth}{0.200000\du}}{0cm}
\pgfsetbuttcap
{
\definecolor{dialinecolor}{rgb}{0.000000, 0.000000, 0.000000}
\pgfsetfillcolor{dialinecolor}
\definecolor{dialinecolor}{rgb}{0.000000, 0.000000, 0.000000}
\pgfsetstrokecolor{dialinecolor}
\draw (5.000000\du,14.600000\du)--(8.300000\du,14.600000\du);
}
\definecolor{dialinecolor}{rgb}{0.000000, 0.000000, 0.000000}
\pgfsetstrokecolor{dialinecolor}
\node[anchor=west] at (5.100000\du,18.000000\du){$x=\{\emptyset\}$};
\definecolor{dialinecolor}{rgb}{0.000000, 0.000000, 0.000000}
\pgfsetstrokecolor{dialinecolor}
\node[anchor=west] at (5.100000\du,13.000000\du){$p=\pi_2(w)$};
\definecolor{dialinecolor}{rgb}{0.000000, 0.000000, 0.000000}
\pgfsetstrokecolor{dialinecolor}
\node[anchor=west] at (2.800000\du,14.600000\du){$\pi_2(w)$};
\definecolor{dialinecolor}{rgb}{0.000000, 0.000000, 0.000000}
\pgfsetstrokecolor{dialinecolor}
\node[anchor=west] at (12.381105\du,15.261912\du){$p=\pi_1(w)$};
\pgfsetlinewidth{0.0500000\du}
\pgfsetdash{{\pgflinewidth}{0.200000\du}}{0cm}
\pgfsetdash{{\pgflinewidth}{0.200000\du}}{0cm}
\pgfsetbuttcap
{
\definecolor{dialinecolor}{rgb}{0.000000, 0.000000, 0.000000}
\pgfsetfillcolor{dialinecolor}
\definecolor{dialinecolor}{rgb}{0.000000, 0.000000, 0.000000}
\pgfsetstrokecolor{dialinecolor}
\draw (8.400000\du,14.600000\du)--(14.000000\du,10.000000\du);
}
\pgfsetlinewidth{0.0500000\du}
\pgfsetdash{{\pgflinewidth}{0.200000\du}}{0cm}
\pgfsetdash{{\pgflinewidth}{0.200000\du}}{0cm}
\pgfsetbuttcap
{
\definecolor{dialinecolor}{rgb}{0.000000, 0.000000, 0.000000}
\pgfsetfillcolor{dialinecolor}
\definecolor{dialinecolor}{rgb}{0.000000, 0.000000, 0.000000}
\pgfsetstrokecolor{dialinecolor}
\draw (8.400000\du,14.600000\du)--(10.600000\du,14.600000\du);
}
\definecolor{dialinecolor}{rgb}{0.000000, 0.000000, 0.000000}
\pgfsetstrokecolor{dialinecolor}
\node[anchor=west] at (9.410850\du,14.073000\du){$45^o$};
\definecolor{dialinecolor}{rgb}{0.000000, 0.000000, 0.000000}
\pgfsetstrokecolor{dialinecolor}
\node[anchor=west] at (8.800000\du,11.000000\du){$x=\{1,2\}$};
\definecolor{dialinecolor}{rgb}{0.000000, 0.000000, 0.000000}
\pgfsetstrokecolor{dialinecolor}
\node[anchor=west] at (8.800000\du,10.200000\du){$p=\pi_2(w)$};
\end{tikzpicture}
\end{subfigure}
\begin{subfigure}[b]{0.45\textwidth}
\centering
\ifx\du\undefined
  \newlength{\du}
\fi
\setlength{\du}{15\unitlength}
\begin{tikzpicture}
\pgftransformxscale{1.000000}
\pgftransformyscale{-1.000000}
\definecolor{dialinecolor}{rgb}{0.000000, 0.000000, 0.000000}
\pgfsetstrokecolor{dialinecolor}
\definecolor{dialinecolor}{rgb}{1.000000, 1.000000, 1.000000}
\pgfsetfillcolor{dialinecolor}
\pgfsetlinewidth{0.100000\du}
\pgfsetdash{}{0pt}
\pgfsetdash{}{0pt}
\pgfsetbuttcap
{
\definecolor{dialinecolor}{rgb}{0.000000, 0.000000, 0.000000}
\pgfsetfillcolor{dialinecolor}
\pgfsetarrowsend{stealth}
\definecolor{dialinecolor}{rgb}{0.000000, 0.000000, 0.000000}
\pgfsetstrokecolor{dialinecolor}
\draw (5.000000\du,20.000000\du)--(18.000000\du,20.000000\du);
}
\pgfsetlinewidth{0.100000\du}
\pgfsetdash{}{0pt}
\pgfsetdash{}{0pt}
\pgfsetbuttcap
{
\definecolor{dialinecolor}{rgb}{0.000000, 0.000000, 0.000000}
\pgfsetfillcolor{dialinecolor}
\pgfsetarrowsend{stealth}
\definecolor{dialinecolor}{rgb}{0.000000, 0.000000, 0.000000}
\pgfsetstrokecolor{dialinecolor}
\draw (5.000000\du,20.000000\du)--(5.000000\du,10.000000\du);
}
\pgfsetlinewidth{0.0500000\du}
\pgfsetdash{{\pgflinewidth}{0.200000\du}}{0cm}
\pgfsetdash{{\pgflinewidth}{0.200000\du}}{0cm}
\pgfsetbuttcap
{
\definecolor{dialinecolor}{rgb}{0.000000, 0.000000, 0.000000}
\pgfsetfillcolor{dialinecolor}
\definecolor{dialinecolor}{rgb}{0.000000, 0.000000, 0.000000}
\pgfsetstrokecolor{dialinecolor}
\draw (8.500000\du,20.000000\du)--(8.500000\du,10.000000\du);
}
\definecolor{dialinecolor}{rgb}{0.000000, 0.000000, 0.000000}
\pgfsetstrokecolor{dialinecolor}
\node[anchor=west] at (7.700000\du,21.000000\du){$\pi_1(w)=\pi_2(w)$};
\definecolor{dialinecolor}{rgb}{0.000000, 0.000000, 0.000000}
\pgfsetstrokecolor{dialinecolor}
\node[anchor=west] at (12.600000\du,14.000000\du){$x=\{1\}$};
\definecolor{dialinecolor}{rgb}{0.000000, 0.000000, 0.000000}
\pgfsetstrokecolor{dialinecolor}
\node[anchor=west] at (6.000000\du,13.000000\du){};
\definecolor{dialinecolor}{rgb}{0.000000, 0.000000, 0.000000}
\pgfsetstrokecolor{dialinecolor}
\node[anchor=west] at (5.100000\du,14.000000\du){$x=$\{2\}};
\definecolor{dialinecolor}{rgb}{0.000000, 0.000000, 0.000000}
\pgfsetstrokecolor{dialinecolor}
\node[anchor=west] at (18.000000\du,20.500000\du){$v_1$};
\definecolor{dialinecolor}{rgb}{0.000000, 0.000000, 0.000000}
\pgfsetstrokecolor{dialinecolor}
\node[anchor=west] at (4.000000\du,9.500000\du){$v_2$};
\pgfsetlinewidth{0.0500000\du}
\pgfsetdash{{\pgflinewidth}{0.200000\du}}{0cm}
\pgfsetdash{{\pgflinewidth}{0.200000\du}}{0cm}
\pgfsetbuttcap
{
\definecolor{dialinecolor}{rgb}{0.000000, 0.000000, 0.000000}
\pgfsetfillcolor{dialinecolor}
\definecolor{dialinecolor}{rgb}{0.000000, 0.000000, 0.000000}
\pgfsetstrokecolor{dialinecolor}
\draw (5.100000\du,16.500000\du)--(8.400000\du,16.500000\du);
}
\definecolor{dialinecolor}{rgb}{0.000000, 0.000000, 0.000000}
\pgfsetstrokecolor{dialinecolor}
\node[anchor=west] at (5.100000\du,18.000000\du){$x=\{\emptyset\}$};
\definecolor{dialinecolor}{rgb}{0.000000, 0.000000, 0.000000}
\pgfsetstrokecolor{dialinecolor}
\node[anchor=west] at (5.100000\du,13.000000\du){$p=\pi_2(w)$};
\definecolor{dialinecolor}{rgb}{0.000000, 0.000000, 0.000000}
\pgfsetstrokecolor{dialinecolor}
\node[anchor=west] at (2.800000\du,16.600000\du){$\pi_2(w)$};
\pgfsetlinewidth{0.0500000\du}
\pgfsetdash{{\pgflinewidth}{0.200000\du}}{0cm}
\pgfsetdash{{\pgflinewidth}{0.200000\du}}{0cm}
\pgfsetmiterjoin
\pgfsetbuttcap
\definecolor{dialinecolor}{rgb}{1.000000, 1.000000, 1.000000}
\pgfsetfillcolor{dialinecolor}
\pgfpathmoveto{\pgfpoint{19.500000\du}{14.000000\du}}
\pgfpathcurveto{\pgfpoint{20.500000\du}{14.000000\du}}{\pgfpoint{19.000000\du}{19.500000\du}}{\pgfpoint{16.000000\du}{16.500000\du}}
\pgfpathcurveto{\pgfpoint{13.000000\du}{13.500000\du}}{\pgfpoint{18.500000\du}{14.000000\du}}{\pgfpoint{19.500000\du}{14.000000\du}}
\pgfusepath{fill}
\definecolor{dialinecolor}{rgb}{0.000000, 0.000000, 0.000000}
\pgfsetstrokecolor{dialinecolor}
\pgfpathmoveto{\pgfpoint{19.500000\du}{14.000000\du}}
\pgfpathcurveto{\pgfpoint{20.500000\du}{14.000000\du}}{\pgfpoint{19.000000\du}{19.500000\du}}{\pgfpoint{16.000000\du}{16.500000\du}}
\pgfpathcurveto{\pgfpoint{13.000000\du}{13.500000\du}}{\pgfpoint{18.500000\du}{14.000000\du}}{\pgfpoint{19.500000\du}{14.000000\du}}
\pgfusepath{stroke}
\definecolor{dialinecolor}{rgb}{0.000000, 0.000000, 0.000000}
\pgfsetstrokecolor{dialinecolor}
\node[anchor=west] at (12.600000\du,13.000000\du){$p=\pi_1(w)$};
\definecolor{dialinecolor}{rgb}{0.000000, 0.000000, 0.000000}
\pgfsetstrokecolor{dialinecolor}
\node[anchor=west] at (16.000000\du,15.000000\du){$x=\{1,2\}$};
\pgfsetlinewidth{0.0500000\du}
\pgfsetdash{{\pgflinewidth}{0.200000\du}}{0cm}
\pgfsetdash{{\pgflinewidth}{0.200000\du}}{0cm}
\pgfsetbuttcap
{
\definecolor{dialinecolor}{rgb}{0.000000, 0.000000, 0.000000}
\pgfsetfillcolor{dialinecolor}
\definecolor{dialinecolor}{rgb}{0.000000, 0.000000, 0.000000}
\pgfsetstrokecolor{dialinecolor}
\draw (8.500000\du,16.500000\du)--(15.000000\du,10.000000\du);
}
\pgfsetlinewidth{0.0500000\du}
\pgfsetdash{{\pgflinewidth}{0.200000\du}}{0cm}
\pgfsetdash{{\pgflinewidth}{0.200000\du}}{0cm}
\pgfsetbuttcap
{
\definecolor{dialinecolor}{rgb}{0.000000, 0.000000, 0.000000}
\pgfsetfillcolor{dialinecolor}
\definecolor{dialinecolor}{rgb}{0.000000, 0.000000, 0.000000}
\pgfsetstrokecolor{dialinecolor}
\draw (8.500000\du,16.500000\du)--(10.700000\du,16.500000\du);
}
\definecolor{dialinecolor}{rgb}{0.000000, 0.000000, 0.000000}
\pgfsetstrokecolor{dialinecolor}
\node[anchor=west] at (9.410850\du,15.973000\du){$45^o$};
\definecolor{dialinecolor}{rgb}{0.000000, 0.000000, 0.000000}
\pgfsetstrokecolor{dialinecolor}
\node[anchor=west] at (8.800000\du,11.000000\du){$x=\{1,2\}$};
\definecolor{dialinecolor}{rgb}{0.000000, 0.000000, 0.000000}
\pgfsetstrokecolor{dialinecolor}
\node[anchor=west] at (8.800000\du,10.200000\du){$p=\pi_2(w)$};
\definecolor{dialinecolor}{rgb}{0.000000, 0.000000, 0.000000}
\pgfsetstrokecolor{dialinecolor}
\node[anchor=west] at (16.000000\du,16.000000\du){$p=\pi_2(w)$};
\pgfsetlinewidth{0.0500000\du}
\pgfsetdash{{\pgflinewidth}{0.200000\du}}{0cm}
\pgfsetdash{{\pgflinewidth}{0.200000\du}}{0cm}
\pgfsetmiterjoin
\pgfsetbuttcap
\definecolor{dialinecolor}{rgb}{1.000000, 1.000000, 1.000000}
\pgfsetfillcolor{dialinecolor}
\pgfpathmoveto{\pgfpoint{14.000000\du}{16.000000\du}}
\pgfpathcurveto{\pgfpoint{17.000000\du}{16.000000\du}}{\pgfpoint{15.000000\du}{19.500000\du}}{\pgfpoint{12.000000\du}{18.500000\du}}
\pgfpathcurveto{\pgfpoint{9.000000\du}{17.500000\du}}{\pgfpoint{11.000000\du}{16.000000\du}}{\pgfpoint{14.000000\du}{16.000000\du}}
\pgfusepath{fill}
\definecolor{dialinecolor}{rgb}{0.000000, 0.000000, 0.000000}
\pgfsetstrokecolor{dialinecolor}
\pgfpathmoveto{\pgfpoint{14.000000\du}{16.000000\du}}
\pgfpathcurveto{\pgfpoint{17.000000\du}{16.000000\du}}{\pgfpoint{15.000000\du}{19.500000\du}}{\pgfpoint{12.000000\du}{18.500000\du}}
\pgfpathcurveto{\pgfpoint{9.000000\du}{17.500000\du}}{\pgfpoint{11.000000\du}{16.000000\du}}{\pgfpoint{14.000000\du}{16.000000\du}}
\pgfusepath{stroke}
\definecolor{dialinecolor}{rgb}{0.000000, 0.000000, 0.000000}
\pgfsetstrokecolor{dialinecolor}
\node[anchor=west] at (11.500000\du,17.000000\du){$x=\{1,2\}$};
\definecolor{dialinecolor}{rgb}{0.000000, 0.000000, 0.000000}
\pgfsetstrokecolor{dialinecolor}
\node[anchor=west] at (11.500000\du,18.000000\du){$p=\pi_2(w)$};
\end{tikzpicture}
\end{subfigure}
\caption{The allocation and payment functions have to look either like the figure on the left (if $\pi_2(w)>\pi_1(w)$) or like the figure on the right (if $\pi_2(w)=\pi_1(w)$).}\label{fig:char1}
\end{figure}
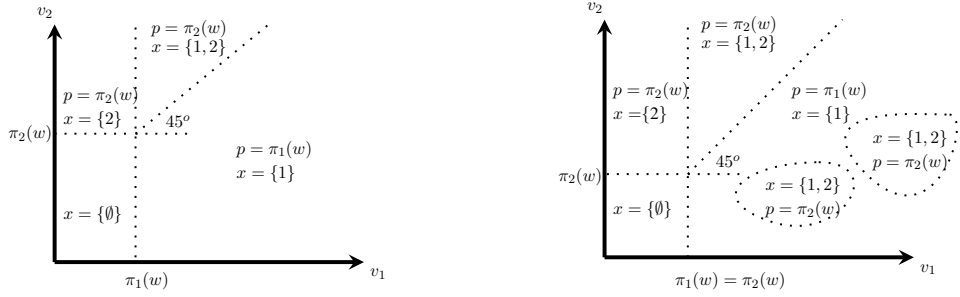
\end{corollary}

\subsection{Arguing Simultaneously for Allocation Function of Both Players}
\label{sec:HApprox}
Now, we assume that the mechanism gets an $H$ approximation for some fixed number $H$. 
Using this we argue that in a particular region of the type space, the mechanism must allocate both items. Using this along with the necessary conditions we have so far, we eventually arrive at a contradiction, 
implying that deterministic DSIC + IR mechanisms cannot get any finite approximation. 
We restrict our attention to the following subset of $\R_+^2$, for some $N \geq H$ and $\epsilon < \frac{1}{H}$.
\[S_{N,\epsilon} := \{(x_1,x_2) \in \R_+^2: 0 < x_1 < \epsilon, x_2  = N\}.\]
\begin{lemma}\label{lem:item2_always_allocated}
	For every deterministic DSIC+IR mechanism that gets an $H$ approximation, $\forall N \geq H, \forall \epsilon < \frac{1}{H}, \forall w \in S_{N,\epsilon}$, we have that 
\begin{enumerate}
	\item $\pi_1(w) < H\epsilon.$ 
	\item If further, $v_1 < H \epsilon$, then Item 2 is always allocated.
\end{enumerate}
\end{lemma}
\begin{proof}
	For any mechanism that gets an $H$-approximation, if item $2$ doesn't arrive, it is necessary that the mechanism allocates item $1$ to bidder $V$ when $v_1\geq H\epsilon$ given that $w_1 < \epsilon$.
	
	 For the second point, we have that $w_1 < \epsilon$ and $v_1 < H\epsilon$, while  $w_2 \geq H$. Thus not allocating item $2$ to any bidder at all  results in an approximation factor of at least $H$ since we have $\epsilon < \frac{1}{H}$. 
\end{proof}
\begin{lemma}
\label{lem:Pi2MinusPi1IndependentOfW1}
For every deterministic DSIC+IR mechanism that gets an $H$ approximation,  
\begin{align}
\label{eq:Pi2MinusPi1IndependentOfW1}
\forall N \geq H, \forall \epsilon < \frac{1}{H}, \forall w \in S_{N,\epsilon},\text{ we have } \pi_2(w) - \pi_1(w) = c_{N}.
\end{align} 
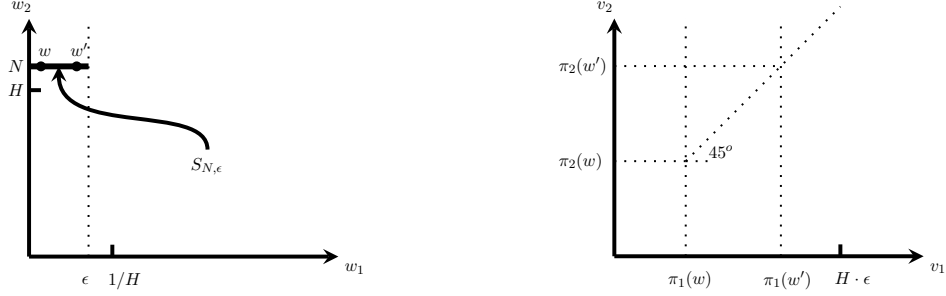
\begin{figure}[h]
\centering
\begin{subfigure}[b]{0.45\textwidth}
\centering
\ifx\du\undefined
  \newlength{\du}
\fi
\setlength{\du}{15\unitlength}
\begin{tikzpicture}
\pgftransformxscale{1.000000}
\pgftransformyscale{-1.000000}
\definecolor{dialinecolor}{rgb}{0.000000, 0.000000, 0.000000}
\pgfsetstrokecolor{dialinecolor}
\definecolor{dialinecolor}{rgb}{1.000000, 1.000000, 1.000000}
\pgfsetfillcolor{dialinecolor}
\pgfsetlinewidth{0.100000\du}
\pgfsetdash{}{0pt}
\pgfsetdash{}{0pt}
\pgfsetbuttcap
{
\definecolor{dialinecolor}{rgb}{0.000000, 0.000000, 0.000000}
\pgfsetfillcolor{dialinecolor}
\pgfsetarrowsend{stealth}
\definecolor{dialinecolor}{rgb}{0.000000, 0.000000, 0.000000}
\pgfsetstrokecolor{dialinecolor}
\draw (4.000000\du,17.000000\du)--(17.000000\du,17.000000\du);
}
\pgfsetlinewidth{0.100000\du}
\pgfsetdash{}{0pt}
\pgfsetdash{}{0pt}
\pgfsetbuttcap
{
\definecolor{dialinecolor}{rgb}{0.000000, 0.000000, 0.000000}
\pgfsetfillcolor{dialinecolor}
\pgfsetarrowsend{stealth}
\definecolor{dialinecolor}{rgb}{0.000000, 0.000000, 0.000000}
\pgfsetstrokecolor{dialinecolor}
\draw (4.000000\du,17.000000\du)--(4.000000\du,7.000000\du);
}
\pgfsetlinewidth{0.0500000\du}
\pgfsetdash{{\pgflinewidth}{0.200000\du}}{0cm}
\pgfsetdash{{\pgflinewidth}{0.200000\du}}{0cm}
\pgfsetbuttcap
{
\definecolor{dialinecolor}{rgb}{0.000000, 0.000000, 0.000000}
\pgfsetfillcolor{dialinecolor}
\definecolor{dialinecolor}{rgb}{0.000000, 0.000000, 0.000000}
\pgfsetstrokecolor{dialinecolor}
\draw (6.500000\du,17.000000\du)--(6.500000\du,7.000000\du);
}
\definecolor{dialinecolor}{rgb}{0.000000, 0.000000, 0.000000}
\pgfsetstrokecolor{dialinecolor}
\node[anchor=west] at (5.000000\du,10.000000\du){};
\definecolor{dialinecolor}{rgb}{0.000000, 0.000000, 0.000000}
\pgfsetstrokecolor{dialinecolor}
\node[anchor=west] at (17.000000\du,17.500000\du){$w_1$};
\definecolor{dialinecolor}{rgb}{0.000000, 0.000000, 0.000000}
\pgfsetstrokecolor{dialinecolor}
\node[anchor=west] at (3.000000\du,6.500000\du){$w_2$};
\pgfsetlinewidth{0.150000\du}
\pgfsetdash{}{0pt}
\pgfsetdash{}{0pt}
\pgfsetbuttcap
{
\definecolor{dialinecolor}{rgb}{0.000000, 0.000000, 0.000000}
\pgfsetfillcolor{dialinecolor}
\definecolor{dialinecolor}{rgb}{0.000000, 0.000000, 0.000000}
\pgfsetstrokecolor{dialinecolor}
\draw (4.000000\du,9.000000\du)--(6.500000\du,9.000000\du);
}
\definecolor{dialinecolor}{rgb}{0.000000, 0.000000, 0.000000}
\pgfsetstrokecolor{dialinecolor}
\node[anchor=west] at (5.900000\du,18.000000\du){$\epsilon$};
\definecolor{dialinecolor}{rgb}{0.000000, 0.000000, 0.000000}
\pgfsetstrokecolor{dialinecolor}
\node[anchor=west] at (2.800000\du,9.000000\du){$N$};
\definecolor{dialinecolor}{rgb}{0.000000, 0.000000, 0.000000}
\pgfsetstrokecolor{dialinecolor}
\node[anchor=west] at (7.000000\du,18.000000\du){$1/H$};
\definecolor{dialinecolor}{rgb}{0.000000, 0.000000, 0.000000}
\pgfsetstrokecolor{dialinecolor}
\node[anchor=west] at (2.8000000\du,10.000000\du){$H$};
\definecolor{dialinecolor}{rgb}{0.000000, 0.000000, 0.000000}
\pgfsetfillcolor{dialinecolor}
\pgfpathellipse{\pgfpoint{6.000000\du}{9.000000\du}}{\pgfpoint{0.125000\du}{0\du}}{\pgfpoint{0\du}{0.125000\du}}
\pgfusepath{fill}
\pgfsetlinewidth{0.100000\du}
\pgfsetdash{}{0pt}
\pgfsetdash{}{0pt}
\definecolor{dialinecolor}{rgb}{0.000000, 0.000000, 0.000000}
\pgfsetstrokecolor{dialinecolor}
\pgfpathellipse{\pgfpoint{6.000000\du}{9.000000\du}}{\pgfpoint{0.125000\du}{0\du}}{\pgfpoint{0\du}{0.125000\du}}
\pgfusepath{stroke}
\definecolor{dialinecolor}{rgb}{0.000000, 0.000000, 0.000000}
\pgfsetstrokecolor{dialinecolor}
\node[anchor=west] at (4.100000\du,8.400000\du){$w$};
\definecolor{dialinecolor}{rgb}{0.000000, 0.000000, 0.000000}
\pgfsetstrokecolor{dialinecolor}
\node[anchor=west] at (5.4500000\du,8.300000\du){$w'$};
\pgfsetlinewidth{0.100000\du}
\pgfsetdash{}{0pt}
\pgfsetdash{}{0pt}
\pgfsetbuttcap
{
\definecolor{dialinecolor}{rgb}{0.000000, 0.000000, 0.000000}
\pgfsetfillcolor{dialinecolor}
\definecolor{dialinecolor}{rgb}{0.000000, 0.000000, 0.000000}
\pgfsetstrokecolor{dialinecolor}
\draw (7.500000\du,17.000000\du)--(7.500000\du,16.500000\du);
}
\pgfsetlinewidth{0.100000\du}
\pgfsetdash{}{0pt}
\pgfsetdash{}{0pt}
\pgfsetmiterjoin
\pgfsetbuttcap
{
\definecolor{dialinecolor}{rgb}{0.000000, 0.000000, 0.000000}
\pgfsetfillcolor{dialinecolor}
\pgfsetarrowsend{stealth}
\definecolor{dialinecolor}{rgb}{0.000000, 0.000000, 0.000000}
\pgfsetstrokecolor{dialinecolor}
\pgfpathmoveto{\pgfpoint{11.500000\du}{12.500000\du}}
\pgfpathcurveto{\pgfpoint{11.500000\du}{10.500000\du}}{\pgfpoint{5.250000\du}{12.000000\du}}{\pgfpoint{5.250000\du}{9.000000\du}}
\pgfusepath{stroke}
}
\definecolor{dialinecolor}{rgb}{0.000000, 0.000000, 0.000000}
\pgfsetstrokecolor{dialinecolor}
\node[anchor=west] at (10.500000\du,13.200000\du){$S_{N,\epsilon}$};
\pgfsetlinewidth{0.100000\du}
\pgfsetdash{}{0pt}
\pgfsetdash{}{0pt}
\pgfsetbuttcap
{
\definecolor{dialinecolor}{rgb}{0.000000, 0.000000, 0.000000}
\pgfsetfillcolor{dialinecolor}
\definecolor{dialinecolor}{rgb}{0.000000, 0.000000, 0.000000}
\pgfsetstrokecolor{dialinecolor}
\draw (4.500000\du,10.000000\du)--(4.000000\du,10.000000\du);
}
\definecolor{dialinecolor}{rgb}{0.000000, 0.000000, 0.000000}
\pgfsetfillcolor{dialinecolor}
\pgfpathellipse{\pgfpoint{4.500000\du}{9.000000\du}}{\pgfpoint{0.125000\du}{0\du}}{\pgfpoint{0\du}{0.125000\du}}
\pgfusepath{fill}
\pgfsetlinewidth{0.100000\du}
\pgfsetdash{}{0pt}
\pgfsetdash{}{0pt}
\definecolor{dialinecolor}{rgb}{0.000000, 0.000000, 0.000000}
\pgfsetstrokecolor{dialinecolor}
\pgfpathellipse{\pgfpoint{4.500000\du}{9.000000\du}}{\pgfpoint{0.125000\du}{0\du}}{\pgfpoint{0\du}{0.125000\du}}
\pgfusepath{stroke}
\end{tikzpicture}
\end{subfigure}
\begin{subfigure}[b]{0.45\textwidth}
\centering
\ifx\du\undefined
  \newlength{\du}
\fi
\setlength{\du}{15\unitlength}
\begin{tikzpicture}
\pgftransformxscale{1.000000}
\pgftransformyscale{-1.000000}
\definecolor{dialinecolor}{rgb}{0.000000, 0.000000, 0.000000}
\pgfsetstrokecolor{dialinecolor}
\definecolor{dialinecolor}{rgb}{1.000000, 1.000000, 1.000000}
\pgfsetfillcolor{dialinecolor}
\pgfsetlinewidth{0.100000\du}
\pgfsetdash{}{0pt}
\pgfsetdash{}{0pt}
\pgfsetbuttcap
{
\definecolor{dialinecolor}{rgb}{0.000000, 0.000000, 0.000000}
\pgfsetfillcolor{dialinecolor}
\pgfsetarrowsend{stealth}
\definecolor{dialinecolor}{rgb}{0.000000, 0.000000, 0.000000}
\pgfsetstrokecolor{dialinecolor}
\draw (4.000000\du,17.000000\du)--(17.000000\du,17.000000\du);
}
\pgfsetlinewidth{0.100000\du}
\pgfsetdash{}{0pt}
\pgfsetdash{}{0pt}
\pgfsetbuttcap
{
\definecolor{dialinecolor}{rgb}{0.000000, 0.000000, 0.000000}
\pgfsetfillcolor{dialinecolor}
\pgfsetarrowsend{stealth}
\definecolor{dialinecolor}{rgb}{0.000000, 0.000000, 0.000000}
\pgfsetstrokecolor{dialinecolor}
\draw (4.000000\du,17.000000\du)--(4.000000\du,7.000000\du);
}
\pgfsetlinewidth{0.0500000\du}
\pgfsetdash{{\pgflinewidth}{0.200000\du}}{0cm}
\pgfsetdash{{\pgflinewidth}{0.200000\du}}{0cm}
\pgfsetbuttcap
{
\definecolor{dialinecolor}{rgb}{0.000000, 0.000000, 0.000000}
\pgfsetfillcolor{dialinecolor}
\definecolor{dialinecolor}{rgb}{0.000000, 0.000000, 0.000000}
\pgfsetstrokecolor{dialinecolor}
\draw (7.000000\du,17.000000\du)--(7.000000\du,7.000000\du);
}
\definecolor{dialinecolor}{rgb}{0.000000, 0.000000, 0.000000}
\pgfsetstrokecolor{dialinecolor}
\node[anchor=west] at (6.000000\du,18.000000\du){$\pi_1(w)$};
\definecolor{dialinecolor}{rgb}{0.000000, 0.000000, 0.000000}
\pgfsetstrokecolor{dialinecolor}
\node[anchor=west] at (5.000000\du,10.000000\du){};
\definecolor{dialinecolor}{rgb}{0.000000, 0.000000, 0.000000}
\pgfsetstrokecolor{dialinecolor}
\node[anchor=west] at (17.000000\du,17.500000\du){$v_1$};
\definecolor{dialinecolor}{rgb}{0.000000, 0.000000, 0.000000}
\pgfsetstrokecolor{dialinecolor}
\node[anchor=west] at (3.000000\du,6.500000\du){$v_2$};
\pgfsetlinewidth{0.0500000\du}
\pgfsetdash{{\pgflinewidth}{0.200000\du}}{0cm}
\pgfsetdash{{\pgflinewidth}{0.200000\du}}{0cm}
\pgfsetbuttcap
{
\definecolor{dialinecolor}{rgb}{0.000000, 0.000000, 0.000000}
\pgfsetfillcolor{dialinecolor}
\definecolor{dialinecolor}{rgb}{0.000000, 0.000000, 0.000000}
\pgfsetstrokecolor{dialinecolor}
\draw (4.000000\du,13.000000\du)--(7.000000\du,13.000000\du);
}
\definecolor{dialinecolor}{rgb}{0.000000, 0.000000, 0.000000}
\pgfsetstrokecolor{dialinecolor}
\definecolor{dialinecolor}{rgb}{0.000000, 0.000000, 0.000000}
\pgfsetstrokecolor{dialinecolor}
\node[anchor=west] at (1.400000\du,13.000000\du){$\pi_2(w)$};
\pgfsetlinewidth{0.0500000\du}
\pgfsetdash{{\pgflinewidth}{0.200000\du}}{0cm}
\pgfsetdash{{\pgflinewidth}{0.200000\du}}{0cm}
\pgfsetbuttcap
{
\definecolor{dialinecolor}{rgb}{0.000000, 0.000000, 0.000000}
\pgfsetfillcolor{dialinecolor}
\definecolor{dialinecolor}{rgb}{0.000000, 0.000000, 0.000000}
\pgfsetstrokecolor{dialinecolor}
\draw (7.000000\du,13.000000\du)--(13.500000\du,6.500000\du);
}
\pgfsetlinewidth{0.0500000\du}
\pgfsetdash{{\pgflinewidth}{0.200000\du}}{0cm}
\pgfsetdash{{\pgflinewidth}{0.200000\du}}{0cm}
\pgfsetbuttcap
{
\definecolor{dialinecolor}{rgb}{0.000000, 0.000000, 0.000000}
\pgfsetfillcolor{dialinecolor}
\definecolor{dialinecolor}{rgb}{0.000000, 0.000000, 0.000000}
\pgfsetstrokecolor{dialinecolor}
\draw (7.000000\du,13.000000\du)--(8.000000\du,13.000000\du);
}
\definecolor{dialinecolor}{rgb}{0.000000, 0.000000, 0.000000}
\pgfsetstrokecolor{dialinecolor}
\node[anchor=west] at (7.700000\du,12.600000\du){$45^o$};
\definecolor{dialinecolor}{rgb}{0.000000, 0.000000, 0.000000}
\pgfsetstrokecolor{dialinecolor}
\node[anchor=west] at (1.400000\du,9.000000\du){$\pi_2(w')$};
\definecolor{dialinecolor}{rgb}{0.000000, 0.000000, 0.000000}
\pgfsetstrokecolor{dialinecolor}
\node[anchor=west] at (10.000000\du,18.000000\du){$\pi_1(w')$};
\pgfsetlinewidth{0.0500000\du}
\pgfsetdash{{\pgflinewidth}{0.200000\du}}{0cm}
\pgfsetdash{{\pgflinewidth}{0.200000\du}}{0cm}
\pgfsetbuttcap
{
\definecolor{dialinecolor}{rgb}{0.000000, 0.000000, 0.000000}
\pgfsetfillcolor{dialinecolor}
\definecolor{dialinecolor}{rgb}{0.000000, 0.000000, 0.000000}
\pgfsetstrokecolor{dialinecolor}
\draw (4.000000\du,9.000000\du)--(11.000000\du,9.000000\du);
}
\pgfsetlinewidth{0.0500000\du}
\pgfsetdash{{\pgflinewidth}{0.200000\du}}{0cm}
\pgfsetdash{{\pgflinewidth}{0.200000\du}}{0cm}
\pgfsetbuttcap
{
\definecolor{dialinecolor}{rgb}{0.000000, 0.000000, 0.000000}
\pgfsetfillcolor{dialinecolor}
\definecolor{dialinecolor}{rgb}{0.000000, 0.000000, 0.000000}
\pgfsetstrokecolor{dialinecolor}
\draw (11.000000\du,17.000000\du)--(11.000000\du,7.000000\du);
}
\definecolor{dialinecolor}{rgb}{0.000000, 0.000000, 0.000000}
\pgfsetstrokecolor{dialinecolor}
\node[anchor=west] at (13.000000\du,17.900000\du){$H\cdot \epsilon$};
\pgfsetlinewidth{0.100000\du}
\pgfsetdash{}{0pt}
\pgfsetdash{}{0pt}
\pgfsetbuttcap
{
\definecolor{dialinecolor}{rgb}{0.000000, 0.000000, 0.000000}
\pgfsetfillcolor{dialinecolor}
\definecolor{dialinecolor}{rgb}{0.000000, 0.000000, 0.000000}
\pgfsetstrokecolor{dialinecolor}
\draw (13.500000\du,17.000000\du)--(13.500000\du,16.500000\du);
}
\end{tikzpicture}
\end{subfigure}
\caption{As the value of player $W$ moves in the set $S_{n,\epsilon}$ (depicted as the thick vertical line in the right figure), the thresholds of player $V$ for the two items, must be changing in such a way, that their point of intersection moves only along some fixed $45^o$ line that depends only on $N$.}\label{fig:const1}
\end{figure}
\end{lemma}
\begin{proof}
Consider any two points $w, w' \in S_{N,\epsilon}$. We show that any mechanism that gets an $H$ approximation has $\pi_2(w) - \pi_1(w) = \pi_2(w') -  \pi_1(w')$ for all such $w$ and $w'$. 
\begin{enumerate}
\item Suppose not. Assume without loss of generality that $\pi_2(w) - \pi_1(w) < \pi_2(w') - \pi_1(w')$.
\item Note that $\max\{\pi_1(w), \pi_1(w')\} < H\epsilon$ from \prettyref{lem:item2_always_allocated}. 
 This and bullet-point 1 above imply that there exists a point $v$ s.t. $H\epsilon > v_1 > \max(\pi_1(w), \pi_1(w'))$ and, $\pi_2(w) - \pi_1(w) < v_2 - v_1 < \pi_2(w') - \pi_1(w')$ (point $A$ in Figure \ref{fig:impos1}). 

\item Such a point $v$ is included in $B_R^V(w')$ and in $T_R^V(w)$. Based on our characterization in Corollary~\ref{cor:AllocationRegions}
such a point $v$ gets allocated the set of items $\{1,2\}$ (i.e., both the items) when bidder $W$ bids $w$ and just $\{1\}$ when bidder $W$ bids $w'$. 
\begin{itemize}
\item The astute reader might point out ``Corollary~\ref{cor:AllocationRegions} says that the allocation in region $B_R^V(w')$ could be $\{1\}$ or $\{1,2\}$, and it is $\{1\}$ whenever $\pi_2(w') > \pi_1(w')$. Do we know that $\pi_2(w') > \pi_1(w')$ to draw this unambiguous conclusion about allocation being $\{1\}$ in $B_R^V(w')$? The answer is yes. To see this, note that by Lemma~\ref{lem:TallRectangle} $\pi_2 \geq \pi_1$ at both $w$ and $w'$. The only question is whether $\pi_2(w') > \pi_1(w')$ or not. But w.l.o.g. we assumed that $\pi_2(w) - \pi_1(w) < \pi_2(w') - \pi_1(w')$. This immediately shows that $\pi_2(w') > \pi_1(w')$. 
\end{itemize}

\item Correspondingly, the allocation for bidder $W$ when changing his bid from $w$ to $w'$ changes from $\emptyset$ to just $\{2\}$ or $\emptyset$ to $\emptyset$. To see this note that when $V$ gets $\{1,2\}$ clearly $W$ gets $\emptyset$. But when $V$ gets $\{1\}$, the bidder $W$ could have possibly gotten $\{2\}$ or $\emptyset$. Note that $W$ getting $\emptyset$ means that item $2$ didn't go to any bidder at all. 
\begin{figure}[h]
\centering
\begin{subfigure}[b]{0.45\textwidth}
\centering
\ifx\du\undefined
  \newlength{\du}
\fi
\setlength{\du}{15\unitlength}
\begin{tikzpicture}
\pgftransformxscale{1.000000}
\pgftransformyscale{-1.000000}
\definecolor{dialinecolor}{rgb}{0.000000, 0.000000, 0.000000}
\pgfsetstrokecolor{dialinecolor}
\definecolor{dialinecolor}{rgb}{1.000000, 1.000000, 1.000000}
\pgfsetfillcolor{dialinecolor}
\pgfsetlinewidth{0.100000\du}
\pgfsetdash{}{0pt}
\pgfsetdash{}{0pt}
\pgfsetbuttcap
{
\definecolor{dialinecolor}{rgb}{0.000000, 0.000000, 0.000000}
\pgfsetfillcolor{dialinecolor}
\pgfsetarrowsend{stealth}
\definecolor{dialinecolor}{rgb}{0.000000, 0.000000, 0.000000}
\pgfsetstrokecolor{dialinecolor}
\draw (4.000000\du,17.000000\du)--(17.000000\du,17.000000\du);
}
\pgfsetlinewidth{0.100000\du}
\pgfsetdash{}{0pt}
\pgfsetdash{}{0pt}
\pgfsetbuttcap
{
\definecolor{dialinecolor}{rgb}{0.000000, 0.000000, 0.000000}
\pgfsetfillcolor{dialinecolor}
\pgfsetarrowsend{stealth}
\definecolor{dialinecolor}{rgb}{0.000000, 0.000000, 0.000000}
\pgfsetstrokecolor{dialinecolor}
\draw (4.000000\du,17.000000\du)--(4.000000\du,7.000000\du);
}
\pgfsetlinewidth{0.0500000\du}
\pgfsetdash{{\pgflinewidth}{0.200000\du}}{0cm}
\pgfsetdash{{\pgflinewidth}{0.200000\du}}{0cm}
\pgfsetbuttcap
{
\definecolor{dialinecolor}{rgb}{0.000000, 0.000000, 0.000000}
\pgfsetfillcolor{dialinecolor}
\definecolor{dialinecolor}{rgb}{0.000000, 0.000000, 0.000000}
\pgfsetstrokecolor{dialinecolor}
\draw (6.500000\du,17.000000\du)--(6.500000\du,7.000000\du);
}
\definecolor{dialinecolor}{rgb}{0.000000, 0.000000, 0.000000}
\pgfsetstrokecolor{dialinecolor}
\node[anchor=west] at (5.000000\du,10.000000\du){};
\definecolor{dialinecolor}{rgb}{0.000000, 0.000000, 0.000000}
\pgfsetstrokecolor{dialinecolor}
\node[anchor=west] at (17.000000\du,17.500000\du){$w_1$};
\definecolor{dialinecolor}{rgb}{0.000000, 0.000000, 0.000000}
\pgfsetstrokecolor{dialinecolor}
\node[anchor=west] at (3.000000\du,6.500000\du){$w_2$};
\pgfsetlinewidth{0.150000\du}
\pgfsetdash{}{0pt}
\pgfsetdash{}{0pt}
\pgfsetbuttcap
{
\definecolor{dialinecolor}{rgb}{0.000000, 0.000000, 0.000000}
\pgfsetfillcolor{dialinecolor}
\definecolor{dialinecolor}{rgb}{0.000000, 0.000000, 0.000000}
\pgfsetstrokecolor{dialinecolor}
\draw (4.000000\du,9.000000\du)--(6.500000\du,9.000000\du);
}
\definecolor{dialinecolor}{rgb}{0.000000, 0.000000, 0.000000}
\pgfsetstrokecolor{dialinecolor}
\node[anchor=west] at (5.900000\du,18.000000\du){$\epsilon$};
\definecolor{dialinecolor}{rgb}{0.000000, 0.000000, 0.000000}
\pgfsetstrokecolor{dialinecolor}
\node[anchor=west] at (2.800000\du,9.000000\du){$N$};
\definecolor{dialinecolor}{rgb}{0.000000, 0.000000, 0.000000}
\pgfsetstrokecolor{dialinecolor}
\node[anchor=west] at (7.000000\du,18.000000\du){$1/H$};
\definecolor{dialinecolor}{rgb}{0.000000, 0.000000, 0.000000}
\pgfsetstrokecolor{dialinecolor}
\node[anchor=west] at (2.8000000\du,10.000000\du){$H$};
\definecolor{dialinecolor}{rgb}{0.000000, 0.000000, 0.000000}
\pgfsetfillcolor{dialinecolor}
\pgfpathellipse{\pgfpoint{6.000000\du}{9.000000\du}}{\pgfpoint{0.125000\du}{0\du}}{\pgfpoint{0\du}{0.125000\du}}
\pgfusepath{fill}
\pgfsetlinewidth{0.100000\du}
\pgfsetdash{}{0pt}
\pgfsetdash{}{0pt}
\definecolor{dialinecolor}{rgb}{0.000000, 0.000000, 0.000000}
\pgfsetstrokecolor{dialinecolor}
\pgfpathellipse{\pgfpoint{6.000000\du}{9.000000\du}}{\pgfpoint{0.125000\du}{0\du}}{\pgfpoint{0\du}{0.125000\du}}
\pgfusepath{stroke}
\definecolor{dialinecolor}{rgb}{0.000000, 0.000000, 0.000000}
\pgfsetstrokecolor{dialinecolor}
\node[anchor=west] at (4.100000\du,8.400000\du){$w$};
\definecolor{dialinecolor}{rgb}{0.000000, 0.000000, 0.000000}
\pgfsetstrokecolor{dialinecolor}
\node[anchor=west] at (5.4500000\du,8.300000\du){$w'$};
\pgfsetlinewidth{0.100000\du}
\pgfsetdash{}{0pt}
\pgfsetdash{}{0pt}
\pgfsetbuttcap
{
\definecolor{dialinecolor}{rgb}{0.000000, 0.000000, 0.000000}
\pgfsetfillcolor{dialinecolor}
\definecolor{dialinecolor}{rgb}{0.000000, 0.000000, 0.000000}
\pgfsetstrokecolor{dialinecolor}
\draw (7.500000\du,17.000000\du)--(7.500000\du,16.500000\du);
}
\pgfsetlinewidth{0.100000\du}
\pgfsetdash{}{0pt}
\pgfsetdash{}{0pt}
\pgfsetmiterjoin
\pgfsetbuttcap
{
\definecolor{dialinecolor}{rgb}{0.000000, 0.000000, 0.000000}
\pgfsetfillcolor{dialinecolor}
\pgfsetarrowsend{stealth}
\definecolor{dialinecolor}{rgb}{0.000000, 0.000000, 0.000000}
\pgfsetstrokecolor{dialinecolor}
\pgfpathmoveto{\pgfpoint{11.500000\du}{12.500000\du}}
\pgfpathcurveto{\pgfpoint{11.500000\du}{10.500000\du}}{\pgfpoint{5.250000\du}{12.000000\du}}{\pgfpoint{5.250000\du}{9.000000\du}}
\pgfusepath{stroke}
}
\definecolor{dialinecolor}{rgb}{0.000000, 0.000000, 0.000000}
\pgfsetstrokecolor{dialinecolor}
\node[anchor=west] at (10.500000\du,13.200000\du){$S_{N,\epsilon}$};
\pgfsetlinewidth{0.100000\du}
\pgfsetdash{}{0pt}
\pgfsetdash{}{0pt}
\pgfsetbuttcap
{
\definecolor{dialinecolor}{rgb}{0.000000, 0.000000, 0.000000}
\pgfsetfillcolor{dialinecolor}
\definecolor{dialinecolor}{rgb}{0.000000, 0.000000, 0.000000}
\pgfsetstrokecolor{dialinecolor}
\draw (4.500000\du,10.000000\du)--(4.000000\du,10.000000\du);
}
\definecolor{dialinecolor}{rgb}{0.000000, 0.000000, 0.000000}
\pgfsetfillcolor{dialinecolor}
\pgfpathellipse{\pgfpoint{4.500000\du}{9.000000\du}}{\pgfpoint{0.125000\du}{0\du}}{\pgfpoint{0\du}{0.125000\du}}
\pgfusepath{fill}
\pgfsetlinewidth{0.100000\du}
\pgfsetdash{}{0pt}
\pgfsetdash{}{0pt}
\definecolor{dialinecolor}{rgb}{0.000000, 0.000000, 0.000000}
\pgfsetstrokecolor{dialinecolor}
\pgfpathellipse{\pgfpoint{4.500000\du}{9.000000\du}}{\pgfpoint{0.125000\du}{0\du}}{\pgfpoint{0\du}{0.125000\du}}
\pgfusepath{stroke}
\node[anchor=west] at (9.000000\du,8.000000\du){$x^W(A,w)=\{\emptyset\}$};
\node[anchor=west] at (9.000000\du,9.000000\du){$x^W(A,w')=\{2\}$};
\end{tikzpicture}
\end{subfigure}
\begin{subfigure}[b]{0.45\textwidth}
\centering
\ifx\du\undefined
  \newlength{\du}
\fi
\setlength{\du}{15\unitlength}
\begin{tikzpicture}
\pgftransformxscale{1.000000}
\pgftransformyscale{-1.000000}
\definecolor{dialinecolor}{rgb}{0.000000, 0.000000, 0.000000}
\pgfsetstrokecolor{dialinecolor}
\definecolor{dialinecolor}{rgb}{1.000000, 1.000000, 1.000000}
\pgfsetfillcolor{dialinecolor}
\pgfsetlinewidth{0.100000\du}
\pgfsetdash{}{0pt}
\pgfsetdash{}{0pt}
\pgfsetbuttcap
{
\definecolor{dialinecolor}{rgb}{0.000000, 0.000000, 0.000000}
\pgfsetfillcolor{dialinecolor}
\pgfsetarrowsend{stealth}
\definecolor{dialinecolor}{rgb}{0.000000, 0.000000, 0.000000}
\pgfsetstrokecolor{dialinecolor}
\draw (4.000000\du,17.000000\du)--(17.000000\du,17.000000\du);
}
\pgfsetlinewidth{0.100000\du}
\pgfsetdash{}{0pt}
\pgfsetdash{}{0pt}
\pgfsetbuttcap
{
\definecolor{dialinecolor}{rgb}{0.000000, 0.000000, 0.000000}
\pgfsetfillcolor{dialinecolor}
\pgfsetarrowsend{stealth}
\definecolor{dialinecolor}{rgb}{0.000000, 0.000000, 0.000000}
\pgfsetstrokecolor{dialinecolor}
\draw (4.000000\du,17.000000\du)--(4.000000\du,7.000000\du);
}
\pgfsetlinewidth{0.0500000\du}
\pgfsetdash{{\pgflinewidth}{0.200000\du}}{0cm}
\pgfsetdash{{\pgflinewidth}{0.200000\du}}{0cm}
\pgfsetbuttcap
{
\definecolor{dialinecolor}{rgb}{0.000000, 0.000000, 0.000000}
\pgfsetfillcolor{dialinecolor}
\definecolor{dialinecolor}{rgb}{0.000000, 0.000000, 0.000000}
\pgfsetstrokecolor{dialinecolor}
\draw (7.000000\du,17.000000\du)--(7.000000\du,7.000000\du);
}
\definecolor{dialinecolor}{rgb}{0.000000, 0.000000, 0.000000}
\pgfsetstrokecolor{dialinecolor}
\node[anchor=west] at (6.000000\du,18.000000\du){$\pi_1(w)$};
\definecolor{dialinecolor}{rgb}{0.000000, 0.000000, 0.000000}
\pgfsetstrokecolor{dialinecolor}
\node[anchor=west] at (5.000000\du,10.000000\du){};
\definecolor{dialinecolor}{rgb}{0.000000, 0.000000, 0.000000}
\pgfsetstrokecolor{dialinecolor}
\node[anchor=west] at (17.000000\du,17.500000\du){$v_1$};
\definecolor{dialinecolor}{rgb}{0.000000, 0.000000, 0.000000}
\pgfsetstrokecolor{dialinecolor}
\node[anchor=west] at (3.000000\du,6.500000\du){$v_2$};
\pgfsetlinewidth{0.0500000\du}
\pgfsetdash{{\pgflinewidth}{0.200000\du}}{0cm}
\pgfsetdash{{\pgflinewidth}{0.200000\du}}{0cm}
\pgfsetbuttcap
{
\definecolor{dialinecolor}{rgb}{0.000000, 0.000000, 0.000000}
\pgfsetfillcolor{dialinecolor}
\definecolor{dialinecolor}{rgb}{0.000000, 0.000000, 0.000000}
\pgfsetstrokecolor{dialinecolor}
\draw (4.000000\du,13.000000\du)--(7.000000\du,13.000000\du);
}
\definecolor{dialinecolor}{rgb}{0.000000, 0.000000, 0.000000}
\pgfsetstrokecolor{dialinecolor}
\definecolor{dialinecolor}{rgb}{0.000000, 0.000000, 0.000000}
\pgfsetstrokecolor{dialinecolor}
\node[anchor=west] at (1.500000\du,13.000000\du){$\pi_2(w)$};
\pgfsetlinewidth{0.0500000\du}
\pgfsetdash{{\pgflinewidth}{0.200000\du}}{0cm}
\pgfsetdash{{\pgflinewidth}{0.200000\du}}{0cm}
\pgfsetbuttcap
{
\definecolor{dialinecolor}{rgb}{0.000000, 0.000000, 0.000000}
\pgfsetfillcolor{dialinecolor}
\definecolor{dialinecolor}{rgb}{0.000000, 0.000000, 0.000000}
\pgfsetstrokecolor{dialinecolor}
\draw (7.000000\du,13.000000\du)--(13.500000\du,6.500000\du);
}
\pgfsetlinewidth{0.0500000\du}
\pgfsetdash{{\pgflinewidth}{0.200000\du}}{0cm}
\pgfsetdash{{\pgflinewidth}{0.200000\du}}{0cm}
\pgfsetbuttcap
{
\definecolor{dialinecolor}{rgb}{0.000000, 0.000000, 0.000000}
\pgfsetfillcolor{dialinecolor}
\definecolor{dialinecolor}{rgb}{0.000000, 0.000000, 0.000000}
\pgfsetstrokecolor{dialinecolor}
\draw (7.000000\du,13.000000\du)--(8.000000\du,13.000000\du);
}
\definecolor{dialinecolor}{rgb}{0.000000, 0.000000, 0.000000}
\pgfsetstrokecolor{dialinecolor}
\node[anchor=west] at (7.600000\du,12.600000\du){$45^o$};
\definecolor{dialinecolor}{rgb}{0.000000, 0.000000, 0.000000}
\pgfsetstrokecolor{dialinecolor}
\node[anchor=west] at (1.500000\du,9.000000\du){$\pi_2(w')$};
\definecolor{dialinecolor}{rgb}{0.000000, 0.000000, 0.000000}
\pgfsetstrokecolor{dialinecolor}
\node[anchor=west] at (8.5000000\du,18.000000\du){$\pi_1(w')$};
\pgfsetlinewidth{0.0500000\du}
\pgfsetdash{{\pgflinewidth}{0.200000\du}}{0cm}
\pgfsetdash{{\pgflinewidth}{0.200000\du}}{0cm}
\pgfsetbuttcap
{
\definecolor{dialinecolor}{rgb}{0.000000, 0.000000, 0.000000}
\pgfsetfillcolor{dialinecolor}
\definecolor{dialinecolor}{rgb}{0.000000, 0.000000, 0.000000}
\pgfsetstrokecolor{dialinecolor}
\draw (4.000000\du,9.000000\du)--(9.000000\du,9.000000\du);
}
\pgfsetlinewidth{0.0500000\du}
\pgfsetdash{{\pgflinewidth}{0.200000\du}}{0cm}
\pgfsetdash{{\pgflinewidth}{0.200000\du}}{0cm}
\pgfsetbuttcap
{
\definecolor{dialinecolor}{rgb}{0.000000, 0.000000, 0.000000}
\pgfsetfillcolor{dialinecolor}
\definecolor{dialinecolor}{rgb}{0.000000, 0.000000, 0.000000}
\pgfsetstrokecolor{dialinecolor}
\draw (9.000000\du,17.000000\du)--(9.000000\du,7.000000\du);
}
\definecolor{dialinecolor}{rgb}{0.000000, 0.000000, 0.000000}
\pgfsetstrokecolor{dialinecolor}
\node[anchor=west] at (13.000000\du,17.900000\du){$H\cdot \epsilon$};
\pgfsetlinewidth{0.100000\du}
\pgfsetdash{}{0pt}
\pgfsetdash{}{0pt}
\pgfsetbuttcap
{
\definecolor{dialinecolor}{rgb}{0.000000, 0.000000, 0.000000}
\pgfsetfillcolor{dialinecolor}
\definecolor{dialinecolor}{rgb}{0.000000, 0.000000, 0.000000}
\pgfsetstrokecolor{dialinecolor}
\draw (13.500000\du,17.000000\du)--(13.500000\du,16.500000\du);
}
\pgfsetlinewidth{0.0500000\du}
\pgfsetdash{{\pgflinewidth}{0.200000\du}}{0cm}
\pgfsetdash{{\pgflinewidth}{0.200000\du}}{0cm}
\pgfsetbuttcap
{
\definecolor{dialinecolor}{rgb}{0.000000, 0.000000, 0.000000}
\pgfsetfillcolor{dialinecolor}
\definecolor{dialinecolor}{rgb}{0.000000, 0.000000, 0.000000}
\pgfsetstrokecolor{dialinecolor}
\draw (9.000000\du,9.000000\du)--(12.000000\du,6.000000\du);
}
\definecolor{dialinecolor}{rgb}{0.000000, 0.000000, 0.000000}
\pgfsetfillcolor{dialinecolor}
\pgfpathellipse{\pgfpoint{11.500000\du}{7.500000\du}}{\pgfpoint{0.125000\du}{0\du}}{\pgfpoint{0\du}{0.125000\du}}
\pgfusepath{fill}
\pgfsetlinewidth{0.100000\du}
\pgfsetdash{}{0pt}
\pgfsetdash{}{0pt}
\definecolor{dialinecolor}{rgb}{0.000000, 0.000000, 0.000000}
\pgfsetstrokecolor{dialinecolor}
\pgfpathellipse{\pgfpoint{11.500000\du}{7.500000\du}}{\pgfpoint{0.125000\du}{0\du}}{\pgfpoint{0\du}{0.125000\du}}
\pgfusepath{stroke}
\definecolor{dialinecolor}{rgb}{0.000000, 0.000000, 0.000000}
\pgfsetstrokecolor{dialinecolor}
\node[anchor=west] at (11.500000\du,7.000000\du){A};
\definecolor{dialinecolor}{rgb}{0.000000, 0.000000, 0.000000}
\pgfsetstrokecolor{dialinecolor}
\node[anchor=west] at (13.000000\du,8.500000\du){$x^V(A,w)=\{1,2\}$};
\definecolor{dialinecolor}{rgb}{0.000000, 0.000000, 0.000000}
\pgfsetstrokecolor{dialinecolor}
\node[anchor=west] at (13.000000\du,9.500000\du){$x^V(A,w')=\{1\}$};
\end{tikzpicture}
\end{subfigure}
\caption{If the intersection point in the allocation of player $V$ (right) does not move along the diagonal, there exists a point $A$ such that the allocation of player $W$ (left), when player $V$ has value $A$, moves from $\emptyset$ to Item $2$. Such a transition in $W$'s allocation is not feasible as we move vertically.}\label{fig:impos1}
\end{figure}
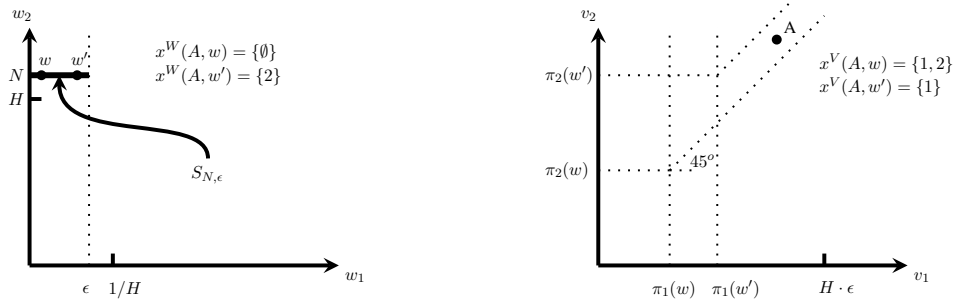
\item But the possibility of item $2$ not going to any bidder at all is ruled out by \prettyref{lem:item2_always_allocated}: note that $w_1' < \epsilon$ and $v_1 < H\epsilon$.  Given that item $2$ should always be allocated, the only possible change of allocation for bidder $W$ when changing his bids from $w$ to $w'$ is from $\emptyset$ to $\{2\}$. 

\item Based on the definition of the $4$ sets it follows immediately that when $w_2$ is held a constant and $w_1$ is varied,  the point $w$ moves either between $B_L^W$ and $B_R^W$ or between $T_L^W$ and $T_R^W$, or between $T_L^W$ and $B_R^W$. \item Our characterization in Corollary~\ref{cor:AllocationRegions} shows that none of the three moves mentioned above can result in an allocation change between $\emptyset$ and $\{2\}$. Indeed an allocation change between $\emptyset$ and $\{2\}$ requires a move between $B_L^W$ and $T_L^W$, which is not among the three moves mentioned above. 
\end{enumerate}
\end{proof}

\begin{lemma}
\label{lem:PisAreFrozen}
For every deterministic DSIC+IR mechanism that gets an $H$-approximation,
\begin{align}
\label{eq:PisAreFrozen}
&\forall N \geq H, \forall \epsilon < \frac{1}{H},\forall w \in S_{N,\epsilon} = \{w \in \R_+^2: 0 < w_1 < \epsilon, w_2 = N\},\text{ we have: }\notag\\
&\qquad\qquad\qquad\qquad\pi_1(w) = c_{1,N},\ and\ \pi_2(w) = c_{2,N}.
\end{align}
\begin{figure}[h]
\centering
\begin{subfigure}[b]{0.45\textwidth}
\centering
\ifx\du\undefined
  \newlength{\du}
\fi
\setlength{\du}{15\unitlength}
\begin{tikzpicture}
\pgftransformxscale{1.000000}
\pgftransformyscale{-1.000000}
\definecolor{dialinecolor}{rgb}{0.000000, 0.000000, 0.000000}
\pgfsetstrokecolor{dialinecolor}
\definecolor{dialinecolor}{rgb}{1.000000, 1.000000, 1.000000}
\pgfsetfillcolor{dialinecolor}
\pgfsetlinewidth{0.100000\du}
\pgfsetdash{}{0pt}
\pgfsetdash{}{0pt}
\pgfsetbuttcap
{
\definecolor{dialinecolor}{rgb}{0.000000, 0.000000, 0.000000}
\pgfsetfillcolor{dialinecolor}
\pgfsetarrowsend{stealth}
\definecolor{dialinecolor}{rgb}{0.000000, 0.000000, 0.000000}
\pgfsetstrokecolor{dialinecolor}
\draw (4.000000\du,17.000000\du)--(17.000000\du,17.000000\du);
}
\pgfsetlinewidth{0.100000\du}
\pgfsetdash{}{0pt}
\pgfsetdash{}{0pt}
\pgfsetbuttcap
{
\definecolor{dialinecolor}{rgb}{0.000000, 0.000000, 0.000000}
\pgfsetfillcolor{dialinecolor}
\pgfsetarrowsend{stealth}
\definecolor{dialinecolor}{rgb}{0.000000, 0.000000, 0.000000}
\pgfsetstrokecolor{dialinecolor}
\draw (4.000000\du,17.000000\du)--(4.000000\du,7.000000\du);
}
\pgfsetlinewidth{0.0500000\du}
\pgfsetdash{{\pgflinewidth}{0.200000\du}}{0cm}
\pgfsetdash{{\pgflinewidth}{0.200000\du}}{0cm}
\pgfsetbuttcap
{
\definecolor{dialinecolor}{rgb}{0.000000, 0.000000, 0.000000}
\pgfsetfillcolor{dialinecolor}
\definecolor{dialinecolor}{rgb}{0.000000, 0.000000, 0.000000}
\pgfsetstrokecolor{dialinecolor}
\draw (6.500000\du,17.000000\du)--(6.500000\du,7.000000\du);
}
\definecolor{dialinecolor}{rgb}{0.000000, 0.000000, 0.000000}
\pgfsetstrokecolor{dialinecolor}
\node[anchor=west] at (5.000000\du,10.000000\du){};
\definecolor{dialinecolor}{rgb}{0.000000, 0.000000, 0.000000}
\pgfsetstrokecolor{dialinecolor}
\node[anchor=west] at (17.000000\du,17.500000\du){$w_1$};
\definecolor{dialinecolor}{rgb}{0.000000, 0.000000, 0.000000}
\pgfsetstrokecolor{dialinecolor}
\node[anchor=west] at (3.000000\du,6.500000\du){$w_2$};
\pgfsetlinewidth{0.150000\du}
\pgfsetdash{}{0pt}
\pgfsetdash{}{0pt}
\pgfsetbuttcap
{
\definecolor{dialinecolor}{rgb}{0.000000, 0.000000, 0.000000}
\pgfsetfillcolor{dialinecolor}
\definecolor{dialinecolor}{rgb}{0.000000, 0.000000, 0.000000}
\pgfsetstrokecolor{dialinecolor}
\draw (4.000000\du,9.000000\du)--(6.500000\du,9.000000\du);
}
\definecolor{dialinecolor}{rgb}{0.000000, 0.000000, 0.000000}
\pgfsetstrokecolor{dialinecolor}
\node[anchor=west] at (5.900000\du,18.000000\du){$\epsilon$};
\definecolor{dialinecolor}{rgb}{0.000000, 0.000000, 0.000000}
\pgfsetstrokecolor{dialinecolor}
\node[anchor=west] at (2.800000\du,9.000000\du){$N$};
\definecolor{dialinecolor}{rgb}{0.000000, 0.000000, 0.000000}
\pgfsetstrokecolor{dialinecolor}
\node[anchor=west] at (7.000000\du,18.000000\du){$1/H$};
\definecolor{dialinecolor}{rgb}{0.000000, 0.000000, 0.000000}
\pgfsetstrokecolor{dialinecolor}
\node[anchor=west] at (2.8000000\du,10.000000\du){$H$};
\definecolor{dialinecolor}{rgb}{0.000000, 0.000000, 0.000000}
\pgfsetfillcolor{dialinecolor}
\pgfpathellipse{\pgfpoint{6.000000\du}{9.000000\du}}{\pgfpoint{0.125000\du}{0\du}}{\pgfpoint{0\du}{0.125000\du}}
\pgfusepath{fill}
\pgfsetlinewidth{0.100000\du}
\pgfsetdash{}{0pt}
\pgfsetdash{}{0pt}
\definecolor{dialinecolor}{rgb}{0.000000, 0.000000, 0.000000}
\pgfsetstrokecolor{dialinecolor}
\pgfpathellipse{\pgfpoint{6.000000\du}{9.000000\du}}{\pgfpoint{0.125000\du}{0\du}}{\pgfpoint{0\du}{0.125000\du}}
\pgfusepath{stroke}
\definecolor{dialinecolor}{rgb}{0.000000, 0.000000, 0.000000}
\pgfsetstrokecolor{dialinecolor}
\node[anchor=west] at (4.100000\du,8.400000\du){$w$};
\definecolor{dialinecolor}{rgb}{0.000000, 0.000000, 0.000000}
\pgfsetstrokecolor{dialinecolor}
\node[anchor=west] at (5.4500000\du,8.300000\du){$w'$};
\pgfsetlinewidth{0.100000\du}
\pgfsetdash{}{0pt}
\pgfsetdash{}{0pt}
\pgfsetbuttcap
{
\definecolor{dialinecolor}{rgb}{0.000000, 0.000000, 0.000000}
\pgfsetfillcolor{dialinecolor}
\definecolor{dialinecolor}{rgb}{0.000000, 0.000000, 0.000000}
\pgfsetstrokecolor{dialinecolor}
\draw (7.500000\du,17.000000\du)--(7.500000\du,16.500000\du);
}
\pgfsetlinewidth{0.100000\du}
\pgfsetdash{}{0pt}
\pgfsetdash{}{0pt}
\pgfsetmiterjoin
\pgfsetbuttcap
{
\definecolor{dialinecolor}{rgb}{0.000000, 0.000000, 0.000000}
\pgfsetfillcolor{dialinecolor}
\pgfsetarrowsend{stealth}
\definecolor{dialinecolor}{rgb}{0.000000, 0.000000, 0.000000}
\pgfsetstrokecolor{dialinecolor}
\pgfpathmoveto{\pgfpoint{11.500000\du}{12.500000\du}}
\pgfpathcurveto{\pgfpoint{11.500000\du}{10.500000\du}}{\pgfpoint{5.250000\du}{12.000000\du}}{\pgfpoint{5.250000\du}{9.000000\du}}
\pgfusepath{stroke}
}
\definecolor{dialinecolor}{rgb}{0.000000, 0.000000, 0.000000}
\pgfsetstrokecolor{dialinecolor}
\node[anchor=west] at (10.500000\du,13.200000\du){$S_{N,\epsilon}$};
\pgfsetlinewidth{0.100000\du}
\pgfsetdash{}{0pt}
\pgfsetdash{}{0pt}
\pgfsetbuttcap
{
\definecolor{dialinecolor}{rgb}{0.000000, 0.000000, 0.000000}
\pgfsetfillcolor{dialinecolor}
\definecolor{dialinecolor}{rgb}{0.000000, 0.000000, 0.000000}
\pgfsetstrokecolor{dialinecolor}
\draw (4.500000\du,10.000000\du)--(4.000000\du,10.000000\du);
}
\definecolor{dialinecolor}{rgb}{0.000000, 0.000000, 0.000000}
\pgfsetfillcolor{dialinecolor}
\pgfpathellipse{\pgfpoint{4.500000\du}{9.000000\du}}{\pgfpoint{0.125000\du}{0\du}}{\pgfpoint{0\du}{0.125000\du}}
\pgfusepath{fill}
\pgfsetlinewidth{0.100000\du}
\pgfsetdash{}{0pt}
\pgfsetdash{}{0pt}
\definecolor{dialinecolor}{rgb}{0.000000, 0.000000, 0.000000}
\pgfsetstrokecolor{dialinecolor}
\pgfpathellipse{\pgfpoint{4.500000\du}{9.000000\du}}{\pgfpoint{0.125000\du}{0\du}}{\pgfpoint{0\du}{0.125000\du}}
\pgfusepath{stroke}
\end{tikzpicture}
\end{subfigure}
\begin{subfigure}[b]{0.45\textwidth}
\centering
\ifx\du\undefined
  \newlength{\du}
\fi
\setlength{\du}{15\unitlength}
\begin{tikzpicture}
\pgftransformxscale{1.000000}
\pgftransformyscale{-1.000000}
\definecolor{dialinecolor}{rgb}{0.000000, 0.000000, 0.000000}
\pgfsetstrokecolor{dialinecolor}
\definecolor{dialinecolor}{rgb}{1.000000, 1.000000, 1.000000}
\pgfsetfillcolor{dialinecolor}
\pgfsetlinewidth{0.100000\du}
\pgfsetdash{}{0pt}
\pgfsetdash{}{0pt}
\pgfsetbuttcap
{
\definecolor{dialinecolor}{rgb}{0.000000, 0.000000, 0.000000}
\pgfsetfillcolor{dialinecolor}
\pgfsetarrowsend{stealth}
\definecolor{dialinecolor}{rgb}{0.000000, 0.000000, 0.000000}
\pgfsetstrokecolor{dialinecolor}
\draw (4.000000\du,17.000000\du)--(17.000000\du,17.000000\du);
}
\pgfsetlinewidth{0.100000\du}
\pgfsetdash{}{0pt}
\pgfsetdash{}{0pt}
\pgfsetbuttcap
{
\definecolor{dialinecolor}{rgb}{0.000000, 0.000000, 0.000000}
\pgfsetfillcolor{dialinecolor}
\pgfsetarrowsend{stealth}
\definecolor{dialinecolor}{rgb}{0.000000, 0.000000, 0.000000}
\pgfsetstrokecolor{dialinecolor}
\draw (4.000000\du,17.000000\du)--(4.000000\du,7.000000\du);
}
\pgfsetlinewidth{0.0500000\du}
\pgfsetdash{{\pgflinewidth}{0.200000\du}}{0cm}
\pgfsetdash{{\pgflinewidth}{0.200000\du}}{0cm}
\pgfsetbuttcap
{
\definecolor{dialinecolor}{rgb}{0.000000, 0.000000, 0.000000}
\pgfsetfillcolor{dialinecolor}
\definecolor{dialinecolor}{rgb}{0.000000, 0.000000, 0.000000}
\pgfsetstrokecolor{dialinecolor}
\draw (7.000000\du,17.000000\du)--(7.000000\du,7.000000\du);
}
\definecolor{dialinecolor}{rgb}{0.000000, 0.000000, 0.000000}
\pgfsetstrokecolor{dialinecolor}
\node[anchor=west] at (6.000000\du,18.000000\du){$\pi_1(w)=\pi_1(w')$};
\definecolor{dialinecolor}{rgb}{0.000000, 0.000000, 0.000000}
\pgfsetstrokecolor{dialinecolor}
\node[anchor=west] at (5.000000\du,10.000000\du){};
\definecolor{dialinecolor}{rgb}{0.000000, 0.000000, 0.000000}
\pgfsetstrokecolor{dialinecolor}
\node[anchor=west] at (17.000000\du,17.500000\du){$v_1$};
\definecolor{dialinecolor}{rgb}{0.000000, 0.000000, 0.000000}
\pgfsetstrokecolor{dialinecolor}
\node[anchor=west] at (3.000000\du,6.500000\du){$v_2$};
\pgfsetlinewidth{0.0500000\du}
\pgfsetdash{{\pgflinewidth}{0.200000\du}}{0cm}
\pgfsetdash{{\pgflinewidth}{0.200000\du}}{0cm}
\pgfsetbuttcap
{
\definecolor{dialinecolor}{rgb}{0.000000, 0.000000, 0.000000}
\pgfsetfillcolor{dialinecolor}
\definecolor{dialinecolor}{rgb}{0.000000, 0.000000, 0.000000}
\pgfsetstrokecolor{dialinecolor}
\draw (4.000000\du,13.000000\du)--(7.000000\du,13.000000\du);
}
\definecolor{dialinecolor}{rgb}{0.000000, 0.000000, 0.000000}
\pgfsetstrokecolor{dialinecolor}
\node[anchor=west] at (-1.250000\du,13.100000\du){$\pi_2(w)=\pi_2(w')$};
\pgfsetlinewidth{0.0500000\du}
\pgfsetdash{{\pgflinewidth}{0.200000\du}}{0cm}
\pgfsetdash{{\pgflinewidth}{0.200000\du}}{0cm}
\pgfsetbuttcap
{
\definecolor{dialinecolor}{rgb}{0.000000, 0.000000, 0.000000}
\pgfsetfillcolor{dialinecolor}
\definecolor{dialinecolor}{rgb}{0.000000, 0.000000, 0.000000}
\pgfsetstrokecolor{dialinecolor}
\draw (7.000000\du,13.000000\du)--(13.500000\du,6.500000\du);
}
\pgfsetlinewidth{0.0500000\du}
\pgfsetdash{{\pgflinewidth}{0.200000\du}}{0cm}
\pgfsetdash{{\pgflinewidth}{0.200000\du}}{0cm}
\pgfsetbuttcap
{
\definecolor{dialinecolor}{rgb}{0.000000, 0.000000, 0.000000}
\pgfsetfillcolor{dialinecolor}
\definecolor{dialinecolor}{rgb}{0.000000, 0.000000, 0.000000}
\pgfsetstrokecolor{dialinecolor}
\draw (7.000000\du,13.000000\du)--(8.000000\du,13.000000\du);
}
\definecolor{dialinecolor}{rgb}{0.000000, 0.000000, 0.000000}
\pgfsetstrokecolor{dialinecolor}
\node[anchor=west] at (7.7000000\du,12.500000\du){$45^o$};
\definecolor{dialinecolor}{rgb}{0.000000, 0.000000, 0.000000}
\pgfsetstrokecolor{dialinecolor}
\node[anchor=west] at (13.000000\du,17.900000\du){$H\cdot \epsilon$};
\pgfsetlinewidth{0.100000\du}
\pgfsetdash{}{0pt}
\pgfsetdash{}{0pt}
\pgfsetbuttcap
{
\definecolor{dialinecolor}{rgb}{0.000000, 0.000000, 0.000000}
\pgfsetfillcolor{dialinecolor}
\definecolor{dialinecolor}{rgb}{0.000000, 0.000000, 0.000000}
\pgfsetstrokecolor{dialinecolor}
\draw (13.500000\du,17.000000\du)--(13.500000\du,16.500000\du);
}
\end{tikzpicture}
\end{subfigure}
\caption{The thresholds in the allocation of $V$ do not change as $W$ moves in $S_{N,\epsilon}$, e.g. from $w$ to $w'$.}\label{fig:const2}
\end{figure}
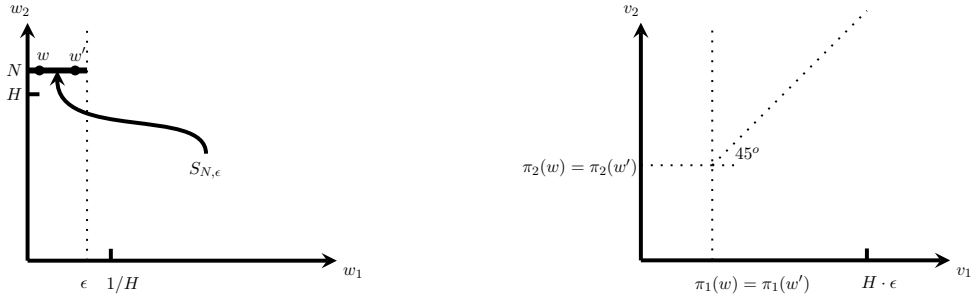
\end{lemma}
\begin{proof}
Consider any two points $w, w' \in S_{N,\epsilon}$. We show that any mechanism that gets an $H$-approximation has $\pi_1(w) =  \pi_1(w')$ and $\pi_2(w) =  \pi_2(w')$  for all such $w$ and $w'$. Suppose not.
\begin{enumerate}
\item  
Lemma~\ref{lem:Pi2MinusPi1IndependentOfW1} says that $\pi_2(w) - \pi_1(w) = \pi_2(w') - \pi_1(w')$. This means that either $\pi_2(w') > \pi_2(w)$ and $\pi_1(w') > \pi_1(w)$, or $\pi_2(w') < \pi_2(w)$ and $\pi_1(w') < \pi_1(w)$. W.l.o.g. we assume that we are in the former case, namely, $\pi_2(w') > \pi_2(w)$ and $\pi_1(w') > \pi_1(w)$.
\item The former statement implies that there exists a point $v \in \R_+^2$ such that  $\pi_1(w) < v_1 < \pi_1(w')$ and $\pi_2(w) < v_2 < \pi_2(w')$, and $v_2 - v_1 > \pi_2(w) - \pi_1(w)$. We also have $\max\{\pi_1(w),\pi_1(w')\} < \epsilon H$ from \prettyref{lem:item2_always_allocated} (point $A$ in Figure \ref{fig:const2}). 
\item Such a point $v$ is included in $T_R^V(w)$ and $B_L^V(w')$. Thus, as the bid of bidder $W$ changes from $w$ to $w'$, by Corollary~\ref{cor:AllocationRegions}, the allocation of bidder $V$ for such a point $v$ changes from $\{1,2\}$ (while in $T_R^V$) to $\emptyset$ (while in $B_L^V$).
\item Correspondingly, the allocation for bidder $W$ when changing his bid from $w$ to $w'$ could have changed from $\emptyset$ to $\emptyset$, or $\emptyset$ to $\{1\}$, or $\emptyset$ to $\{2\}$ or $\emptyset$ to $\{1,2\}$. 
From \Cref{lem:Pi1NotZero,lem:item2_always_allocated}, we should always allocate both items for the bid ranges we have considered. Thus, the only possible allocation change for bidder $W$ as his bid changes from $w$ to $w'$ is from $\emptyset$ to $\{1,2\}$. 
\item Based on the definition of the $4$ sets it follows immediately that 
when $w_2$ is held a constant and $w_1$ is varied,  the point $w$ moves either between $B_L^W$ and $B_R^W$ or between $T_L^W$ and $T_R^W$, or between $T_L^W$ and $B_R^W$. 
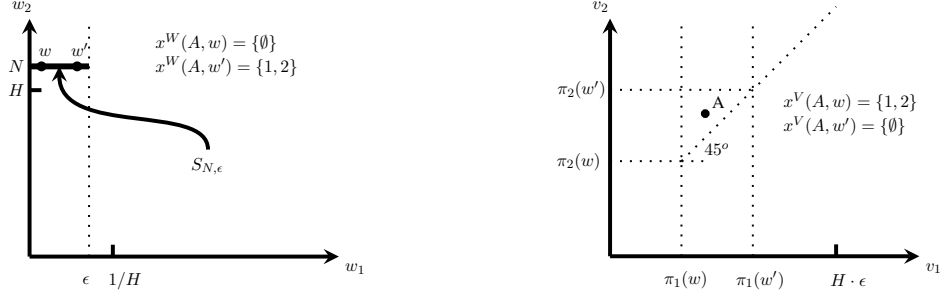
\begin{figure}[h]
\centering
\begin{subfigure}[b]{0.45\textwidth}
\centering
\ifx\du\undefined
  \newlength{\du}
\fi
\setlength{\du}{15\unitlength}
\begin{tikzpicture}
\pgftransformxscale{1.000000}
\pgftransformyscale{-1.000000}
\definecolor{dialinecolor}{rgb}{0.000000, 0.000000, 0.000000}
\pgfsetstrokecolor{dialinecolor}
\definecolor{dialinecolor}{rgb}{1.000000, 1.000000, 1.000000}
\pgfsetfillcolor{dialinecolor}
\pgfsetlinewidth{0.100000\du}
\pgfsetdash{}{0pt}
\pgfsetdash{}{0pt}
\pgfsetbuttcap
{
\definecolor{dialinecolor}{rgb}{0.000000, 0.000000, 0.000000}
\pgfsetfillcolor{dialinecolor}
\pgfsetarrowsend{stealth}
\definecolor{dialinecolor}{rgb}{0.000000, 0.000000, 0.000000}
\pgfsetstrokecolor{dialinecolor}
\draw (4.000000\du,17.000000\du)--(17.000000\du,17.000000\du);
}
\pgfsetlinewidth{0.100000\du}
\pgfsetdash{}{0pt}
\pgfsetdash{}{0pt}
\pgfsetbuttcap
{
\definecolor{dialinecolor}{rgb}{0.000000, 0.000000, 0.000000}
\pgfsetfillcolor{dialinecolor}
\pgfsetarrowsend{stealth}
\definecolor{dialinecolor}{rgb}{0.000000, 0.000000, 0.000000}
\pgfsetstrokecolor{dialinecolor}
\draw (4.000000\du,17.000000\du)--(4.000000\du,7.000000\du);
}
\pgfsetlinewidth{0.0500000\du}
\pgfsetdash{{\pgflinewidth}{0.200000\du}}{0cm}
\pgfsetdash{{\pgflinewidth}{0.200000\du}}{0cm}
\pgfsetbuttcap
{
\definecolor{dialinecolor}{rgb}{0.000000, 0.000000, 0.000000}
\pgfsetfillcolor{dialinecolor}
\definecolor{dialinecolor}{rgb}{0.000000, 0.000000, 0.000000}
\pgfsetstrokecolor{dialinecolor}
\draw (6.500000\du,17.000000\du)--(6.500000\du,7.000000\du);
}
\definecolor{dialinecolor}{rgb}{0.000000, 0.000000, 0.000000}
\pgfsetstrokecolor{dialinecolor}
\node[anchor=west] at (5.000000\du,10.000000\du){};
\definecolor{dialinecolor}{rgb}{0.000000, 0.000000, 0.000000}
\pgfsetstrokecolor{dialinecolor}
\node[anchor=west] at (17.000000\du,17.500000\du){$w_1$};
\definecolor{dialinecolor}{rgb}{0.000000, 0.000000, 0.000000}
\pgfsetstrokecolor{dialinecolor}
\node[anchor=west] at (3.000000\du,6.500000\du){$w_2$};
\pgfsetlinewidth{0.150000\du}
\pgfsetdash{}{0pt}
\pgfsetdash{}{0pt}
\pgfsetbuttcap
{
\definecolor{dialinecolor}{rgb}{0.000000, 0.000000, 0.000000}
\pgfsetfillcolor{dialinecolor}
\definecolor{dialinecolor}{rgb}{0.000000, 0.000000, 0.000000}
\pgfsetstrokecolor{dialinecolor}
\draw (4.000000\du,9.000000\du)--(6.500000\du,9.000000\du);
}
\definecolor{dialinecolor}{rgb}{0.000000, 0.000000, 0.000000}
\pgfsetstrokecolor{dialinecolor}
\node[anchor=west] at (5.900000\du,18.000000\du){$\epsilon$};
\definecolor{dialinecolor}{rgb}{0.000000, 0.000000, 0.000000}
\pgfsetstrokecolor{dialinecolor}
\node[anchor=west] at (2.800000\du,9.000000\du){$N$};
\definecolor{dialinecolor}{rgb}{0.000000, 0.000000, 0.000000}
\pgfsetstrokecolor{dialinecolor}
\node[anchor=west] at (7.000000\du,18.000000\du){$1/H$};
\definecolor{dialinecolor}{rgb}{0.000000, 0.000000, 0.000000}
\pgfsetstrokecolor{dialinecolor}
\node[anchor=west] at (2.8000000\du,10.000000\du){$H$};
\definecolor{dialinecolor}{rgb}{0.000000, 0.000000, 0.000000}
\pgfsetfillcolor{dialinecolor}
\pgfpathellipse{\pgfpoint{6.000000\du}{9.000000\du}}{\pgfpoint{0.125000\du}{0\du}}{\pgfpoint{0\du}{0.125000\du}}
\pgfusepath{fill}
\pgfsetlinewidth{0.100000\du}
\pgfsetdash{}{0pt}
\pgfsetdash{}{0pt}
\definecolor{dialinecolor}{rgb}{0.000000, 0.000000, 0.000000}
\pgfsetstrokecolor{dialinecolor}
\pgfpathellipse{\pgfpoint{6.000000\du}{9.000000\du}}{\pgfpoint{0.125000\du}{0\du}}{\pgfpoint{0\du}{0.125000\du}}
\pgfusepath{stroke}
\definecolor{dialinecolor}{rgb}{0.000000, 0.000000, 0.000000}
\pgfsetstrokecolor{dialinecolor}
\node[anchor=west] at (4.100000\du,8.400000\du){$w$};
\definecolor{dialinecolor}{rgb}{0.000000, 0.000000, 0.000000}
\pgfsetstrokecolor{dialinecolor}
\node[anchor=west] at (5.4500000\du,8.300000\du){$w'$};
\pgfsetlinewidth{0.100000\du}
\pgfsetdash{}{0pt}
\pgfsetdash{}{0pt}
\pgfsetbuttcap
{
\definecolor{dialinecolor}{rgb}{0.000000, 0.000000, 0.000000}
\pgfsetfillcolor{dialinecolor}
\definecolor{dialinecolor}{rgb}{0.000000, 0.000000, 0.000000}
\pgfsetstrokecolor{dialinecolor}
\draw (7.500000\du,17.000000\du)--(7.500000\du,16.500000\du);
}
\pgfsetlinewidth{0.100000\du}
\pgfsetdash{}{0pt}
\pgfsetdash{}{0pt}
\pgfsetmiterjoin
\pgfsetbuttcap
{
\definecolor{dialinecolor}{rgb}{0.000000, 0.000000, 0.000000}
\pgfsetfillcolor{dialinecolor}
\pgfsetarrowsend{stealth}
\definecolor{dialinecolor}{rgb}{0.000000, 0.000000, 0.000000}
\pgfsetstrokecolor{dialinecolor}
\pgfpathmoveto{\pgfpoint{11.500000\du}{12.500000\du}}
\pgfpathcurveto{\pgfpoint{11.500000\du}{10.500000\du}}{\pgfpoint{5.250000\du}{12.000000\du}}{\pgfpoint{5.250000\du}{9.000000\du}}
\pgfusepath{stroke}
}
\definecolor{dialinecolor}{rgb}{0.000000, 0.000000, 0.000000}
\pgfsetstrokecolor{dialinecolor}
\node[anchor=west] at (10.500000\du,13.200000\du){$S_{N,\epsilon}$};
\pgfsetlinewidth{0.100000\du}
\pgfsetdash{}{0pt}
\pgfsetdash{}{0pt}
\pgfsetbuttcap
{
\definecolor{dialinecolor}{rgb}{0.000000, 0.000000, 0.000000}
\pgfsetfillcolor{dialinecolor}
\definecolor{dialinecolor}{rgb}{0.000000, 0.000000, 0.000000}
\pgfsetstrokecolor{dialinecolor}
\draw (4.500000\du,10.000000\du)--(4.000000\du,10.000000\du);
}
\definecolor{dialinecolor}{rgb}{0.000000, 0.000000, 0.000000}
\pgfsetfillcolor{dialinecolor}
\pgfpathellipse{\pgfpoint{4.500000\du}{9.000000\du}}{\pgfpoint{0.125000\du}{0\du}}{\pgfpoint{0\du}{0.125000\du}}
\pgfusepath{fill}
\pgfsetlinewidth{0.100000\du}
\pgfsetdash{}{0pt}
\pgfsetdash{}{0pt}
\definecolor{dialinecolor}{rgb}{0.000000, 0.000000, 0.000000}
\pgfsetstrokecolor{dialinecolor}
\pgfpathellipse{\pgfpoint{4.500000\du}{9.000000\du}}{\pgfpoint{0.125000\du}{0\du}}{\pgfpoint{0\du}{0.125000\du}}
\pgfusepath{stroke}
\node[anchor=west] at (9.000000\du,8.000000\du){$x^W(A,w)=\{\emptyset\}$};
\node[anchor=west] at (9.000000\du,9.000000\du){$x^W(A,w')=\{1,2\}$};
\end{tikzpicture}
\end{subfigure}
\begin{subfigure}[b]{0.45\textwidth}
\centering
\ifx\du\undefined
  \newlength{\du}
\fi
\setlength{\du}{15\unitlength}
\begin{tikzpicture}
\pgftransformxscale{1.000000}
\pgftransformyscale{-1.000000}
\definecolor{dialinecolor}{rgb}{0.000000, 0.000000, 0.000000}
\pgfsetstrokecolor{dialinecolor}
\definecolor{dialinecolor}{rgb}{1.000000, 1.000000, 1.000000}
\pgfsetfillcolor{dialinecolor}
\pgfsetlinewidth{0.100000\du}
\pgfsetdash{}{0pt}
\pgfsetdash{}{0pt}
\pgfsetbuttcap
{
\definecolor{dialinecolor}{rgb}{0.000000, 0.000000, 0.000000}
\pgfsetfillcolor{dialinecolor}
\pgfsetarrowsend{stealth}
\definecolor{dialinecolor}{rgb}{0.000000, 0.000000, 0.000000}
\pgfsetstrokecolor{dialinecolor}
\draw (4.000000\du,17.000000\du)--(17.000000\du,17.000000\du);
}
\pgfsetlinewidth{0.100000\du}
\pgfsetdash{}{0pt}
\pgfsetdash{}{0pt}
\pgfsetbuttcap
{
\definecolor{dialinecolor}{rgb}{0.000000, 0.000000, 0.000000}
\pgfsetfillcolor{dialinecolor}
\pgfsetarrowsend{stealth}
\definecolor{dialinecolor}{rgb}{0.000000, 0.000000, 0.000000}
\pgfsetstrokecolor{dialinecolor}
\draw (4.000000\du,17.000000\du)--(4.000000\du,7.000000\du);
}
\pgfsetlinewidth{0.0500000\du}
\pgfsetdash{{\pgflinewidth}{0.200000\du}}{0cm}
\pgfsetdash{{\pgflinewidth}{0.200000\du}}{0cm}
\pgfsetbuttcap
{
\definecolor{dialinecolor}{rgb}{0.000000, 0.000000, 0.000000}
\pgfsetfillcolor{dialinecolor}
\definecolor{dialinecolor}{rgb}{0.000000, 0.000000, 0.000000}
\pgfsetstrokecolor{dialinecolor}
\draw (7.000000\du,17.000000\du)--(7.000000\du,7.000000\du);
}
\definecolor{dialinecolor}{rgb}{0.000000, 0.000000, 0.000000}
\pgfsetstrokecolor{dialinecolor}
\node[anchor=west] at (6.000000\du,18.000000\du){$\pi_1(w)$};
\definecolor{dialinecolor}{rgb}{0.000000, 0.000000, 0.000000}
\pgfsetstrokecolor{dialinecolor}
\node[anchor=west] at (5.000000\du,10.000000\du){};
\definecolor{dialinecolor}{rgb}{0.000000, 0.000000, 0.000000}
\pgfsetstrokecolor{dialinecolor}
\node[anchor=west] at (17.000000\du,17.500000\du){$v_1$};
\definecolor{dialinecolor}{rgb}{0.000000, 0.000000, 0.000000}
\pgfsetstrokecolor{dialinecolor}
\node[anchor=west] at (3.000000\du,6.500000\du){$v_2$};
\pgfsetlinewidth{0.0500000\du}
\pgfsetdash{{\pgflinewidth}{0.200000\du}}{0cm}
\pgfsetdash{{\pgflinewidth}{0.200000\du}}{0cm}
\pgfsetbuttcap
{
\definecolor{dialinecolor}{rgb}{0.000000, 0.000000, 0.000000}
\pgfsetfillcolor{dialinecolor}
\definecolor{dialinecolor}{rgb}{0.000000, 0.000000, 0.000000}
\pgfsetstrokecolor{dialinecolor}
\draw (4.000000\du,13.000000\du)--(7.000000\du,13.000000\du);
}
\definecolor{dialinecolor}{rgb}{0.000000, 0.000000, 0.000000}
\pgfsetstrokecolor{dialinecolor}
\node[anchor=west] at (1.500000\du,13.000000\du){$\pi_2(w)$};
\pgfsetlinewidth{0.0500000\du}
\pgfsetdash{{\pgflinewidth}{0.200000\du}}{0cm}
\pgfsetdash{{\pgflinewidth}{0.200000\du}}{0cm}
\pgfsetbuttcap
{
\definecolor{dialinecolor}{rgb}{0.000000, 0.000000, 0.000000}
\pgfsetfillcolor{dialinecolor}
\definecolor{dialinecolor}{rgb}{0.000000, 0.000000, 0.000000}
\pgfsetstrokecolor{dialinecolor}
\draw (7.000000\du,13.000000\du)--(13.500000\du,6.500000\du);
}
\pgfsetlinewidth{0.0500000\du}
\pgfsetdash{{\pgflinewidth}{0.200000\du}}{0cm}
\pgfsetdash{{\pgflinewidth}{0.200000\du}}{0cm}
\pgfsetbuttcap
{
\definecolor{dialinecolor}{rgb}{0.000000, 0.000000, 0.000000}
\pgfsetfillcolor{dialinecolor}
\definecolor{dialinecolor}{rgb}{0.000000, 0.000000, 0.000000}
\pgfsetstrokecolor{dialinecolor}
\draw (7.000000\du,13.000000\du)--(8.000000\du,13.000000\du);
}
\definecolor{dialinecolor}{rgb}{0.000000, 0.000000, 0.000000}
\pgfsetstrokecolor{dialinecolor}
\node[anchor=west] at (7.700000\du,12.500000\du){$45^o$};
\definecolor{dialinecolor}{rgb}{0.000000, 0.000000, 0.000000}
\pgfsetstrokecolor{dialinecolor}
\node[anchor=west] at (13.000000\du,18.000000\du){$H\cdot \epsilon$};
\pgfsetlinewidth{0.100000\du}
\pgfsetdash{}{0pt}
\pgfsetdash{}{0pt}
\pgfsetbuttcap
{
\definecolor{dialinecolor}{rgb}{0.000000, 0.000000, 0.000000}
\pgfsetfillcolor{dialinecolor}
\definecolor{dialinecolor}{rgb}{0.000000, 0.000000, 0.000000}
\pgfsetstrokecolor{dialinecolor}
\draw (13.500000\du,17.000000\du)--(13.500000\du,16.500000\du);
}
\definecolor{dialinecolor}{rgb}{0.000000, 0.000000, 0.000000}
\pgfsetstrokecolor{dialinecolor}
\node[anchor=west] at (1.500000\du,10.000000\du){$\pi_2(w')$};
\definecolor{dialinecolor}{rgb}{0.000000, 0.000000, 0.000000}
\pgfsetstrokecolor{dialinecolor}
\node[anchor=west] at (9.000000\du,18.000000\du){$\pi_1(w')$};
\pgfsetlinewidth{0.0500000\du}
\pgfsetdash{{\pgflinewidth}{0.200000\du}}{0cm}
\pgfsetdash{{\pgflinewidth}{0.200000\du}}{0cm}
\pgfsetbuttcap
{
\definecolor{dialinecolor}{rgb}{0.000000, 0.000000, 0.000000}
\pgfsetfillcolor{dialinecolor}
\definecolor{dialinecolor}{rgb}{0.000000, 0.000000, 0.000000}
\pgfsetstrokecolor{dialinecolor}
\draw (4.000000\du,10.000000\du)--(10.000000\du,10.000000\du);
}
\pgfsetlinewidth{0.0500000\du}
\pgfsetdash{{\pgflinewidth}{0.200000\du}}{0cm}
\pgfsetdash{{\pgflinewidth}{0.200000\du}}{0cm}
\pgfsetbuttcap
{
\definecolor{dialinecolor}{rgb}{0.000000, 0.000000, 0.000000}
\pgfsetfillcolor{dialinecolor}
\definecolor{dialinecolor}{rgb}{0.000000, 0.000000, 0.000000}
\pgfsetstrokecolor{dialinecolor}
\draw (10.000000\du,17.000000\du)--(10.000000\du,7.000000\du);
}
\definecolor{dialinecolor}{rgb}{0.000000, 0.000000, 0.000000}
\pgfsetfillcolor{dialinecolor}
\pgfpathellipse{\pgfpoint{8.000000\du}{11.000000\du}}{\pgfpoint{0.10000\du}{0\du}}{\pgfpoint{0\du}{0.10000\du}}
\pgfusepath{fill}
\pgfsetlinewidth{0.100000\du}
\pgfsetdash{}{0pt}
\pgfsetdash{}{0pt}
\definecolor{dialinecolor}{rgb}{0.000000, 0.000000, 0.000000}
\pgfsetstrokecolor{dialinecolor}
\pgfpathellipse{\pgfpoint{8.000000\du}{11.000000\du}}{\pgfpoint{0.10000\du}{0\du}}{\pgfpoint{0\du}{0.10000\du}}
\pgfusepath{stroke}
\definecolor{dialinecolor}{rgb}{0.000000, 0.000000, 0.000000}
\pgfsetstrokecolor{dialinecolor}
\node[anchor=west] at (8.000000\du,10.500000\du){A};
\definecolor{dialinecolor}{rgb}{0.000000, 0.000000, 0.000000}
\pgfsetstrokecolor{dialinecolor}
\node[anchor=west] at (11.000000\du,10.500000\du){$x^V(A,w)=\{1,2\}$};
\definecolor{dialinecolor}{rgb}{0.000000, 0.000000, 0.000000}
\pgfsetstrokecolor{dialinecolor}
\node[anchor=west] at (11.000000\du,11.500000\du){$x^V(A,w')=\{\emptyset\}$};
\end{tikzpicture}
\end{subfigure}
\caption{If the thresholds for player $v$ change as the value of $W$ ranges over $S_{N,\epsilon}$, e.g. from $w$ to $w'$, then when $V$ has type $A$
the allocation of $W$ switches from the $\emptyset$, under $w$, to winning both items, under $w'$. The latter is impossible, given our characterization of feasible allocation functions.}\label{fig:impos2}
\end{figure}
\item Our characterization in Corollary~\ref{cor:AllocationRegions} shows that only one of the three moves mentioned above could possibly result in an allocation change between $\emptyset$ and $\{1,2\}$, namely from $B_L^W$ to $B_R^W$. But $B_R^W$ has an allocation of $\{1,2\}$ only when $\phi_2 = \phi_1$. The point $w$ being in $B_L^W$ and $w'$ in $B_R^W$ means that $w_1 < \phi_1 = \phi_2 < w_1'$. But note that $w_1, w_1' < \epsilon$, where as $w_2 = w_2' \geq H >> \epsilon$. On the other hand since $w \in B_L^W$ it means that $\phi_2 > w_2 \geq H$. It is impossible to have the following conditions satisfied:
\begin{itemize}
\item $\phi_2 = \phi_1$ 
\item $w_2 = w_2' \geq H$ and $0 < w_1, w_1' < \epsilon$
\item $w_2 < \phi_2$ and $w_1 < \phi_1 = \phi_2 < w_1'$
\end{itemize}

\item This means that our contrary assumption can't hold. This proves the Lemma, namely, $\pi_2(w) = \pi_2(w')$ and $\pi_1(w) = \pi_1(w')$. 
\end{enumerate}
\end{proof}

\begin{proofof}{Main Theorem}
From Corollary~\ref{cor:AllocationRegions} it is intuitively clear that once $\pi_2(w)$ and $\pi_1(w)$ are determined, they almost entirely determine the allocation function in the type-space of bidder $V$. But now Lemma~\ref{lem:PisAreFrozen} (basically Equation~\ref{eq:PisAreFrozen}) shows that $\pi_2$ and $\pi_1$ remain unchanged when $w$ moves in the set $S_{N,\epsilon}$. This therefore means, when $w$ moves in the set $S_{N,\epsilon}$, the allocation to bidder $V$ for any given bid $v \in (0,H\epsilon)\times \R_+$ (except for the indifference points in the boundary region) remains unchanged. Therefore, the allocation of bidder $W$, which is simply the complement of the allocation to bidder $V$, also remains unchanged, as $w$ moves in $S_{N,\epsilon}$. This leads to an unbounded approximation factor. 
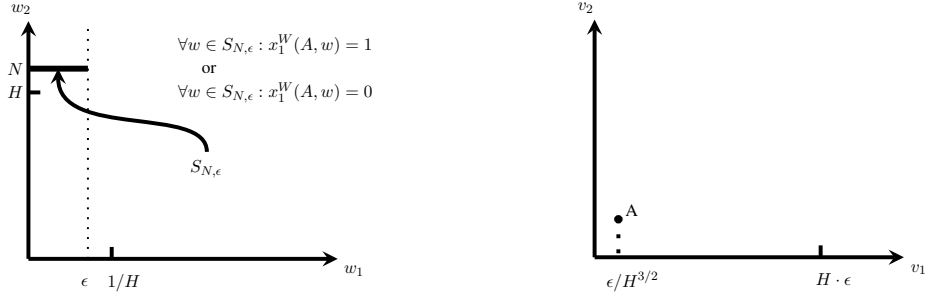
\begin{figure}[h]
\centering
\begin{subfigure}[b]{0.45\textwidth}
\centering
\ifx\du\undefined
  \newlength{\du}
\fi
\setlength{\du}{15\unitlength}
\begin{tikzpicture}
\pgftransformxscale{1.000000}
\pgftransformyscale{-1.000000}
\definecolor{dialinecolor}{rgb}{0.000000, 0.000000, 0.000000}
\pgfsetstrokecolor{dialinecolor}
\definecolor{dialinecolor}{rgb}{1.000000, 1.000000, 1.000000}
\pgfsetfillcolor{dialinecolor}
\pgfsetlinewidth{0.100000\du}
\pgfsetdash{}{0pt}
\pgfsetdash{}{0pt}
\pgfsetbuttcap
{
\definecolor{dialinecolor}{rgb}{0.000000, 0.000000, 0.000000}
\pgfsetfillcolor{dialinecolor}
\pgfsetarrowsend{stealth}
\definecolor{dialinecolor}{rgb}{0.000000, 0.000000, 0.000000}
\pgfsetstrokecolor{dialinecolor}
\draw (4.000000\du,17.000000\du)--(17.000000\du,17.000000\du);
}
\pgfsetlinewidth{0.100000\du}
\pgfsetdash{}{0pt}
\pgfsetdash{}{0pt}
\pgfsetbuttcap
{
\definecolor{dialinecolor}{rgb}{0.000000, 0.000000, 0.000000}
\pgfsetfillcolor{dialinecolor}
\pgfsetarrowsend{stealth}
\definecolor{dialinecolor}{rgb}{0.000000, 0.000000, 0.000000}
\pgfsetstrokecolor{dialinecolor}
\draw (4.000000\du,17.000000\du)--(4.000000\du,7.000000\du);
}
\pgfsetlinewidth{0.0500000\du}
\pgfsetdash{{\pgflinewidth}{0.200000\du}}{0cm}
\pgfsetdash{{\pgflinewidth}{0.200000\du}}{0cm}
\pgfsetbuttcap
{
\definecolor{dialinecolor}{rgb}{0.000000, 0.000000, 0.000000}
\pgfsetfillcolor{dialinecolor}
\definecolor{dialinecolor}{rgb}{0.000000, 0.000000, 0.000000}
\pgfsetstrokecolor{dialinecolor}
\draw (6.500000\du,17.000000\du)--(6.500000\du,7.000000\du);
}
\definecolor{dialinecolor}{rgb}{0.000000, 0.000000, 0.000000}
\pgfsetstrokecolor{dialinecolor}
\node[anchor=west] at (5.000000\du,10.000000\du){};
\definecolor{dialinecolor}{rgb}{0.000000, 0.000000, 0.000000}
\pgfsetstrokecolor{dialinecolor}
\node[anchor=west] at (17.000000\du,17.500000\du){$w_1$};
\definecolor{dialinecolor}{rgb}{0.000000, 0.000000, 0.000000}
\pgfsetstrokecolor{dialinecolor}
\node[anchor=west] at (3.000000\du,6.500000\du){$w_2$};
\pgfsetlinewidth{0.150000\du}
\pgfsetdash{}{0pt}
\pgfsetdash{}{0pt}
\pgfsetbuttcap
{
\definecolor{dialinecolor}{rgb}{0.000000, 0.000000, 0.000000}
\pgfsetfillcolor{dialinecolor}
\definecolor{dialinecolor}{rgb}{0.000000, 0.000000, 0.000000}
\pgfsetstrokecolor{dialinecolor}
\draw (4.000000\du,9.000000\du)--(6.500000\du,9.000000\du);
}
\definecolor{dialinecolor}{rgb}{0.000000, 0.000000, 0.000000}
\pgfsetstrokecolor{dialinecolor}
\node[anchor=west] at (5.900000\du,18.000000\du){$\epsilon$};
\definecolor{dialinecolor}{rgb}{0.000000, 0.000000, 0.000000}
\pgfsetstrokecolor{dialinecolor}
\node[anchor=west] at (2.800000\du,9.000000\du){$N$};
\definecolor{dialinecolor}{rgb}{0.000000, 0.000000, 0.000000}
\pgfsetstrokecolor{dialinecolor}
\node[anchor=west] at (7.000000\du,18.000000\du){$1/H$};
\definecolor{dialinecolor}{rgb}{0.000000, 0.000000, 0.000000}
\pgfsetstrokecolor{dialinecolor}
\node[anchor=west] at (2.8000000\du,10.000000\du){$H$};
\definecolor{dialinecolor}{rgb}{0.000000, 0.000000, 0.000000}
\pgfsetfillcolor{dialinecolor}
\pgfsetlinewidth{0.100000\du}
\pgfsetdash{}{0pt}
\pgfsetdash{}{0pt}
\definecolor{dialinecolor}{rgb}{0.000000, 0.000000, 0.000000}
\pgfsetstrokecolor{dialinecolor}
\pgfusepath{stroke}
\definecolor{dialinecolor}{rgb}{0.000000, 0.000000, 0.000000}
\pgfsetstrokecolor{dialinecolor}
\definecolor{dialinecolor}{rgb}{0.000000, 0.000000, 0.000000}
\pgfsetstrokecolor{dialinecolor}
\pgfsetlinewidth{0.100000\du}
\pgfsetdash{}{0pt}
\pgfsetdash{}{0pt}
\pgfsetbuttcap
{
\definecolor{dialinecolor}{rgb}{0.000000, 0.000000, 0.000000}
\pgfsetfillcolor{dialinecolor}
\definecolor{dialinecolor}{rgb}{0.000000, 0.000000, 0.000000}
\pgfsetstrokecolor{dialinecolor}
\draw (7.500000\du,17.000000\du)--(7.500000\du,16.500000\du);
}
\pgfsetlinewidth{0.100000\du}
\pgfsetdash{}{0pt}
\pgfsetdash{}{0pt}
\pgfsetmiterjoin
\pgfsetbuttcap
{
\definecolor{dialinecolor}{rgb}{0.000000, 0.000000, 0.000000}
\pgfsetfillcolor{dialinecolor}
\pgfsetarrowsend{stealth}
\definecolor{dialinecolor}{rgb}{0.000000, 0.000000, 0.000000}
\pgfsetstrokecolor{dialinecolor}
\pgfpathmoveto{\pgfpoint{11.500000\du}{12.500000\du}}
\pgfpathcurveto{\pgfpoint{11.500000\du}{10.500000\du}}{\pgfpoint{5.250000\du}{12.000000\du}}{\pgfpoint{5.250000\du}{9.000000\du}}
\pgfusepath{stroke}
}
\definecolor{dialinecolor}{rgb}{0.000000, 0.000000, 0.000000}
\pgfsetstrokecolor{dialinecolor}
\node[anchor=west] at (10.500000\du,13.200000\du){$S_{N,\epsilon}$};
\pgfsetlinewidth{0.100000\du}
\pgfsetdash{}{0pt}
\pgfsetdash{}{0pt}
\pgfsetbuttcap
{
\definecolor{dialinecolor}{rgb}{0.000000, 0.000000, 0.000000}
\pgfsetfillcolor{dialinecolor}
\definecolor{dialinecolor}{rgb}{0.000000, 0.000000, 0.000000}
\pgfsetstrokecolor{dialinecolor}
\draw (4.500000\du,10.000000\du)--(4.000000\du,10.000000\du);
}
\definecolor{dialinecolor}{rgb}{0.000000, 0.000000, 0.000000}
\pgfsetfillcolor{dialinecolor}
\pgfsetlinewidth{0.100000\du}
\pgfsetdash{}{0pt}
\pgfsetdash{}{0pt}
\definecolor{dialinecolor}{rgb}{0.000000, 0.000000, 0.000000}
\pgfsetstrokecolor{dialinecolor}
\pgfusepath{stroke}

\definecolor{dialinecolor}{rgb}{0.000000, 0.000000, 0.000000}
\pgfsetstrokecolor{dialinecolor}
\node[anchor=west] at (10.000000\du,8.000000\du){$\forall w\in S_{N,\epsilon}: x_1^W(A,w)=1$};
\definecolor{dialinecolor}{rgb}{0.000000, 0.000000, 0.000000}
\pgfsetstrokecolor{dialinecolor}
\node[anchor=west] at (11.000000\du,9.000000\du){or};
\definecolor{dialinecolor}{rgb}{0.000000, 0.000000, 0.000000}
\pgfsetstrokecolor{dialinecolor}
\node[anchor=west] at (10.000000\du,10.000000\du){$\forall w\in S_{N,\epsilon}: x_1^W(A,w)=0$};
\end{tikzpicture}
\end{subfigure}
\begin{subfigure}[b]{0.45\textwidth}
\centering
\ifx\du\undefined
  \newlength{\du}
\fi
\setlength{\du}{15\unitlength}
\begin{tikzpicture}
\pgftransformxscale{1.000000}
\pgftransformyscale{-1.000000}
\definecolor{dialinecolor}{rgb}{0.000000, 0.000000, 0.000000}
\pgfsetstrokecolor{dialinecolor}
\definecolor{dialinecolor}{rgb}{1.000000, 1.000000, 1.000000}
\pgfsetfillcolor{dialinecolor}
\pgfsetlinewidth{0.100000\du}
\pgfsetdash{}{0pt}
\pgfsetdash{}{0pt}
\pgfsetbuttcap
{
\definecolor{dialinecolor}{rgb}{0.000000, 0.000000, 0.000000}
\pgfsetfillcolor{dialinecolor}
\pgfsetarrowsend{stealth}
\definecolor{dialinecolor}{rgb}{0.000000, 0.000000, 0.000000}
\pgfsetstrokecolor{dialinecolor}
\draw (4.000000\du,17.000000\du)--(17.000000\du,17.000000\du);
}
\pgfsetlinewidth{0.100000\du}
\pgfsetdash{}{0pt}
\pgfsetdash{}{0pt}
\pgfsetbuttcap
{
\definecolor{dialinecolor}{rgb}{0.000000, 0.000000, 0.000000}
\pgfsetfillcolor{dialinecolor}
\pgfsetarrowsend{stealth}
\definecolor{dialinecolor}{rgb}{0.000000, 0.000000, 0.000000}
\pgfsetstrokecolor{dialinecolor}
\draw (4.000000\du,17.000000\du)--(4.000000\du,7.000000\du);
}
\definecolor{dialinecolor}{rgb}{0.000000, 0.000000, 0.000000}
\pgfsetstrokecolor{dialinecolor}
\node[anchor=west] at (5.000000\du,10.000000\du){};
\definecolor{dialinecolor}{rgb}{0.000000, 0.000000, 0.000000}
\pgfsetstrokecolor{dialinecolor}
\node[anchor=west] at (17.000000\du,17.500000\du){$v_1$};
\definecolor{dialinecolor}{rgb}{0.000000, 0.000000, 0.000000}
\pgfsetstrokecolor{dialinecolor}
\node[anchor=west] at (3.000000\du,6.500000\du){$v_2$};
\definecolor{dialinecolor}{rgb}{0.000000, 0.000000, 0.000000}
\pgfsetstrokecolor{dialinecolor}
\node[anchor=west] at (13.000000\du,17.900000\du){$H\cdot \epsilon$};
\pgfsetlinewidth{0.100000\du}
\pgfsetdash{}{0pt}
\pgfsetdash{}{0pt}
\pgfsetbuttcap
{
\definecolor{dialinecolor}{rgb}{0.000000, 0.000000, 0.000000}
\pgfsetfillcolor{dialinecolor}
\definecolor{dialinecolor}{rgb}{0.000000, 0.000000, 0.000000}
\pgfsetstrokecolor{dialinecolor}
\draw (13.500000\du,17.000000\du)--(13.500000\du,16.500000\du);
}
\definecolor{dialinecolor}{rgb}{0.000000, 0.000000, 0.000000}
\pgfsetfillcolor{dialinecolor}
\pgfpathellipse{\pgfpoint{5.000000\du}{15.400000\du}}{\pgfpoint{0.095134\du}{0\du}}{\pgfpoint{0\du}{0.089598\du}}
\pgfusepath{fill}
\pgfsetlinewidth{0.100000\du}
\pgfsetdash{}{0pt}
\pgfsetdash{}{0pt}
\definecolor{dialinecolor}{rgb}{0.000000, 0.000000, 0.000000}
\pgfsetstrokecolor{dialinecolor}
\pgfpathellipse{\pgfpoint{5.000000\du}{15.400000\du}}{\pgfpoint{0.095134\du}{0\du}}{\pgfpoint{0\du}{0.089598\du}}
\pgfusepath{stroke}
\pgfsetlinewidth{0.100000\du}
\pgfsetdash{{\pgflinewidth}{0.200000\du}}{0cm}
\pgfsetdash{{\pgflinewidth}{0.280000\du}}{0cm}
\pgfsetbuttcap
{
\definecolor{dialinecolor}{rgb}{0.000000, 0.000000, 0.000000}
\pgfsetfillcolor{dialinecolor}
\definecolor{dialinecolor}{rgb}{0.000000, 0.000000, 0.000000}
\pgfsetstrokecolor{dialinecolor}
\draw (5.000000\du,15.400000\du)--(5.000000\du,17.000000\du);
}
\definecolor{dialinecolor}{rgb}{0.000000, 0.000000, 0.000000}
\pgfsetstrokecolor{dialinecolor}
\node[anchor=west] at (4.200000\du,18.000000\du){$\epsilon/H^{3/2}$};
\definecolor{dialinecolor}{rgb}{0.000000, 0.000000, 0.000000}
\pgfsetstrokecolor{dialinecolor}
\node[anchor=west] at (5.000000\du,15.000000\du){A};
\end{tikzpicture}
\end{subfigure}
\caption{Our characterizations imply that the whether player $W$ wins item $1$ or not, remains fixed for all $w\in S_{N,\epsilon}$ for any $v$ with $v_1\leq H\cdot \epsilon$ (in particular for point $A$). This leads to unbounded approximation factor when only item $1$ arrives.}\label{fig:char_final}
\end{figure}
We will now make this formal by first showing the following. 
\begin{enumerate}
\item For any set $S$ of $w$'s where $\pi_2(w)$ and $\pi_1(w)$ remain unchanged, let $\mathcal{Q}(S) = \{v \in \R_+^2: v_1, v_2 >0, v_1 < H \epsilon, \ v_1 \neq \pi_1, v_2 \neq \pi_2, v_2 - \pi_2 \neq v_1 - \pi_1\}$ be the set of $v$ with $v_1\leq H\epsilon$, excluding the indifference points in bidder $V$'s type space (the $\pi$'s are defined w.r.t. some arbitrary $w \in S$). The allocation for indifference points could swing arbitrarily as $w$ moves and therefore we will not focus on them. We show that:

\begin{align}
\label{eq:AllocationVIsFrozen}
&\forall N \geq 2H, \forall \epsilon < \frac{1}{H}, \forall w \in S_{N,\epsilon} = \{ 0 < w_1 < \epsilon,\ w_2 = N\}, \forall v \in \mathcal{Q}(S_{N,\epsilon}) \text{ we have: }\notag\\
&\qquad\qquad\qquad\qquad\allocv_1(v,w) = f_{1,N}(v),\ and,\ \allocv_2(v,w) = f_{2,N}(v).
\end{align}

To prove Equation~\ref{eq:AllocationVIsFrozen}, note that Corollary~\ref{cor:AllocationRegions} says that the only region in bidder $V$'s type-space where $\pi_2$ and $\pi_1$ don't accurately determine the allocation functions for bidder $V$ is the region $B_R^V$ where the allocation could be $\{1\}$ or $\{1,2\}$ in the case where $\pi_2 = \pi_1$. Except for this region, since $\pi_2$ and $\pi_1$ entirely determine the allocation of bidder $V$, Equation~\ref{eq:PisAreFrozen} directly implies Equation~\ref{eq:AllocationVIsFrozen}. We now show that this corner case of $\pi_2 = \pi_1$ can't happen by proving that 

$$\forall N \geq 2H, \forall \epsilon < \frac{1}{H}, \forall w \in S_{N,\epsilon} = \{ 0 < w_1 < \epsilon,\ w_2 = N\}, \text{ we have } \pi_2(w) > \pi_1(w).$$ 

To see this, note that it is necessary that $\pi_2 \geq 1$ because otherwise, the approximation ratio for giving away item $2$ with $w_2 \geq 2H$ to bidder $V$ with $v_2 < 1$ is at least $\frac{2H}{1+\epsilon} > H$ (the $\epsilon$ in the denominator is for capturing item $1$'s welfare of at most $\epsilon$). On the other hand, as discussed in point 2 of the proof of Lemma~\ref{lem:Pi2MinusPi1IndependentOfW1}, we have $\pi_1 < H\epsilon < 1$. Thus $\pi_1 < 1 \leq \pi_2$, eliminating the corner case we hoped to remove.

\item Notice that the allocation to bidder $W$ is simply the complement of allocation to bidder $V$ because the mechanism allocates both items, according to \Cref{lem:Pi1NotZero,lem:item2_always_allocated}. Thus,

\begin{align}
\label{eq:AllocationWIsFrozen}
&\forall N \geq 2H, \forall \epsilon < \frac{1}{H}, \forall w \in S_{N,\epsilon} = \{ 0 < w_1 < \epsilon,\ w_2 = N\}, \forall v \in \mathcal{Q}(S_{N,\epsilon}) \text{ we have: }\notag\\
&\qquad\qquad\qquad\qquad\allocw_1(v,w) = 1 - f_{1,N}(v),\ and,\ \allocw_2(v,w) = 1 - f_{2,N}(v).
\end{align}

Equation~\ref{eq:AllocationWIsFrozen} leads to an unbounded approximation factor. To see this, consider the $S_{2H,\epsilon} = \{w \in \R_+^2: 0 < w_1 < \epsilon, w_2 = 2H\}$. Fix a $v$ s.t. $v_1 = \frac{\epsilon}{H^{3/2}}$ (if such a $v_1$ happens to be a boundary point w.r.t. $S_{2H,\epsilon}$, pick a $v_1$ arbitrarily close to that point). Equation ~\ref{eq:AllocationWIsFrozen} implies that for all $w \in S_{2H,\epsilon}$, we have $\allocw_1(v,w) = 1$ or $\allocw_1(v,w) = 0$, i.e., always $1$ or always $0$. 
\begin{enumerate}
\item Suppose $\allocw_1(v,w) = 1$ for all $w \in S_{2H,\epsilon}$. Then consider a point $(w_1 = \frac{\epsilon}{H^3}, w_2 = 2H)$, and focus on the case where item $2$ never arrives. Giving item $1$ to bidder $W$ results in approximation factor of $\frac{\epsilon/H^{3/2}}{\epsilon/H^3} = H^{3/2} > H$. Thus a $H$-approximation factor is not possible in this case.

\item Suppose $\allocw_1(v,w) = 0$ for all $w \in S_{2H,\epsilon}$. Then consider a point $(w_1 = \epsilon/2, w_2 = 2H)$ and focus on the case where item $2$ never arrives. Giving item $1$ to bidder $V$ results in approximation factor of $\frac{\epsilon/2}{\epsilon/H^{3/2}} = \frac{H^{3/2}}{2} > H$. Thus a $H$-approximation factor is not possible in this case.
\end{enumerate}
This means that a deterministic DSIC+IR mechanism cannot obtain a $H$-approximation for any $H$. This proves the theorem. 
\end{enumerate}
\end{proofof}

\section{Conclusion and Open Questions}\label{sec:conclusion}
In the largely unexplored space of truthful mechanisms for multi-parameter 
bidders with online supply, we ask the first question that anyone would ask, 
and show a strong impossibility result.  On the technical side, our results 
shed ample light on the nature of restrictions that truthfulness and online 
supply place on the kinds of allocation profiles possible. These restrictions 
are quite non-trivial and do not follow from simple weak-monotonicity/cyclic 
monotonicity conditions. On the conceptual side, our results point toward what 
might be reasonably expected in terms of truthfulness in an online supply 
setting. For instance, it says that there must be some source of randomness for 
a mechanism to be able to guarantee any reasonable social welfare. In this 
regard, there are several interesting directions that our work opens up.
\paragraph{Does randomness help?}
 There are three prominent sources of randomness: a) Randomized mechanisms b) 
 Stochastic arrival and c) Bayesian setting (values drawn from a distribution). 
For each of these sources of randomness, without using the other sources of randomness, what is the optimal approximation factor one can get? Let $n$ be the number of bidders and $m$ be the number of items that could possibly arrive (lesser number of items could arrive too). 
\paragraph{Randomized mechanisms}
It is easy to get a $\min\{n,m\}$ approximation with a randomized mechanism. 
Choosing one of the $n$ bidders uniformly at random and allocating all items to 
her gives an $n$ approximation. Choosing one of the $m$ items uniformly at 
random, and selling that item using a second price auction, throwing away all 
other items gives a factor $m$ approximation (note that the adversary who 
decides the supply doesn't know the mechanism's coin toss). Is there a 
mechanism that gets \emph{anything better} than this trivial $\min\{m,n\}$ 
approximation? What is the best approximation factor one can get?  Logarithmic? 
Constant? Note that the constant factor impossibility result 
of~\cite{BabaioffBR15} doesn't apply here because we allow payments to be 
computed after knowing the entire supply. We conjecture that nothing better 
than the $\min\{n,m\}$ approximation with a randomized mechanism. 

\paragraph{Stochastic arrival}
What happens if we allow arrival to be stochastic, according to a \emph{known} distribution, but ask for mechanisms to be deterministic? For the $2$ bidder, $2$ item setting we studied, here is a stochastic arrival model. Item $1$ arrives for sure, and item $2$ arrives with a known probability $p$. If all we ask for is a truthful-in-expectation mechanism, then one can get the optimal welfare by just implementing the VCG mechanism, namely, compute $\max\{v_1 + pw_2, w_1 + pv_2\}$. If the former is larger, bidder $V$ gets $1$ and $W$ gets $2$, and if the latter is larger, bidder $V$ gets $2$ and $W$ gets $1$. Truthful payments are easy to compute. 

But what if we ask for universal truthfulness over the stochasticity in arrival? For the simple setting of $2$ bidders and $2$ items, there is a $2$ approximation mechanism that allocates the first item to bidder $V$ if $v_1 \geq \max\{w_1, pw_2\}$ or bidder $W$ if $w_1 \geq \max\{v_1, pv_2\}$. When second item arrives, offer it to bidder $V$ at a price of $\max\{\frac{w_1}{p}, w_2\} + (v_1 - \max\{w_1, pw_2\})^+$ and similarly offer it to bidder $W$ at a price of $\max\{\frac{v_1}{p}, v_2\} + (w_1 - \max\{v_1, pv_2\})^+$. Whoever wins the second item pays the posted price and nothing more, even if they won the first item. If bidder $V$ wins the first item alone (either because the second item doesn't arrive, or because he lost it), he pays $\max\{w_1, pw_2\}$ and similarly when bidder $W$ wins the first item alone, he pays $\max\{v_1, pv_2\}$. It is easy to check that this mechanism is truthful, and gets a $2$ approximation. Basically, the mechanism is guaranteed to get a welfare of $\max\{v_1, w_1, pv_2, pw_2\}$, which gives a $2$ approximation. Is there a natural generalization of this mechanism to $n$ bidders and $m$ items that gives a $\min\{m,n\}$ approximation? Are there are better approximations?

\paragraph{Bayesian setting}

What if bidders' values are drawn from a (known) joint distribution? If the bidder's distributions are independent, how good an approximation factor can we get with deterministic mechanisms? With randomized mechanisms can we get much better than the trivial $\min\{m, n\}$ approximation factor? With the further restriction that the distributions across items are independent too, what are the best factors possible?  We currently don't have answers to these questions even for the 2 bidder 2 item case.

\bibliographystyle{plainnat}
\bibliography{poa_survey}

\appendix
\section{Proof of missing lemmas and theorems}
\label{app:missing_proofs}
\begin{lemma}[Restatement of Lemma~\ref{lem:VerticalLine}]
For every deterministic DSIC+IR mechanism, there must exist a threshold 
function $\pi_1(w)$  such that:
\begin{align*}
&\text{$\forall w \in \R_{+}^2$, $\forall v \in \R_{+}^2:$}\ \  \allocv_1(v,w) 
= 
\begin{cases}
1 & \text{ if } v_1 > \pi_1(w)\\
0 & \text{ if } v_1 < \pi_1(w)
\end{cases}
\comment{
&\text{$\forall v \in \R_{+}^2$, $\forall w \in \R_{+}^2:$}\ \ 
\allocw_1(v,w) = 
\begin{cases}
1 & \text{ if } w_1 > \phi_1(v)\\
0 & \text{ if } w_1 < \phi_1(v)
\end{cases}
}
\end{align*}
\end{lemma}
\begin{proof}
Consider the case where item $1$ alone arrives. In this case, DSIC implies that 
$\allocv_1(v,w)$ is monotone in $v_1$. I.e., for a fixed $w$ and $v_2$, there 
exists a threshold $\pi_1$ that is independent of $v_1$ such that 
$\allocv_1(v,w) = 1$ when $v_1 > \pi_1$ and $\allocv_1(v,w) = 0$ when $v_1 < 
\pi_1$.  Furthermore, when $v_1 > \pi_1$ it must be that $\pricev(v,w) = \pi_1$ 
for DSIC and IR to hold (i.e., bidder $V$ pays the ``critical'' or minimal 
payment required to get the allocation he gets). This implies that $\pi_1$ is 
independent of $v_2$, because otherwise bidder $A$ will have an incentive to 
report the $v_2$ for which $\pi_1$ is the smallest. 
Thus $\pi_1$ can be a function only of $w$. This proves the lemma. 
\end{proof}

\begin{lemma}[Restatament of Lemma~\ref{lem:HorizontalLine}]
For every deterministic DSIC+IR mechanism, there must exist a  threshold 
function $\pi_2(w)$ such that:
\begin{align*}
\text{$\forall w \in \R_+^2$, } \forall v \in \R_{+}^2 \text{ s.t. } v_1 < 
\pi_1(w):   &~~\allocv_2(v,w) = 
\begin{cases}
1 & \text{ if  } v_2 > \pi_2(w) \\
0 & \text{ if } v_2 < \pi_2 (w)
\end{cases}\\
&~~
\pricev(v,w) = 
\begin{cases}
\pi_2(w) & \text{ if  } v_2 > \pi_2(w) \\
0  & \text{ if } v_2 < \pi_2(w)
\end{cases}
\end{align*}
\end{lemma}
\begin{proof}
	Fix a $w$. Consider the set $S = \{v \in \R_{+}^2: v_1 < \pi_1, 
	\allocv_2(v,w) = 1\}$.
	\begin{enumerate}
		\item Lemma~\ref{lem:VerticalLine} says that $\allocv_1(v,w) = 0$ for 
		all $v \in S$. This means that for all $v \in S$ ,we have 
		$\allocv_1(v,w) = 0$ and $\allocv_2(v,w) = 1$, i.e., for all $v \in S$, 
		the allocation for bidder $V$ remains fixed. DSIC therefore implies 
		that $\pricev(v,w)$ remains fixed for all $v \in S$. 
		\item $\pricev(v,w)$ being fixed for all $v \in S$, together with DSIC 
		means that there exists a threshold function $\pi_2$ independent of 
		$v_2$ (but possibly dependent on $v_1$, and of course $w$) such that 
		the set $S$ is sandwiched as $S_{-} \subseteq S \subseteq S_+$, where 
		$S_{-} = \{v: v_1 < \pi_1, v_2 > \pi_2\}$, and $S_+ = \{v: v_1 < \pi_1, 
		v_2 \geq \pi_2\}$.  I.e., for each fixed $v_1$, the allocation of item 
		$2$ for bidder $V$ along the $v_2$ axis jumps from $0$ to $1$ when $v_2 
		> \pi_2(v_1, w)$ and stays $1$ there after. 
		\item But does $\pi_2(v_1,w)$ really depend on $v_1$? Point 2 above, 
		along with DSIC and IR implies that for all $v \in S_{-}$, we have 
		$\pricev(v,w) = \pi_2(v_1,w)$ (i.e., DSIC demands that bidder $V$ pay 
		the minimal payment to get the allocation he gets). This means that for 
		all $v: v_1 < \pi_1$, we need $\pi_2$ to be independent of $v_1$ 
		because otherwise a $v \in S_{-}$ has an incentive to report a bid in 
		$\argmin_{v \in S_{-}} \pi_2(v_1, w)$. Thus $\pi_2$ is purely a 
		function of $w$, proving the lemma.
	\end{enumerate}
\end{proof}

\begin{lemma}[Restatement of Lemma~\ref{lem:1sPayPi1}]
For every deterministic DSIC+IR mechanism,  
$$\forall w \in \R_+^2, \forall v \in \R_{+}^2  \text{ s.t. } \{\allocv_1(v,w) 
= 1, \ \allocv_2(v,w) = 0\},\text{ we have } \pricev(v,w) = \pi_1$$ 
\end{lemma}
\begin{proof}
Fix a $w$. Let $S = \{v: \allocv_1(v,w) = 1, \allocv_2(v,w) = 0\}$. DSIC 
implies that $\pricev(v,w)$ remains the same for all $v \in S$. Consider a $v$ 
s.t. $v_1 \geq \pi_1$ and $v_2 = 0$. We claim that for such a $v$, 
$\pricev(v,w) = \pi_1$. Suppose not and let $\pricev(v,w) = p$. 
	\begin{enumerate}
		\item If $ p < \pi_1$, a bidder $V$ with value $(\frac{p+\pi_1}{2}, 0)$ 
		will deviate to report a $v$ s.t. $v_1 > \pi_1$ and $v_2 = 0$ to get a 
		positive utility instead of the $0$ utility he currently gets. 
		\item If on the other-hand $p > \pi_1$, then a bidder $V$ with value 
		$(\frac{p+\pi_1}{2}, 0)$ will deviate to report a $v$ s.t. $v_1 < 
		\pi_1$ and $v_2 = 0$, to get a $0$ utility instead of the negative 
		utility he currently gets. This means that $p = \pi_1$. 
	\end{enumerate}
\end{proof}
\end{document}